\documentclass{vldb}
\usepackage{amsfonts}
\usepackage{balance}
\usepackage{times,amsmath}
\usepackage{psfrag,graphicx}
\usepackage{amssymb}
\usepackage{algorithm}
\usepackage[noend]{algorithmic}
\usepackage{url}
\usepackage{subfigure}
\usepackage[hidelinks, bookmarks=false]{hyperref}

\newtheorem{definition}{Definition}
\newtheorem{example}{Example}

\newtheorem{theorem}{Theorem}
\begin{document}

\title{Willingness Optimization for Social Group Activity}
\vspace{-10pt}
\numberofauthors{1}
\author{
\alignauthor
Hong-Han Shuai~~~~~~~~~~~De-Nian Yang~~~~~~~~~~~~~~~Philip S. Yu~~~~~~~~~~~~~~~Ming-Syan Chen\\
       \affaddr{National Taiwan Univ.~~~~~~~~~~~~Academia Sinica~~~~~~~~~Univ. of Illinois at Chicago~~~~~~~National Taiwan Univ.}\\
       \email{~~~~\{hhshuai,dnyang\}@iis.sinica.edu.tw~~~~~~~~~~~~~~~~~~~~psyu@cs.uic.edu~~~~~~mschen@cc.ee.ntu.edu.tw~~~}
}

\maketitle

\begin{abstract}
Studies show that a person is willing to join a social group activity if the
activity is interesting, and if some close friends also join the activity
as companions. The literature has demonstrated that the interests of a
person and the social tightness among friends can be effectively derived and
mined from social networking websites. However, even with the above two kinds of
information widely available, social group activities still need to
be coordinated manually, and the process is tedious and time-consuming for
users, especially for a large social group activity, due to complications of social connectivity and the diversity of possible interests among friends. To
address the above important need, this paper proposes to automatically
select and recommend potential attendees of a social group activity, which
could be very useful for social networking websites as a value-added
service. We first formulate a new problem, named Willingness mAximization
for Social grOup (WASO). This paper points out that the solution obtained by
a greedy algorithm is likely to be trapped in a local optimal solution.
Thus, we design a new randomized algorithm to effectively and efficiently
solve the problem. Given the available computational budgets, the proposed
algorithm is able to optimally allocate the resources and find a solution
with an approximation ratio. We implement the proposed algorithm in
Facebook, and the user study demonstrates that social groups obtained by the
proposed algorithm significantly outperform the
solutions manually configured by users.
\end{abstract}
\vspace{-9pt}

\section{Introduction}

\label{Introduction}

Studies show that two important criteria are usually involved in the
decision of a person joining a group activity \cite{DeutschGDM,KaplanGDM} at her available time. First, the person is interested in the
intrinsic properties of the activity, which may be in line with her favorite
hobby or exercise. Second, other people who are important to the person,
such as her close friends, will join the activity as companions\footnote{%
There are other criteria that are also important, e.g., activity time, and activity location. However, to consider the above factors, a promising way is to preprocess and filter out the people who are not available, live too far, etc.}. For
example, if a person who appreciates abstract art has complimentary tickets for
a modern art exhibition at MoMA, she would probably want to invite her
friends and friends of friends with this shared interest. Nowadays, many
people are accustomed to sharing information with their friends on social
networking websites, like Facebook, Meetup, and
LikeALittle, and a recent line of studies \cite{PowerLawClauset09,MisloveMI} has introduced effective algorithms to
quantify the interests of a person according to the interest attributes in
her personal profile and the contextual information in her interaction with
friends. Moreover, social connectivity models have been widely studied \cite%
{ChaojiNnodeLink12} for evaluating the tightness
between two friends in the above websites. Nonetheless, even with the above
knowledge available, to date there has been neither published work nor a real system explores how to
leverage the above two crucial factors for \textit{automatic planning and
recommending} of a \textit{group activity}, which is potentially very
useful for social networking websites as a value-added service\footnote{%
The privacy of a person in automatic activity planning can follow the
current privacy setting policy in social networking websites when the person
subscribes the service. The details of privacy setting are beyond the scope
of this paper.}. At present, many social networking websites only act as a
platform for information sharing and exchange in activity planning. The
attendees of a group activity still need to be selected manually, and such
manual coordination is usually tedious and time-consuming, especially for a
large social activity, given the complicated link structure in social
networks and the diverse interests of friends.

To solve this problem, this paper makes an initial attempt to
incorporate the interests of people and their social tightness as two key
factors to find a group of attendees for automatic planning and
recommendation. It is desirable to choose more attendees who like and enjoy
the activity and to invite more friends with the shared interest in the
activity as companions. In fact, Psychology \cite{DeutschGDM,KaplanGDM} and recent study in social
networks \cite{YangSI12,LeesI12} have modeled the \textit{willingness}
to attend an activity or a social event as the sum of the interest of each
attendee on the activity and the social tightness between friends that are
possible to join it. It is envisaged that the selected attendees are more
inclined to join the activity if the willingness of the group increases.

With this objective in mind, we formulate a new fundamental optimization
problem, named \textit{Willingness mAximization for Social grOup (WASO)}.
The problem is given a social graph $G$, where each node represents a
candidate person and is associated with an interest score of the person for
the activity, and each edge has a social tightness score to indicate the
mutual familiarity between the two persons. Let $k$ denote the number of
expected attendees. Given the user-specified $k$, the goal of automatic
activity planning is to maximize the willingness of the selected group $F$,
while the induced graph on $F$ is a connected subgraph for each attendee to
become acquainted with another attendee according to a social path\footnote{%
For some group activities, it is not necessary to ensure that the solution
group is a connected subgraph. Later in Section \ref{Prilim}, we will
show that WASO without a connectivity constraint can be easily solved by
the proposed algorithm with simple modification.}. For the activities
without an \textit{a priori} fixed size, it is reasonable for a user to
specify a proper range for the group size, and our algorithm can find the
solution for each $k$ within the range and return the solutions for the user
to decide the most suitable group size and the corresponding attendees.\footnote{%
The parameter settings in WASO to fit varied scenarios in everyday life will
be introduced in more details in Section \ref{scenario}.}

Naturally, to incrementally construct the group, a greedy algorithm sequentially chooses an attendee that leads to the largest increment in the
willingness at each iteration. For example, Figure \ref{counterexample}
presents an illustrative example with $k=3$. Node $v_{1}$ is first selected
since its interest score is the maximum one among all nodes. Afterward, node 
$v_{2}$ is then extracted. Finally, $v_{3}$, instead of $v_{4}$, is chosen
because it generates the largest increment on willingness, i.e., $10$, and
leads to a group with a willingness of $27$. Note the greedy algorithm,
though simple, tends to be trapped in a local optimal solution, since it
facilitates the selection of nodes only suitable at the corresponding
iterations. In this simple example, the above algorithm is not able to find
the optimal solution because it makes a greedy selection at each iteration
and only chooses $v_{1}$ as the start node, who enjoys the activity the most
at the first iteration, but the optimal solution is \{$v_{2}$, $v_{3}$, $%
v_{4}$\} with the total willingness being $30$. 

\begin{figure}[t]
\centering
\includegraphics[height=0.659in, width=2.2567in]{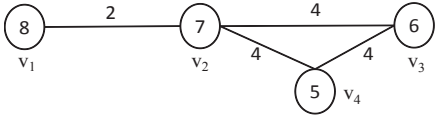} \vspace{-3mm} 
\caption{Counterexample of greedy algorithm}
\vspace{-20pt}
\label{counterexample}
\end{figure}
Another approach is to examine the willingness of every possible combination
of $k$ attendees. However, this enumeration approach needs to evaluate $%
C_{k}^{n}$ candidate groups, where $n$ is the number of nodes in $G$. In
current social networking websites, the number of candidate groups is still
huge even when we focus on only the candidates located in the same area,
e.g., about ten thousand users in British Virgin Islands\footnote{%
http://www.socialbakers.com/facebook-statistics/.}. When $k=50$, the number
of candidate groups is in the order of $10^{135}$. Thus, this
approach is computationally intractable for a massive social network.

Indeed, we show that the problem is challenging and prove that it is
NP-hard. As shown in Figure \ref{counterexample}, the greedy approach
improperly chooses $v_{1}$ as the start node and explores only a single
sequence of nodes in the solution space. To increase the search space, randomized algorithms have been proposed as a
simple but effective strategy to solve the problems with large instances 
\cite{MitzenmacherRand05}. To avoid being trapped in a local optimal
solution, a simple randomized algorithm for WASO is to randomly choose
multiple start nodes. Each start node is considered as partial solution, and
a node neighboring the partial solution is randomly chosen and added to
the partial solution at each iteration afterward, until $k$\textbf{\ }nodes
are included as a final solution. This randomized algorithm is more
efficient than the greedy approach, because the computation of willingness
is not involved during the selection of a node. For the
problem with a large $k$, numerous candidate nodes neighboring the
partial solution are necessary to be examined in the greedy approach to sum
up the willingness, in order to find the one that generates the largest
willingness. In contrast, the randomized algorithm simply chooses one
neighboring node at random.

With randomization, the aforementioned algorithm is able to effectively
avoid being trapped in a local optimal solution. It suffers, however, two disadvantages. Firstly, a start node that has the potential to
generate final solutions with high willingness is not invested with more
computational budgets for randomization in the following iterations. Each
start node in the randomized algorithm is expanded to only one final
solution. Thus, a start nodes, which has the potential to grow and become
the solution with high willingness, may fail to generate a final solution
with high willingness because only one solution is randomly constructed and
expanded from the start node. The second disadvantage is that the expansion of the partial solution does
not differentiate the selection of the neighboring nodes. Each neighboring
node is treated equally and chosen uniformly at random for each iteration. In
contrast, a simple way to remedy this issue is to assign the probability to
each neighboring node according to its interest score and social tightness
of incident edges. However, this assignment is similar to the greedy
algorithm in that it limits the scope to the local information corresponding to each
node and is not expected to generate a solution with high willingness.

Keeping in mind the above observations in an effort to guide an efficient search of the solution
space, we propose two randomized algorithms, called \emph{CBAS}
(Computational Budget Allocation for Start nodes) and \emph{CBAS-ND }%
(Computation Budget Allocation for Start nodes with Neighbor
Differentiation), to address the above two crucial factors in selecting
start nodes and expanding the partial solutions, respectively. This paper
exploits the notion of Optimal Computing Budget Allocation (OCBA) \cite%
{OCBA10} in randomization, in order to optimally invest more computational
budgets in the start nodes with the potential to generate the solutions with
high willingness. \emph{CBAS} first selects $m$ start nodes\footnote{%
The setting of $m$ and other parameters is important
and will be studied in the end of Section \ref{Exp}.} and then\emph{\ }randomly
adds neighboring nodes to expand the partial solution stage-by-stage, until $%
k$ nodes are included as a final solution. Each start node in \emph{CBAS} is
expanded to multiple final solutions. To properly invest the computational
budgets, \emph{CBAS} at each stage identifies the start nodes worth more
computational budgets according to sampled results of the previous stages.
Equipped with the allocation strategy of computational resources, \emph{CBAS}
is enhanced to \emph{CBAS-ND} to adaptively assign the probability to each
neighboring node during the expansion of the partial solution according to
the cross entropy method. We prove that the allocation of computational
budgets for start nodes and the assignment of the probability to each node
are both optimal in \emph{CBAS} and \emph{CBAS-ND}, respectively. We further
show that \emph{CBAS} can achieve an approximation ratio, while \emph{CBAS-ND}
needs much smaller computational budgets than \emph{CBAS} to acquire the
same solution quality.

The contributions of this paper can be summarized as follows. \vspace{-2mm}

\begin{itemize}
\item We formulate a new optimization problem, namely WASO, to consider the
topic interest of users and social tightness among friends for automatic
planning of activities. We prove that WASO is NP-hard. To the best of the
authors' knowledge, there is no real system or existing work in the
literature that addresses the issue of automatic activity planning based on
both topic interest and social relationship. \vspace{-3mm}

\item We design Algorithm \emph{CBAS} and \emph{CBAS-ND} to find the
solution to WASO with an approximation ratio. Experimental results demonstrate
that the solution returned by \emph{CBAS-ND }is very close to the optimal
solution obtained by IBM\ CPLEX, which is widely regarded as the fastest
general parallel optimizer, and \emph{CBAS-ND} is faster
than CPLEX.\vspace{-3mm}

\item We implement \emph{CBAS-ND} in Facebook and conduct a user
study with 137 people. Currently, people are used to organizing an activity
manually without being aware of the quality of the organized group, because
there is no automatic group recommendation service available for comparison.
Compared with the manual solutions, we observe that the solutions obtained by \emph{CBAS-ND} are 50.6\% better. In addition, 98.5\% of users
conclude that the group recommended by \emph{CBAS-ND} is better or acceptable.
Therefore, this research result\textbf{\ }has the potential to be adopted in
social networking websites as a value-added service.
\end{itemize}\vspace{-3mm}

The rest of this paper is summarized as follows. Section \ref{Prilim}
formulates WASO and surveys the related works. Sections \ref{SIIGQ-CBA} and %
\ref{CE_method} explain \emph{CBAS} and \emph{%
CBAS-ND} and derive the approximation ratio. User study and experimental
results are presented in Section \ref{Exp}, and we conclude this paper in
Section \ref{Conclu}.
\vspace{-2mm}
\section{Preliminary}

\label{Prilim}\vspace{-2mm}

\subsection{Problem Definition}

\label{SGQProb}

Given a social network $G=(V,E)$, where each node $v_{i}\in V$ and each edge 
$e_{i,j}\in E$ are associated with an interest score $\eta _{i}$ and a
social tightness score $\tau _{i,j}$ assigned according to the literature 
\cite{PowerLawClauset09} and \cite%
{ChaojiNnodeLink12} respectively, this paper studies a new optimization
problem WASO for finding a set $F$ of vertices with size $k$ to maximize the
willingness $W(F)$, i.e.,\vspace{-2.5mm} 
\begin{equation}
\max_{F}W(F)=\max_{F}\sum_{v_{i}\in F}(\eta _{i}+\sum_{v_{j}\in F:e_{i,j}\in
E}\tau _{i,j})\text{,}
\end{equation}
where $F$ is a connected subgraph in $G$ to encourage each attendee to be
acquainted with another attendee according to a social path in $F$. Notice that the social tightness between $v_{i}$ and $v_{j}$ is not necessarily symmetric, i.e., $\tau_{i,j}$ can be different with $\tau_{j,i}$. Therefore, the willingness in Eq. (1) considers both $\tau_{i,j}$ and $\tau_{j,i}$. It is worth noting that the illustrating example in the paper is symmetric for simplicity. As
demonstrated in the previous works in Psychology and social networks \cite%
{YangSI12,LeesI12} that jointly consider the social and interest
domains, the willingness of a group is represented as the sum of the topic
interest of nodes and social tightness between them\footnote{%
Different
weights $\lambda $ and (1-$\lambda $) can be assigned to the interest scores
and social tightness such that $W(F)=\sum_{v_{i}\in F}(\lambda _{i}\eta _{i}$%
+\newline
$(1-\lambda _{i})\sum_{v_{j}\in F:e_{i,j}\in E}\tau _{i,j})$. $\lambda _{i}$
can be set directly by a user or according to the existing model \cite{LeesI12}.
The impacts of different $\lambda $ will be studied later in Section \ref%
{Exp}.}.

Notice that the network with $\eta$ as 0 or $\tau$ as 0 is a special case of
WASO. Previous works \cite{DeutschGDM,KaplanGDM} demonstrated that both
social tightness and topic interest are intrinsic criteria involved in the
decision of a person to join a group activity, which is in line with the
results of our user study presented in Section \ref{Exp}. WASO is
challenging due to the tradeoff in interest and social tightness, while the
constraint that assures that the $F$ is connected also complicates this problem
because it is no longer able to arbitrarily choose any nodes from $G$.
Indeed, the following theorem shows that WASO is NP-hard.
\vspace{-2mm} 
\begin{theorem}
WASO is NP-hard.
\end{theorem}

\begin{proof}
\textbf{\ }We prove that WASO is NP-hard with the reduction from Dense $k$-Subgraph (DkS) \cite{Feige01}. Given a graph $%
G_{D}=(V_{D},E_{D})$, DkS finds a subgraph with $k$ nodes $F_{D}$ to
maximize the density of the subgraph. In other words, the purpose of DkS is to maximize the
number of edges $E(F_{D})$ in the subgraph induced by the selected nodes.
For each instance of DkS, we construct an instance for WASO by letting $%
G=G_{D}$, where $\eta _{i}$ of each node $v_{i}\in V$ is set as $0$, and $%
\tau _{i,j}$ of each edge $e_{i,j}\in E$ is assigned as $1$. We first prove
the sufficient condition. For each instance of DkS with solution node set $%
F_{D}$, we let $F=F_{D}$. If the number of edges $E(F_{D})$ in
the subgraph of DkS is $\epsilon $, the willingness of WASO $W(F)$ is also $%
\epsilon $ because $F=F_{D}$. We then prove the necessary condition. For
each instance of WASO with $F$, we select the same nodes
for $F_{D}$, and the number of edges $E(F_{D})$ must be the same as $W(F)$.
The theorem follows. \vspace{-2mm}
\end{proof}

\subsection{Scenarios}
\label{scenario}
In the following, we present the parameter settings of WASO to fit the need
of different scenarios.

\emph{Couple and Foe:} For any two people required to
be selected together, such as a couple, the two corresponding nodes $v_{i}$
and $v_{j}$ in $G$ are merged into one node $v_{a}$ with the interest score $\eta_{a}=\eta
_{i}+\eta _{j}$ and social tightness score $\tau_{a,b}=\tau_{i,b}+\tau_{j,b}$ for each neighboring node $v_{b}$ of $v_{i}$ or $v_{j}$. Similarly, more people can be merged as well, but the group size $k$ is required to be adjusted accordingly. On the other hand, if $v_{i}$ is a foe of $v_{j}$, their social
tightness score $\tau _{j,i}$ is assigned a large negative value, such that
any group consisting of the two nodes leads to a negative willingness and
thereby will not be selected. The relationship of foes can be discovered by blacklists and learnt from historical records. Similarly, $\eta _{i}$ is allowed to be assigned a negative value.

\emph{Invitation:} A piano player plans to hold a small concert. In this case, the player might prefer inviting people that are very
good friends with him/her, but it is not necessary for them to be pair-wise acquainted. For this scenario,
the activity candidates are the neighboring nodes of $v_{i}$, which is denoted as $N(v_{i})$, where $v_{i}$ represents the inviter (piano player), and we set $\lambda _{j}$ as 1 for every $j \in N(v_i)$ since the social tightness among the friends may not be important in this scenario.

\emph{Exhibition and house-warming party:} The British Museum plans to hold an exhibition of Van Gogh and would like to send e-mails to potential visitors. In this scenario, the topic interest is expected to play a crucial role, and $\lambda_{i}$ is suitable to set as 1 for all $i \in V$. On the other hand, for social activities such as a house-warming
party, $\lambda_{i}$ is 0 for all $i \in V$, and only social tightness is considered.

\emph{Separate Groups:} The government plans to organize a camping trip on Big Bear Lake to promote environmental protection. In this case, the group does not need to be connected, and a simple way is to
add a virtual node $\overline{v}$ to $V$ with the interest score of $\overline{v}$ as $\eta_{\overline{v}}$=$\epsilon$+$\sum_{v_{i}\in V}(\eta _{i}+\sum_{v_{j}\in V:e_{i,j}\in E}\tau _{i,j})$, where $\epsilon$ is any positive real number. In addition, $\overline{v}$ is connected to every other node $v_{j} \in V$ with the social tightness score $\tau _{\overline{v},j}$=$0$, and the set of new edges incident to $\overline{v}$ is denoted as $E_{\overline{v}}$. It is necessary to choose $\overline{v}$ so that $\overline{v}$ will connect to multiple
disconnected subgraphs to support the above group activities. In this case, $%
k+1$ nodes need to be included in the final solution.

Now WASO-dis denote the counterpart of WASO without the connectivity constraint. Indeed, WASO-dis is simpler than WASO because the constraint is not incorporated. In the following, we prove that WASO can be reduced from WASO-dis. In other words, any algorithm for WASO can also solve WASO-dis. More specifically, given a graph $G_{d}=(V_{d},E_{d})$, WASO-dis finds a subgraph $F_{d}$ with $k$ nodes to maximize the total willingness of the subgraph, i.e., $\max_{F_{d}}\sum_{v_{i}\in F_{d}}(\eta _{i}+\sum_{v_{j}\in F_{d}:e_{i,j}\in E_{d}}\tau _{i,j})$, where $F_{d}$ is a subgraph in $G_{d}$ and not required to be a connected subgraph. For any instance $G_{d}$ of WASO-dis, we construct an instance $G$ for WASO by adding a new node $\overline{v}$ as follows. Let $G=(V,E)$=$(V_{d}\bigcup \overline{v}$, $E_{d}\bigcup E_{\overline{v}})$, where the interest score of $\overline{v}$ is $\eta_{\overline{v}}$=$\epsilon$+$\sum_{v_{i}\in V_{d}}(\eta _{i}+\sum_{v_{j}\in V_{d}:e_{i,j}\in E_{d}}\tau _{i,j})$ with $\epsilon$ as any positive real number. In addition, $\overline{v}$ is connected to every other node $v_{j} \in V_{d}$ with the social tightness score $\tau _{\overline{v},j}$=$0$, and the set of new edges incident to $\overline{v}$ is denoted as $E_{\overline{v}}$. Therefore, $\overline{v}$ will appear in the optimal solution of any WASO instance due to its high interest score.

\vspace{-2mm} 
\begin{theorem}
\label{LemmaWasoDis}
$F^*_d$ is the $k$-node optimal solution of any WASO-dis instance $G_d$ if and only if $F^*$ is the $k+1$-node optimal solution of WASO instance $G$, where $F^*=F^*_d \bigcup {\overline{v}}$.
\end{theorem}
\begin{proof}
We first prove the sufficient condition. Since the optimal solution of WASO must include $\overline{v}$, if $F^*_d \bigcup {\overline{v}}$ is not optimal to WASO, there exists a better solution $\overline{F}$ with $W(\overline{F})>W(F^*_d \bigcup {\overline{v}})$, which implies that $W(\overline{F}-{\overline{v}})>W(F^*_d)$. Because $\overline{F}-{\overline{v}}$ can act as a feasible solution to WASO-dis, $W(\overline{F}-{\overline{v}})>W(F^*_d)$ contradicts that $F^*_d$ is optimal to WASO-dis. Therefore, $F^*=F^*_d \bigcup {\overline{v}}$ is the optimal solution to WASO. 

We then prove the necessary condition. Since the optimal solution of WASO must include $\overline{v}$, if $F^*-{\overline{v}}$ is not optimal to WASO-dis, there exists a better solution $\overline{F}_d$ with $W(\overline{F}_d)>W(F^*-{\overline{v}})$, implying that $W(\overline{F}_d+{\overline{v}})>W(F^*)$, contradicting that $F^*$ is optimal to WASO because $\overline{F}_d+{\overline{v}}$ is also a feasible solution to WASO. The lemma follows.
\end{proof}

\subsection{Related Works}

\label{RelatedWorks}

Given the growing importance of varied social networking applications, there has been a
recent push on the study of user interest scores and social tightness scores
from real social networking data. It has been demonstrated that unknown user
interest attributes can be effectively inferred from a social network
according to the revealed attributes of the friends \cite{MisloveMI}. On the other hand, Wilson%
\textit{\ et al.} \cite{WilsonSocialTies09} derived a new model to quantify
the social tightness between any two friends in Facebook. The number of wall
postings is also demonstrated to be an effective indicator for social tightness \cite%
{GilbertSocialTies09}. Thus, the above studies provide
a sound foundation to quantify the user interest and social tightness scores
in social networks. Moreover, Yang \cite{YangSI12} and Lee \cite%
{LeesI12} sum up the two factors as willingness for marketing and recommendation.
Nevertheless, the above factors crucial in social networks have not been
leveraged for automatic activity planning explored in this paper.

Expert team formation in social networks has attracted extensive research
interests. The problem of constructing an expert team is to find a
set of people owning the specified skills, while the communications cost
among the chosen friends is minimized to ensure the rapport among the team
members for an efficient operation. Two communications costs, diameter and minimum spanning tree, were evaluated. Several extended models have been
studied. For example, each skill $i$ needs to contain at least $k_{i}$
people in order to form a strong team \cite{LiExperts10}, while all-pair
shortest paths are incorporated to describe the communications costs more
precisely \cite{KargarExperts11}. Moreover, a skill leader is selected for
each skill with the goal to minimize the social distance from the skill
members to each skill leader \cite{KargarExperts11}, while the density of a
team is also considered\textit{\ }\cite{GajewarExperts12}.

In addition to expert team formation, community detection as well as graph
clustering and graph partitioning have been explored to find groups of nodes
mostly based on the graph structure \cite{BrandesCD08}. The quality of an
obtained community is usually measured according to the structure inside the
community, together with the connectivity within the community and between the rest
of the nodes in the graph, such as the density of local edges, deviance from
a random null model, and conductance \cite{SeshadhriCD12}. Sozio et al. \cite%
{SozioCocktail10}, for example, detected community by minimizing the total degree of a
community with specified nodes. However, the objective function of WASO is different from community detection. Each node and each edge in WASO are associated with an interest
score and social tightness score in the problem studied in this paper, in order to
maximize the willingness of the attendees with a specified group size, which
can be very useful for social networking websites as a
value-added service.
\vspace{-3mm}
\section{Algorithm Design for WASO}

\label{SIIGQ-CBA}

To solve WASO, a greedy approach incrementally constructs the solution
by sequentially choosing an attendee that leads to the largest increment in
the willingness at each iteration. However, while this approach is simple, it tends to be trapped in a local optimal solution. The search
space of the greedy algorithm is limited because only a single sequence of
nodes is explored. To address the above issues, this paper first proposes a
randomized algorithm \emph{CBAS} to randomly choose $m$ start nodes. Each
start node acts as a seed to be expanded to multiple final solutions. At
each iteration, a partial solution, which consists of only a start node at
the first iteration or a connected set of nodes at any iteration afterward,
is expanded by uniformly selecting at random a node neighboring the
partial solution, until $k$ nodes are included. We leverage the notion of
Optimal Computing Budget Allocation (OCBA) \cite{OCBA10} to randomly
generate more final solutions from each start node that has more potential
to generate the final solutions with high willingness. Later we will prove
that the number of final solutions generated from each start node is
optimally assigned.

After this, we enhance \emph{CBAS} to \emph{CBAS-ND} by differentiating the
selection of the nodes neighboring each partial solution. During each
iteration of \emph{CBAS}, each neighboring node is treated equally and
chosen uniformly at random. A simple way to improve \emph{CBAS} is to
associate each neighboring node with a different probability according to
its interest score and social tightness scores of incident edges.
Yet, this assignment is similar to the greedy algorithm insofar as it limits the
scope to only the local information associated with each node thereby making it difficult to generate a final solution with high willingness. To prevent the generation of only a local optimal solution, \emph{CBAS-ND} deploys
the cross entropy method according to results at the previous stages in
order to optimally assign a probability to each neighboring node.

One advantage of the proposed randomized algorithms is that the tradeoff
between the solution quality and execution time can be easily controlled by
assigning different $T$, which denotes the number of randomly generated
final solutions. Under a given $T$, if $m$ start nodes are generated, the
above algorithms can optimally divide $T$ into $m$ parts for the $m$ start
nodes to find final solutions with high willingness. Moreover, we prove that 
\emph{CBAS} is able to find a solution with an approximation ratio. Compared
with \emph{CBAS}, we further prove that the solution quality of \emph{CBAS-ND} is better with the same computation budget\footnote{%
It is worth noting that randomization is performed only in expanding a start
node to a final solution, not in the selection of a start node. This is because the
approximation ratio is not able to be achieved if a start node is decided
randomly.}. The detailed settings of $T$ and $m$ will be analyzed in Section %
\ref{Exp}. In addition, the notation table with their impacts on the solution are shown in Table \ref{table:notation}.

\begin{table}[t]
\caption{Parameter Summary}
\label{table:notation}
\begin{tabular}{c|p{3cm}|p{3cm}}
\hline\hline
Notation & Description & Impact \\ \hline 
$\tau_{i,j}$ & social tightness score between node $v_{i}$ and $v_{j}$ & set to be a negative value if $v_{i}$ and $v_{j}$ are foes\\ 
$\eta_{i}$ &  interest score of node $v_{i}$ & set to be a negative value if $v_{i}$ does not like the topic\\ 
$\lambda{i}$ & weighting between interest score and tightness score of node $v_{i}$ & set to be zero if $v_{i}$ only considers social tightness and one if $v_{i}$ only concerns topic interest\\
$T$ &  total computation budget & trade-off between solution quality and computation time\\
$m$ &  number of start nodes & sampling coverage\\ \hline
\end{tabular}
\vspace{-15pt}
\end{table}

In the following, we first present \emph{CBAS} to optimally allocate the
computational budgets to different start nodes (Section \ref{CBAAlgorithmdescription}) and then derive the approximation ratio in Section %
\ref{CBA_Thm}. Algorithm \emph{CBAS-ND} will be presented in Section \ref%
{CE_method}.
\vspace{-1mm}
\subsection{Allocation of Computational Budget for \newline
Start Nodes}

\label{CBAAlgorithmdescription}

Given the total computational budgets $T$ specified by users, a simple
approach first randomly selects $m$ start nodes and then expands each
start node to $\frac{T}{m}$ final solutions. However, this homogeneous
approach does not give priority to the start nodes that have more potential
to generate final solutions with high willingness. In contrast, \emph{%
CBAS} optimally allocates more resources to the start nodes with high willingness
with the following phases.
\vspace{-2mm}
\begin{enumerate}
\item \emph{Selection and Evaluation of Start Nodes:} This phase first
selects $m$ start nodes according to the interest scores and social
tightness scores. Afterward, each start node is randomly expanded to a few
final solutions. We iteratively select and add a neighboring node uniformly
at random to a partial solution, until $k$ nodes are selected. The
willingness of each final solution is evaluated for the next phase to
allocate different computational budgets to different start nodes.
\vspace{-3mm}
\item \emph{Allocation of Computational Budgets:} This phase derives the
computational resources optimally allocated to each start node according to
the previous sampled willingness.
\end{enumerate}
\vspace{-3mm}
To optimally allocate the computational budgets for each start node, we first
define the solution quality as follows.

\vspace{-2mm}
\begin{definition}
The solution quality, denoted by Q, is defined as the maximum willingness among all maximal sampled results of the m start nodes, 
\end{definition}%
\vspace{-10pt}
\begin{equation*}
Q=max\{J_{1}^{\ast },J_{2}^{\ast },...,J_{i}^{\ast },...,J_{m}^{\ast }\},
\end{equation*}%
where $J_{i}^{\ast }$ is a random variable representing the maximal
willingness sampled from a final solution expanded from start node $v_{i}$.

Since the maximal sampled result $J_{i}^{\ast }$ of start node $v_{i}$ is
related to the number of sampling times $N_{i}$, i.e., the number of final
solutions randomly generated from $v_{i}$, the mathematical formulation to
optimize the computational budget allocation is defined as\vspace{-2mm} 
\begin{equation*}
\max_{N_{1},N_{2},...,N_{m}}Q,
\end{equation*}%
\vspace{-4mm}
\begin{equation*}
\text{s.t. }N_{1}+N_{2}+...+N_{m}=T.
\end{equation*}%
Let $v_{b}$ denote the start node that are able to generate the solution
with the highest willingness. Obviously, the optimal solution in the above
maximization problem is to allocate all the computational budgets to $v_{b}$%
. However, since $v_{b}$ is not given \textit{a priori}, \emph{CBAS}
divides the resource allocation into $r$ stages, and each stage adjusts the
allocation of computational budgets $\frac{T}{r}$ to different start nodes
according to the sampled willingness from the partial solutions in previous
stages.

For each node, phase 1 of \emph{CBAS }first adds the interest score and the
social tightness scores of incident edges and then chooses the $m$ nodes
with the largest sums as the $m$ start nodes. On the other hand, allocating more
computational budgets to the start node with a larger sum, similar to the
greedy algorithm, does not tend to generate a final solution with high
willingness. For this reason, phase 2 evaluates the sampled willingness to
allocate different computational budgets to each start node.

In stage $t$ of phase 2, let $N_{i,t}$ denote the computational budgets
allocated to start node $v_{i}$ at the $t$-th stage. The ratio of
computational budgets $N_{i,t}$ and $N_{j,t}$ allocated to any two start
nodes $v_{i}$ and $v_{j}$ is \vspace{-2mm} 
\begin{equation*}
\frac{N_{i,t}}{N_{j,t}}=(\frac{d_{i}-c_{b}}{d_{j}-c_{b}})^{N_{_{b}}},
\end{equation*}%
where $d_{i}$ denotes the best sampled willingness of the partial solutions
expanded from start node $v_{i}$ in the previous stages $1,...,t-1$. Notice
that $v_{b}$ here is the start node that enjoys the highest willingness
sampled in the previous stages, $N_{b}$ is the overall computational budgets
allocated to $v_{b}$ in the previous stages, and $c_{b}$ denotes
the worst sampled willingness of the partial solution expanded from start
node $v_{b}$ in the previous stages. Later, we will prove that the
above budgets allocation in each stage is optimal. However, if the
allocated computational budgets for a start node is $0$ at the $t$-th stage,
we prune off the start node in the following $(t+1)$-th stage.

\begin{example}
Figure \ref{IllustratingExample} presents an illustrative
example for \emph{CBAS} with $n=10$, $k=5$, and $m=2$.
Phase 1 first chooses two start nodes by summing up the topic interest score
and the social tightness scores for every node. Therefore, $v_{3}$ with $%
0.8+0.6+0.5+0.9+1+0.4=4.2$ and $v_{10}$ with $0.9+0.6+1+0.9+0.8=4.2$ are
selected. Next, let $T=20$, $P_{b}=0.7$ and $\alpha =0.9$ in this example,
and the number of stages is thus $r\leq \frac{Tk\ln \alpha }{n\ln (\frac{%
2(1-P_{b})}{\frac{n}{k}-1})}=\frac{20\cdot 5\ln 0.9}{10\ln (0.6)}\approx 2.$
Each start node generates $5$ samples at the first stage. In the
beginning, the node selection probability of start node $v_{3}$, i.e., $%
\overrightarrow{p}_{3,1}$, is set to be $\langle\frac{4}{9}$,$\frac{4}{9},1,%
\frac{4}{9},\frac{4}{9},\frac{4}{9},\frac{4}{9},\frac{4}{9},\frac{4}{9},%
\frac{4}{9}\rangle$. The intermediate solution obtained so far is denoted as 
$V_{S}$, and the candidate attendees extracted so far is denoted as $V_{A}$.
Therefore, the total willingness of $V_{S}$\ $=\{v_{3}\}$\ is $0.8$, and $%
V_{A}=\{v_{1},v_{2},v_{4},v_{5},v_{6}\}$. Since the node selection
probability is homogeneous in the first stage, we randomly select $v_{6}$%
\ from $V_{A}$\ to expand $V_{S}$. Now the total willingness of $V_{S}$\ $=\{v_{3},v_{6}\}$ is $W(V_{S})=0.8+0.4+0.9=2.1$, and $%
V_{A}=\{v_{1},v_{2},v_{4},v_{5},v_{7},v_{8},v_{10}\}$. The process of
expanding $V_{S}$\ continues until the cardinality of $V_{S}$\ reaches $5$,
and we record the first sample result $X_{3,1}=\langle
1,0,1,1,1,1,0,0,0,0\rangle$ with the total willingness $8.9$, the worst result
of $v_{3}$ ($c_{3}=8.9$), and the best result of $v_{3}$ ($d_{3}=8.9$). The
other sample results from start node $v_{3}$\ are $X_{3,2}=\langle
1,1,1,1,1,0,0,0,0,0\rangle$ with the total willingness $8.9$, $X_{3,3}=\langle
0,1,1,0,1,1,0,1,0,0\rangle$ with the total willingness $5.9$, $X_{3,4}=\langle
0,1,1,1,1,0,1,0,0,0\rangle$ with the total willingness $7.9$, and $%
X_{3,5}=\langle 0,0,1,0,1,1,1,0,$\ $0,1\rangle$ with the total willingness $9.2$.
The worst and the best results of $v_{3}$\ are updated to $c_{3}=5.9$ and $%
d_{3}=9.2$, respectively. After sampling from node $v_{3}$, we repeat the above process for start node $v_{10}$. The worst result is $c_{10}=6.9$, and the best result is $d_{10}=8.9$.

To allocate the computational budgets for the second stage, i.e., $r=2$, we
first find the allocation ratio $r_{3}:r_{10}$=$1:(\frac{8.8-5.9}{9.2-5.9})^{5}$%
=$1:0.524$. Therefore, the allocated computational budgets for start nodes $%
v_{3}$\ and $v_{10}$\ are $\frac{10}{1.524}\approx 7$\ and $\frac{5.24}{1.524}%
\approx 3$, respectively. At the second stage, the best results of $v_{3}$\
and $v_{10}$\ are $9.2$\ and $8.9$, respectively. Finally, we obtain the
solution $\{v_{3},v_{5},v_{6},v_{7},v_{10}\}$\ with the total willingness $9.2$.
\end{example}

\subsection{Theoretical Result of CBAS}
\label{CBA_Thm}

To correctly allocate the computational budgets $T$ to $m$ start nodes, we
first derive the optimal ratio of computational budgets for any two start
nodes. Afterward, we find the probability $P_{b}$ that node $v_{b}$ is
actually the start node which is able to generate the highest willingness in
each stage. Finally, we derive the approximation ratio and analyze the complexity of \emph{CBAS}.

\vspace{-1mm}
\begin{definition}
A random variable, denoted as  $J_{i}$, is defined to be the sampled value in start node $v_{i}$.
\end{definition}%
\vspace{-1mm}
The literature of OCBA indicates that the distribution of random variable $%
J_{i}$ in most applications is a normal distribution, but the allocation
results are very close to the one with the uniform distribution \cite{OCBA10,OCBA2000}. Therefore, given space constraints, $J_{i}$ here is first
handled as the uniform distribution in [$c_{i},d_{i}$], and the derivation
for the normal distribution is presented in Appendix. The probability density function and cumulative distribution function are
formulated as \vspace{-6pt} 
\begin{equation*}
p_{J_{i}}(x)=\left\{ 
\begin{aligned}
\frac{1}{d_{i}-c_{i}} &~~ \text{     if }c_{i}\leq x\leq d_{i}\\ 
0 &~~ \text{     otherwise.} \\
\end{aligned}\right.
\end{equation*}%
\vspace{-2pt} 
\begin{equation*}
P_{J_{i}}(x)=\left\{ 
\begin{aligned}
0&~~\text{     if }x\leq c_{i}. \\ 
\frac{x-c_{i}}{d_{i}-c_{i}}&~~\text{     if }c_{i}\leq x\leq d_{i}. \\ 
1&~~\text{     otherwise.}%
\end{aligned}
\right.
\end{equation*}%
Therefore, for the maximal value $J_{i}^{\ast }$, \vspace{-2pt} 
\begin{equation*}
p_{J_{i}^{\ast }}(x)=N_{i}P_{J_{i}}(x)^{N_{i}-1}p_{J_{i}}(x),
\end{equation*}%
\vspace{-3mm}
\begin{equation*}
P_{J_{i}^{\ast }}(x)=P_{J_{i}}(x)^{N_{i}}.
\end{equation*}%
\vspace{-18pt}

\begin{theorem}
\label{Prababilitya}Given the best start node $v_{b}$, the probability that $%
J_{i}^{\ast }$ exceeds $J_{b}^{\ast }$ is at most $\frac{1}{2}(\frac{%
d_{i}-c_{b}}{d_{b}-c_{b}})^{N_{b}}.$
\end{theorem}

\vspace{-2mm}
\begin{proof}
For $p(J_{b}^{\ast }\leq J_{i}^{\ast }),$ 
\begin{eqnarray*}
p(J_{b}^{\ast }-J_{i}^{\ast } &\leq &z) \\
&=&\int_{c_{b}}^{d_{b}}p_{J_{b}^{\ast }}(x)(1-P_{J_{i}^{\ast }}(x-z))dx \\
\text{ \ \ \ } &=&\int_{c_{b}}^{d_{b}}p_{J_{b}^{\ast
}}(x)dx-\int_{c_{b}}^{d_{b}}p_{J_{b}^{\ast }}(x)P_{J_{i}^{\ast }}(x-z)dx.
\end{eqnarray*}%
Let $z$ equal zero. $p$($J_{b}^{\ast }-J_{i}^{\ast }$ $\leq$ $0$) 
\begin{eqnarray*}
&=&1-%
\int_{c_{b}}^{d_{b}}N_{b}P_{J_{b}}(x)^{N_{b}-1}p_{J_{b}}(x)P_{J_{i}}(x)^{N_{i}}dx
\end{eqnarray*}%
\begin{eqnarray*}
&=&-(\int_{c_{b}}^{d_{i}}N_{b}(\frac{x-c_{b}}{d_{b}-c_{b}})^{N_{b}-1}\frac{1%
}{d_{b}-c_{b}}(\frac{x-c_{i}}{d_{i}-c_{i}})^{N_{i}}dx \\
&&+\int_{d_{i}}^{d_{b}}N_{b}(\frac{x-c_{b}}{d_{b}-c_{b}})^{N_{b}-1}\frac{1}{%
d_{b}-c_{b}}dx)+1.
\end{eqnarray*}%
It is worth noting that $d_{i}>c_{b}$ holds in the above equation. Otherwise, the
probability that $J_{b}^{\ast }$ is smaller than $J_{i}^{\ast }$ will be
zero, i.e., $p(J_{b}^{\ast }<J_{i}^{\ast })=0$. We further change the
variables by letting $\frac{x-c_{b}}{d_{b}-c_{b}}$ be $u$, and $%
(d_{b}-c_{b})du=dx.$ 
\vspace{-6pt} 
\begin{eqnarray*}
&&\frac{1}{d_{b}-c_{b}}\int_{c_{b}}^{d_{i}}N_{b}(\frac{x-c_{b}}{d_{b}-c_{b}}%
)^{N_{b}-1}(\frac{x-c_{i}}{d_{i}-c_{i}})^{N_{i}}dx \\
&&+\int_{d_{i}}^{d_{b}}N_{b}(\frac{x-c_{b}}{d_{b}-c_{b}})^{N_{b}-1}dx
\end{eqnarray*}%
\vspace{-10pt} 
\begin{eqnarray*}
&=&N_{b}\int_{0}^{\frac{d_{i}-c_{b}}{d_{b}-c_{b}}}u^{N_{b}-1}(\frac{%
u(d_{b}-c_{b})+c_{b}-c_{i}}{d_{i}-c_{i}})^{N_{i}}du \\
&&+\frac{1}{d_{b}-c_{b}}\int_{\frac{d_{i}-c_{b}}{d_{b}-c_{b}}%
}^{1}N_{b}u^{N_{b}-1}(d_{b}-c_{b})du \\
&=&\int_{0}^{\frac{d_{i}-c_{b}}{d_{b}-c_{b}}}N_{b}u^{N_{b}-1}(\frac{%
u(d_{b}-c_{b})+c_{b}-c_{i}}{d_{i}-c_{i}})^{N_{i}}du \\
&&+(1-(\frac{d_{i}-c_{b}}{d_{b}-c_{b}})^{N_{b}}).
\end{eqnarray*}
For ease of reading, we denote $g_{b,i}$ as $\frac{d_{b}-c_{b}}{d_{i}-c_{i}}$
and $h_{b,i}$ as $\frac{c_{b}-c_{i}}{d_{i}-c_{i}}$. Then the binomial
theorem is employed for expanding the polynomial term. \vspace{-2mm} 
\begin{eqnarray*}
&&\int_{0}^{\frac{d_{i}-c_{b}}{d_{b}-c_{b}}}N_{b}u^{N_{b}-1}(\frac{%
u(d_{b}-c_{b})+c_{b}-c_{i}}{d_{i}-c_{i}})^{N_{i}}du \\
&=&N_{b}\int_{0}^{\frac{d_{i}-c_{b}}{d_{b}-c_{b}}%
}u^{N_{b}-1}(g_{b,i}u+h_{b,i})^{N_{i}}du \\
&=&N_{b}\int_{0}^{\frac{d_{i}-c_{b}}{d_{b}-c_{b}}%
}\sum_{q=0}^{N_{i}}C_{q}^{N_{i}}g_{b,i}{}^{q}(h_{b,i})^{N_{i}-q}u^{N_{b}-1+q}du\\
&=&N_{b}%
\sum_{q=0}^{N_{i}}C_{q}^{N_{i}}g_{b,i}{}^{q}(h_{b,i})^{N_{i}-q}u^{N_{b}+q}%
\left. \frac{1}{N_{b}+q} \right |^\frac{d_{i}-c_{b}}{d_{b}-c_{b}}_{0}
\end{eqnarray*}
Since $\frac{g_{b,i}{}}{h_{b,i}}\frac{d_{i}-c_{b}}{d_{b}-c_{b}}=\frac{%
d_{i}-c_{b}}{c_{b}-c_{i}}$, the above equation can be further simplified to \vspace*{-2mm} 
\begin{equation*}
(h_{b,i})^{N_{i}}(\frac{d_{i}-c_{b}}{d_{b}-c_{b}})^{N_{b}}N_{b}%
\sum_{q=0}^{N_{i}}C_{q}^{N_{i}}(\frac{g_{b,i}{}}{h_{b,i}})^{q}(\frac{%
d_{i}-c_{b}}{d_{b}-c_{b}})^{q}\frac{1}{N_{b}+q}
\end{equation*}%
\begin{equation}
=(h_{b,i})^{N_{i}}(\frac{d_{i}-c_{b}}{d_{b}-c_{b}})^{N_{b}}N_{b}%
\sum_{q=0}^{N_{i}}C_{q}^{N_{i}}(\frac{d_{i}-c_{b}}{c_{b}-c_{i}})^{q}\frac{1}{%
N_{b}+q}.  \label{Eq:NotSimp}
\end{equation}%
\vspace{-2mm} Then, the probability that $J_{i}^{\ast }$ is better than $%
J_{b}^{\ast }$, i.e., $p(J_{b}^{\ast }\leq J_{i}^{\ast })$ 
\begin{eqnarray*}
&=&(\frac{d_{i}-c_{b}}{d_{b}-c_{b}})^{N_{b}}(1-(h_{b,i})^{N_{i}}N_{b}%
\sum_{q=0}^{N_{i}}C_{q}^{N_{i}}(\frac{d_{i}-c_{b}}{c_{b}-c_{i}})^{q}\frac{1}{%
N_{b}+q}) \\
\end{eqnarray*}
\vspace{-26pt}
\begin{eqnarray*}
&\leq &(\frac{d_{i}-c_{b}}{d_{b}-c_{b}})^{N_{b}}(1-(h_{b,i})^{N_{i}}N_{b}%
\frac{1}{2N_{b}}\sum_{q=0}^{N_{i}}C_{q}^{N_{i}}(\frac{d_{i}-c_{b}}{%
c_{b}-c_{i}})^{q}) \\
&=&(\frac{d_{i}-c_{b}}{d_{b}-c_{b}})^{N_{b}}(1-\frac{1}{2}%
(h_{b,i})^{N_{i}}(1+\frac{d_{i}-c_{b}}{c_{b}-c_{i}})^{N_{i}}) \\
&=&\frac{1}{2}(\frac{d_{i}-c_{b}}{d_{b}-c_{b}})^{N_{b}}.\end{eqnarray*}\vspace{-8pt}
\end{proof}

With the result above, we allocate the computational budgets by
\begin{equation}
\frac{N_{i}}{N_{j}}=\frac{P(J_{i}^{\ast }\geq J_{b}^{\ast })}{P(J_{j}^{\ast
}\geq J_{b}^{\ast })}=(\frac{d_{i}-c_{b}}{d_{j}-c_{b}})^{N_{b}}.
\label{Eq:Thm1Eq}
\end{equation}

Since it is impossible to enumerate every final solution expanded from a
start node, the ratio of the computational budget allocation is optimal in
OCBA \cite{OCBA10} if the first equality in Eq. (\ref{Eq:Thm1Eq})
holds. Thus, it is optimal to allocate the computational budgets to $N_{i}$
and $N_{j}$ according to the ratio $(\frac{d_{i}-c_{b}}{d_{j}-c_{b}}%
)^{N_{b}} $. Notice that if $d_{i}$ is smaller than $c_{b}$, the probability
that $J_{b}^{\ast }$ is smaller than $J_{i}^{\ast }$ is zero.

Intuitively, the above result indicates that if the best random sample,
i.e., $d_{i}$, from a start node is small, it is unnecessary to repeat the
sampling process too many times since the users nearby the start node are
not really interested in the activity or they have an estranged friendship. On the
other hand, as the number of sample times increases, it is expected that the
identified best start node enjoys the highest willingness.

The following theorem first analyzes the probability $P_{b}$ that $v_{b}$, as decided according to the samples in the previous stages, is actually the
start node that generates the highest willingness. Let $\alpha$ denote the closeness ratio between the maximum of the start node with the highest willingness and the maximum of other start nodes, i.e., $\alpha=(d_{a}-c_{b})/(d_{b}-c_{b})$, where $v_{a}$ generates the maximum willingness among other start nodes. Therefore, in addition to $0$ and $1$, $\alpha$ is allowed to be any other value from $0$ to $1$.

\vspace{-2pt}
\begin{theorem}
\label{ControlComputationThm1} For WASO with parameter $(m,T)$, where $m$ is
the number of start nodes and $T$ is the total computational budgets, the
probability $P_{b}$ that $v_{b}$ selected according to the previous stages
is actually the start node with the highest willingness is at least $1-\frac{%
1}{2}(m-1)\alpha ^{\frac{T}{mr}}$.
\end{theorem}

\begin{proof}
According to the Bonferroni inequality, $p\{\cap _{i=1}^{m}(Y_{i}<0)\}\geq
1-\sum_{i=1}^{m}[1-p(Y_{i}<0)]$. In our case, $Y_{i}$ is replaced by $%
J_{i}^{\ast }-J_{b}^{\ast }$ to acquire a lower bound for the probability that $v_{b}$ enjoys the highest willingness. Therefore, by using Theorem \ref%
{Prababilitya},
\vspace{-1mm}
\begin{eqnarray*}
P_{b} &=&p\{\cap _{i=1,i\neq b}^{m}(J_{i}^{\ast }-J_{b}^{\ast }\leq 0)\} \\
&\geq &1-\sum_{i=1,i\neq b}^{m}[1-p(J_{i}^{\ast }-J_{b}^{\ast }\leq 0)] \\
\end{eqnarray*}%
\vspace{-8mm}
\begin{eqnarray*}
&=&1-\sum_{i=1,i\neq b}^{m}p(J_{b}^{\ast }\leq J_{i}^{\ast }) \\
&\geq &1-\frac{1}{2}\sum_{i=1,i\neq b}^{m}(\frac{d_{i}-c_{b}}{d_{b}-c_{b}}%
)^{N_{b}}.
\end{eqnarray*}\vspace{-0.2mm} 
Let $d_{i}=c_{b}+\alpha (d_{b}-c_{b})$, where $\alpha $ is close to 1. The above equation can be further simplified to $P_{b}\geq 1-\frac{1}{2}(m-1)\alpha^{N_{b}}$ $\geq$ $1-\frac{1}{2}(m-1)\alpha ^{\frac{T}{rm}}$.
\end{proof}\vspace{-2mm} Given the total budgets $T$, the following theorem
derives a lower bound of the solution obtained by \emph{CBAS}. \vspace{-2mm}

\begin{theorem}
\label{ControlComputationThm2} For a WASO optimization problem with $r$%
-stage computational budget allocation, the maximum willingness $E[Q]$ from
the solution of \emph{CBAS} is at least $N_{b}(\frac{1}{N_{b}+1})^{\frac{%
N_{b}+1}{N_{b}}}\cdot Q^{\ast }$, where $N_{b}$ after $r$ stages is $\frac{4+m(r-1)}{4rm}T$,
and $Q^{\ast }$ is the optimal solution.
\end{theorem}

\vspace{-2mm} 
\begin{proof}
We first derive the lower bound of $E[Q]$ as follows.
The random variable $Q$ is
denoted as$\ max\{J_{1}^{\ast },...,J_{m}^{\ast }\}$. The cumulative density
function is%
\begin{eqnarray*}
F_{Q}(Q &\leq &\Delta )=F(max\{J_{1}^{\ast },...,J_{m}^{\ast }\}\leq \Delta )
\\
&=&F(J_{1}^{\ast }\leq \Delta ,J_{2}^{\ast }\leq \Delta ,...,J_{m}^{\ast
}\leq \Delta ) \\
\end{eqnarray*}
\vspace{-10mm} 
\begin{eqnarray*}
&=&F_{J_{1}^{\ast }}(\Delta )F_{J_{2}^{\ast }}(\Delta )...F_{J_{m}^{\ast
}}(\Delta ) \\
&=&(\frac{\Delta -c_{1}}{d_{1}-c_{1}})^{N_{1}}(\frac{\Delta -c_{2}}{%
d_{2}-c_{2}})^{N_{2}}...(\frac{\Delta -c_{m}}{d_{m}-c_{m}})^{N_{m}},
\end{eqnarray*}%
where $F_{J_{i}^{\ast }}(\Delta )=1,$ for $\Delta \geq d_{i}$. After exploiting
Markov's Inequality,\vspace{-2mm} 
\begin{eqnarray*}
F_{Q}(Q &\geq &\Delta )\leq \frac{E[Q]}{\Delta }. \\
E[Q] &\geq &\Delta F_{Q}(Q\geq \Delta ) \\
&=&\Delta (1-(\frac{\Delta -c_{1}}{d_{1}-c_{1}})^{N_{1}}(\frac{\Delta -c_{2}%
}{d_{2}-c_{2}})^{N_{2}}...(\frac{\Delta -c_{m}}{d_{m}-c_{m}})^{N_{m}}) \\
&\geq &\Delta (1-(\frac{\Delta -c_{b}}{d_{b}-c_{b}})^{N_{b}}).
\end{eqnarray*}%
We normalize the lower bound and upper bound with $%
c_{b}=0$ and $d_{b}=1$. Let $\Delta $ be the top-$\rho $ percentile solution
value, i.e. $\Delta =c_{b}+(1-\rho )(d_{b}-c_{b}).$ Therefore,
\begin{equation*}
E[\widetilde{Q}]\geq (1-\rho )(1-(1-\rho )^{N_{b}}).
\end{equation*}%
To find the maximum $(1-\rho )(1-(1-\rho )^{N_{b}})$, we let\vspace{-1mm} 
\begin{equation*}
\frac{\partial (1-\rho )(1-(1-\rho )^{^{N_{b}}})}{\partial \rho }=0.
\end{equation*}%
The maximum $(1-\rho )(1-(1-\rho )^{N_{b}})$ is acquired when $\rho $ is $%
1-(N_{b}+1)^{-\frac{1}{N_{b}}}$. Therefore, 
\begin{equation*}
E[\widetilde{Q}]\geq N_{b}(\frac{1}{N_{b}+1})^{\frac{N_{b}+1}{N_{b}}}.
\end{equation*}%
Since $\widetilde{Q}$ is a lower bound of $\frac{Q}{Q^{\ast }}$,
\begin{equation*}
E[Q]\geq N_{b}(\frac{1}{N_{b}+1})^{\frac{N_{b}+1}{N_{b}}}\cdot Q^{\ast }.
\end{equation*}%
If the computational budget allocation is $r-$stages with $T\geq $\\ $mr\frac{\ln (m-1)}{\ln (\frac{1}{\alpha })}$, $N_{b}$\emph{\ }is $\frac{T}{r}/m+\frac{1}{2}\frac{r-1}{2r}T$, which
is \\$\frac{4+m(r-1)}{4rm}T$.
\end{proof}

\textbf{Time Complexity of CBAS.} The time complexity of \emph{CBAS }contains two
parts. The first phase selects $m$ start nodes with $O( E+n+$ $m\log
n)$ time, where $O(E)$ is to sum up the interest and social tightness
scores, $O(n+m\log n)$ is to build a heap and extract $m$ nodes with the
largest sum. Afterward, the second phase of \emph{CBAS }includes $r$ stages,
and each stage allocates the computation resources with $O(m)$ time and
generates $O(\frac{T}{r})$ new partial solutions with $k$ nodes for all start nodes.
Therefore, the time complexity of the second phase is $O\left( r(m+\frac{T}{r%
}k)\right) =O(kT)$, and \emph{CBAS }therefore needs $O(E+m\log n+kT)$ running
time.

\section{Neighbor Differentiation in \newline
Randomization}

\label{CE_method}

\subsection{Greedy Neighbor Differentiation}

In Section \ref{CBAAlgorithmdescription}, \emph{CBAS} includes two phases.
The first phase initiates the start nodes, while the second phase allocates
different computational budgets to each start node to generate different
numbers of final solutions. During the growth of a partial solution, \emph{%
CBAS} chooses a neighboring node uniformly at random at each iteration. In
other words, each neighboring node of the partial solution is treated
equally. It is expected that this homogeneous strategy needs more
computational budgets, because a neighboring node inclined to generate a
final solution with high willingness is not associated with a higher
probability.

To remedy this issue, a simple algorithm \emph{RGreedy} (randomized greedy)%
\textit{\ }associates each neighboring node with a different
probability according to its interest score and social tightness scores of
the edges incident to the partial solution $S_{t-1}$ obtained in the
previous stage, which is similar to the concept in the greedy algorithm. Given $%
S_{t-1}$, the ratio of the probabilities that \emph{RGreedy} selects nodes $%
v_{i}$ and $v_{j}$ at iteration $t$ is \vspace{-2mm} 
\begin{equation*}
\frac{P(v_{i}|S_{t-1})}{P(v_{j}|S_{t-1})}=\frac{W(\{v_{i}\}\cup S_{t-1})}{%
W(\{v_{j}\}\cup S_{t-1})},  \label{Eq:COND_ratio}
\end{equation*}%
where $W(\{v_{i}\}\cup S_{t-1})$ denotes the willingness of the node set $%
\{\{v_{i}\}\cup S_{t-1}\}$. At each iteration, \emph{RGreedy}
randomly selects a vertex in accordance with $W(\{v_{j}\}\cup S_{t-1})$,
until $k$ nodes are included.

Intuitively, \emph{RGreedy} can be regarded as a \textit{randomized version}
of the greedy algorithm with $m$ start nodes, while the greedy algorithm is
a deterministic algorithm with only one start node. Thus, similar to the
greedy algorithm, the assignment of the probability limits the scope to only
the local information associated with each node and incident edges. It is
envisaged that \textit{RGreedy} is difficult to generate a final solution
with high willingness, which is also demonstrated in Section \ref{Exp}. In
contrast, we propose \emph{CBAS-ND} by exploiting the cross entropy method
according to the sampling partial solutions in previous stages, in order to
optimally assign a probability to each neighboring node.

\subsection{Neighbor Differentiation with Cross Entropy}

\label{CBACEAlgorithmdescription}

We enhance \emph{CBAS} to \emph{CBAS-ND} to differentiate the selection of a
node neighboring each partial solution. Algorithm \emph{CBAS} is divided
into $r$ stages. In each stage, it optimally adjusts the computational
budgets allocated to each start node according to the sampled maximum and
minimum willingness in previous stages. To effectively improve \emph{CBAS}, 
\emph{CBAS-ND} takes advantage of the cross entropy method \cite{RubinsteinCE01} to
achieve importance sampling by adaptively assigning a different probability to each neighboring node from the
sampled results in previous stages. In contrast to \emph{RGreedy} with a
greedy-based probability vector assigned to the neighboring nodes, it is
expected that \emph{CBAS-ND} is able to obtain final solutions with better quality. Indeed, later in Section \ref{CBACE_Thm}, we
prove that the solution quality of \emph{CBAS-ND} is better than \emph{CBAS} with the same computational budget.

The flowchart of \emph{CBAS-ND} is shown in Figure \ref{CBACE_flowchart}. We
first define the node selection probability vector in \emph{CBAS-ND}, which
specifies the probability to add a node in $G$ to the current partial
solution expanded from a start node.

\begin{definition}
Let $\overrightarrow{p}_{i,t}$ denote the node selection probability vector
for start node $v_{i}$ in stage $t$.\newline
\centerline{$\overrightarrow{p}_{i,t}$= $\langle
p_{i,t,1}$,...,$p_{i,t,j}$,..., $p_{i,t,n} \rangle$,}\newline
where $p_{i,t,j}$ is the probability of selecting node $v_{j}$ for start
node $v_{i}$ in the $t$-th stage.
\end{definition}

In the first stage, the node selection probability vector $\overrightarrow{p}%
_{i,1}$ for each start node $v_{i}$ is initialized homogeneously for every
node, i.e. $\overrightarrow{p}_{i,1,j}=(k-1)/|V|$, $\forall v_{j}\in G$, $%
v_{j}\neq v_{i}$. That is, computational budgets $\frac{T_{1}}{m}$
are identically assigned to each start node, and the probability associated
with every node is also the same. However, different from \emph{CBAS} and 
\emph{RGreedy}, \emph{CBAS-ND} here examines the top-$\rho $ samples for
each start node $v_{i}$ to generate $\overrightarrow{p}_{i,2}$, so that
the node probability will be differentiated according to sampled result in stage $1$.

\begin{figure}[t]
\centering
\includegraphics[height=1.6036in, width=2.7809in]{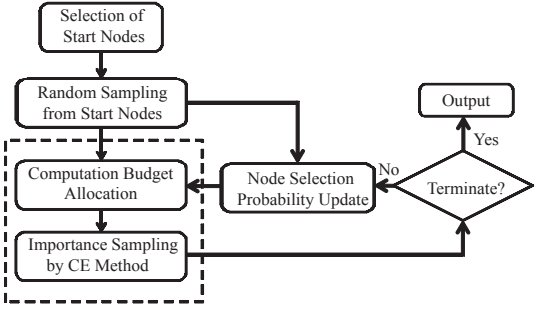} 
\caption{Flowchart of CBAS-ND}
\vspace{-18pt}
\label{CBACE_flowchart}
\end{figure}

\begin{definition}
A Bernoulli sample vector, denoted as $X_{i,q}=\langle x_{i,q,1},...,$ $%
x_{i,q,j},...,x_{i,q,n}\rangle$, is defined to be the $q$-th sample vector
from start node $v_{i}$, where $x_{i,q,j}$ is $1$ if node $v_{j}$ is
selected in the $q$-th sample and $0$ otherwise.
\end{definition}

\vspace{-2mm}

\begin{definition}
$\gamma_{i,t}$ is denoted as the top-$\rho $ sample quantile of the
performances in the $t$-th stage of start node $v_{i}$, i.e., $\gamma_{i,t}$=%
\newline
$W_{(\left\lceil \rho N_{i,t}\right\rceil )}$.
\end{definition}

Specifically, after collecting $N_{i,1}$ samples $X_{i,1},X_{i,2},...,$ $%
X_{i,q},...,$ $X_{i,N_{i,1}}$ generated from $\overrightarrow{p}_{i,1}$ for
start node $v_{i}$, Node Selection Probability Update\textbf{\ }in Figure %
\ref{CBACE_flowchart} calculates the total willingness \newline
$W(X_{i,q})$ for each sample, and sorts them in the descending order, $%
W_{(1)}\geq ...\geq W_{(N_{i,1})}$, while $\gamma _{i,1}$ denotes the
willingness of the top-$\rho $ performance sample, i.e. $\gamma
_{i,1}=W_{(\left\lceil \rho N_{i,1}\right\rceil )}$ . With those sampled
results, the selection probability $p_{i,2,j}$ of every node $v_{j}$ in the
second stage is derived according to the following equation, \vspace{-1.5mm} 
\begin{equation}
p_{i,t+1,j}=\frac{\sum_{q=1}^{N_{i,t}}I_{\{W(X_{i,q})\geq \gamma
_{i,t}\}}x_{i,q,j}}{\sum_{q=1}^{N_{i,r}}I_{\{W(X_{i,q})\geq \gamma _{i,t}\}}}%
,
\label{Eq:ce-opt}
\end{equation}%
where the indicator function $I_{\{W(X_{i,q})\geq \gamma
_{i,t}\}}$ is defined on the feasible solution space $\chi $ such that $%
I_{\{W(X_{i,q})\geq \gamma _{i,t}\}}$ is $1$ if the willingness of sample $%
X_{i,q}$ exceeds a threshold $\gamma _{i,t}$ $\in $ $\mathbb{R}$, and $0$
otherwise. Eq. (\ref{Eq:ce-opt}) derives the node selection probability
vector by fitting the distribution of top-$\rho $ performance samples.
Intuitively, if node $v_{j}$ is included in most top-$\rho $ performance
samples in $t$-th stage, $p_{i,t+1,j}$ will approach 1 and be selected in $(t+1)$-th stage.

Later in Section \ref{CBACE_Thm}, we prove that the above probability
assignment scheme is optimal from the perspective of cross entropy. Eq.
(\ref{Eq:ce-opt})\textbf{\ }minimizes the Kullback-Leibler cross entropy
(KL) distance \cite{RubinsteinCE01} between node selection probability $%
\overrightarrow{p}_{i,t}$ and the distribution of top-$\rho $ performance
samples, such that the performance of random samples in $t+1$ is guaranteed
to be closest to the top-$\rho $ performance samples in $t$. Therefore, by
picking the top-$\rho $\ performance samples to generate the partial
solutions in the next stage, the performance of random samples is expected
to be improved after multiple stages. Most importantly, by minimizing the KL
distance, the convergence rate is maximized.

Moreover, it is worth noting that a smoothing technique is necessary to be
included in adjusting the selection probability vector, \vspace{-1.5mm} 
\begin{equation*}
\overrightarrow{p}_{i,t+1}=w\overrightarrow{p}_{i,t+1}+(1-w)\overrightarrow{p%
}_{i,t},
\end{equation*}%
to avoid setting $0$ or $1$ in the selection probability for any node $v_{j}$%
, because $v_{j}$ will no longer appear or always appear in this case. An
example illustrating \emph{CBAS-ND} is provided as follows. As
demonstrated in Section \ref{CBACE_Thm}, the solution quality of \emph{CBAS-ND} is better than \emph{CBAS} with the same computation budget.

\begin{figure}[t]
\centering
\includegraphics[height=1.2182in, width=3.009in]{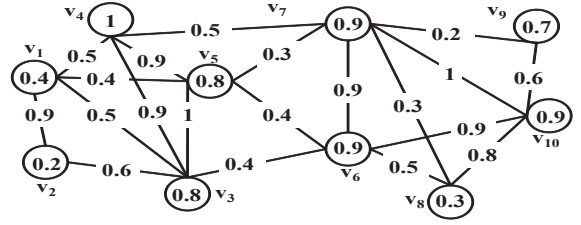} \vspace{-2pt}
\caption{An illustrative example for CBAS and CBAS-ND}
\vspace{-12pt}
\label{IllustratingExample}
\end{figure}

\begin{example}
\vspace{-2mm} Take Figure \ref{IllustratingExample} as an illustrating example of \emph{CBAS-ND}. Since \emph{CBAS-ND} is different from \emph{CBAS} in the second phase to
obtain the node selection probability vector, we continue from the result of the first phase in Section 3, i.e., the allocated computational budgets for start node $v_{3}$ and $v_{10}$ are $7$ and $3$ respectively, and illustrate the second phase of \emph{CBAS-ND} with Figure \ref{IllustratingExample}. 

By sorting the willingness samples $X_3$ to $W=\langle9.2,8.9,8.9,7.9$,\ $5.9\rangle$, $\gamma _{1}$\ is equal to $W_{(\left\lceil \frac{1}{2}%
5\right\rceil )}=8.9.$\ Therefore, the samples with the total willingness
exceeding $8.9$ include $X_{1}$, $X_{2}$, and $X_{5}$, which are
used to update\ the node selection probability $\overrightarrow{p}_{3,2}$\
to $\langle\frac{2}{3},\frac{1}{3},1,\frac{2}{3}$,\ $1,\frac{2}{3},\frac{1}{3},0,0,0\rangle$. Then, the smoothing technique is adopted with $w=0.6$, and the node
selection probability $\overrightarrow{p}_{3,2}$\ becomes \vspace{-4pt} 
\begin{eqnarray*}
\widehat{p\text{ }}_{3,2} &=&0.6\langle\frac{2}{3}\emph{,}\frac{1}{3},1,%
\frac{2}{3},1,\frac{2}{3},\frac{1}{3},0,0,0\rangle \\
&&+0.4\langle\frac{4}{9}\text{\emph{,}}\frac{4}{9}\text{,}1\text{,}\frac{4}{9%
}\text{,}\frac{4}{9}\text{,}\frac{4}{9}\text{,}\frac{4}{9}\text{,}\frac{4}{9}%
\text{,}\frac{4}{9}\text{,}\frac{4}{9}\rangle \\
&=&\langle\frac{5.2}{9}\emph{,}\frac{3.4}{9},1,\frac{5.2}{9},\frac{7}{9},%
\frac{5.2}{9},\frac{3.4}{9},\frac{1.6}{9},\frac{1.6}{9},\frac{1.6}{9}\rangle.
\end{eqnarray*}%
After sampling from node $v_{3}$, we repeat the above process for start node 
$v_{10}$. The worst result is $c_{10}=6.9$, the best result is $d_{10}=8.9$, and
the node selection probability is $\overrightarrow{p}_{10,2}=\langle\frac{1.6}{9%
},\frac{1.6}{9},\frac{1.6}{9},\frac{3.4}{9},$\ $\frac{5.2}{9},\frac{5.2}{9},%
\frac{7}{9},\frac{5.2}{9},\frac{5.2}{9},1\rangle.$ At the second stage, the best results of $v_{3}$ and $v_{10}$\ are $9.7$ and $8.9$, respectively. Finally, we obtain the solution $\{v_{3},v_{4},v_{5},v_{6},v_{7}\}$ with the total willingness $%
9.7 $, which is also the optimal solution in this example and outperforms the solution obtained from \emph{CBAS}.
\end{example}
\vspace{-2mm}

\subsection{Theoretical Result of CBAS-ND}

\label{CBACE_Thm}
In the following, we prove that the probability assignment with
the cross-entropy method \cite{RubinsteinCE01} in Eq. (\ref{Eq:ce-opt}) is optimal. The idea of
cross-entropy method originates from importance sampling\footnote{%
Importance sampling \cite{RubinsteinCE01} is used to estimate the properties
of a target distribution by using the observations from a different
distribution. By changing the distribution, the "important" values can be
effectively extracted and emphasized by sampling more frequently to reduce
the sample variance.}, i.e., by changing the distribution of sampling on
different neighbors such that the neighbors having the potential to
boost the willingness are able to be identified and included. Therefore, we
first derive the probability of a random sample according to the sampling
results in previous stages. After this, we introduce importance sampling and
derive the node selection probability vector in the WASO problem to
replace the original sampling vector such that the Kullback-Leibler cross
entropy (KL) distance between the sampling vector and the optimal importance
sampling vector is minimized. Intuitively, a small KL distance
ensures that two distributions are very close and implies that the node selection probability vector is optimal because the KL distance between the node selection probability vector in \emph{CBAS-ND} and optimal node selection probability vector is minimized. Equipped with importance sampling vector, later in this
section we prove that the solution quality of \emph{CBAS-ND }is better than 
\emph{CBAS}.

More specifically, let $\chi $ denote the feasible solution space, and $X$
is a feasible solution in $\chi $, i.e., $X\in \chi $. WASO chooses a
group of attendee $X^{\ast }$ to find the maximum willingness $\gamma
^{\ast }$,\vspace{-1.5mm} 
\begin{equation*}
W(X^{\ast })=\gamma ^{\ast }=\max_{X\in \chi }W(X).
\end{equation*}%
To derive the probability that the willingness of a random sample $X$
exceeds a large value $\gamma $, i.e. $W(X)\geq \gamma $, it is necessary
for \emph{CBAS} to generate many samples given that it uniformly selects a neighboring
node at random. In contrast, \emph{CBAS-ND} leverages the notion
of importance sampling to change the distribution of sampling on different
neighbors. In the following, we first derive the optimal distribution of
sampling. First, for the initial partial solution with one start node, let $%
f(X;\overrightarrow{p})$ denote the probability density function of
generating a sample $X$ according a real-valued vector $\overrightarrow{p}$,
and $f(\cdot ;\overrightarrow{p})$ is a family of probability density
functions on $\chi $, i.e., \vspace{-1.2mm} 
\begin{equation*}
f(\cdot ;\overrightarrow{p})=\{f(X;\overrightarrow{p})|X\in \chi \}.
\end{equation*}%
\emph{CBAS} can be regarded as a special case of \emph{CBAS-ND} with the
homogeneous assignment on the above vector. A random sample $X(%
\overrightarrow{p})$ for $\overrightarrow{p}=\{p_{1}$,...,$p_{j}$,...,$%
p_{n}\}$ is generated with probability $f(X(\overrightarrow{p});%
\overrightarrow{p})$, where $p_{j}$ denotes the probability of selecting
node $v_{j}$ and is the same for all $j$ in \emph{CBAS}. The probability $%
P_{_{\overrightarrow{p}}}(\gamma )$ that the willingness of $X(%
\overrightarrow{p})$ exceeds the threshold $\gamma $ is \vspace{-0.5mm} 
\begin{equation*}
P_{\overrightarrow{p}}(\gamma )=\mathbb{P}_{\overrightarrow{p}}(W(X(%
\overrightarrow{p}))\geq \gamma )
\end{equation*}%
\vspace{-5mm} 
\begin{equation*}
\hspace{80.5pt}=\sum\limits_{X\in \chi }I_{\{W(X(\overrightarrow{p}))\geq
\gamma \}}f(X(\overrightarrow{p});\overrightarrow{p}).
\end{equation*}

However, the above equation is impractical and inefficient for a large
solution space, because it is necessary to scan the whole solution space $%
\chi $ and sum up the probability $f(X(\overrightarrow{p});\overrightarrow{p}%
)$ of every sample $X$ with $W(X(\overrightarrow{p}))\geq \gamma $. To more efficiently address this issue, a direct way to derive the estimator $%
\widehat{P}_{\overrightarrow{p}}(\gamma )$ of $P_{\overrightarrow{p}}(\gamma )$ is by employing a crude Monte-Carlo simulation and drawing $N$ random
samples $X_{1}(\overrightarrow{p})$,..., $X_{N}(\overrightarrow{p})$ by $%
f(\cdot ,\overrightarrow{p})$ to find $P_{\overrightarrow{p}}(\gamma )$%
, \vspace*{-2mm} 
\begin{equation*}
\widehat{P}_{\overrightarrow{p}}(\gamma )=\frac{1}{N}%
\sum\limits_{i=1}^{N}I_{\{W(X_{i}(\overrightarrow{p}))\geq \gamma \}}.
\end{equation*}%
However, the crude Monte-Carlo simulation poses a serious problem when $%
\{W(X(\overrightarrow{p}))\geq \gamma \}$ is a rare event since rare events
are difficult to be sampled, and thus a large sample number $N$ is
necessary to estimate $P_{\overrightarrow{p}}(\gamma )$ correctly.

Based on the above observations, \emph{CBAS-ND} attempts to find the distribution $%
f(X(\overrightarrow{p});\overrightarrow{p})$ based on another importance
sampling pdf $f(X(\overrightarrow{p_{g}});\overrightarrow{p_{g}})$ to reduce the required
sample number. For instance, consider a network with $3$ nodes, i.e. $%
V=\{v_{1},v_{2},v_{3}\}$, and the 2-node group where the maximum willingness $%
\gamma ^{\ast }$ is $\{v_{1},v_{2}\}$. The expected number of samples with
node selection vector $\{\frac{2}{3},\frac{2}{3},\frac{2}{3}\}$ in \emph{CBAS} is larger than the node selection vector of $\{1,1,0\}$ in \emph{CBAS-ND}. In finer detail, let $X_{i}(\overrightarrow{p_{g}})$ denote
the $i$-th random sample generated by $f(X(\overrightarrow{p_{g}});%
\overrightarrow{p_{g}})$. \emph{CBAS-ND} first creates random samples $%
X_{1}(\overrightarrow{p_{g}})$,..., $X_{N}(\overrightarrow{p_{g}})$
generated by $\overrightarrow{p_{g}}$ on $\chi $ and then estimates $%
\widehat{P}_{\overrightarrow{p}}(\gamma )$ according to the likelihood ratio
(LR) estimator $\frac{f(X_{i}(\overrightarrow{p_{g}});\overrightarrow{p})}{%
f(X_{i}(\overrightarrow{p_{g}});\overrightarrow{p_{g}})}$, \vspace{-2.7mm} 
\begin{equation*}
\widehat{P}_{\overrightarrow{p}}(\gamma )=\frac{1}{N}\sum%
\limits_{i=1}^{N}I_{\{W(X_{i}(\overrightarrow{p}))\geq \gamma \}}
\end{equation*}%
\vspace{-3.5mm} 
\begin{equation}
\text{ \ \ \ \ \ \ \ \ \ \ \ \ \ \ \ \ \ \ \ \ \ \ \ }=\frac{1}{N}%
\sum\limits_{i=1}^{N}\{I_{\{W(X_{i}(\overrightarrow{p_{g}}))\geq \gamma \}}%
\frac{f(X_{i}(\overrightarrow{p_{g}});\overrightarrow{p})}{f(X_{i}(%
\overrightarrow{p_{g}});\overrightarrow{p_{g}})}.  \label{ab}
\end{equation}%
Notice that the above equation holds when $N$ is infinity, but in most cases 
$N$ only needs to be sufficiently large in practical implementation \cite{CostaCEConvergence07}. Now the question becomes how to derive $\overrightarrow{p_{g}%
}$ for importance sampling pdf $f(X(\overrightarrow{p_{g}});\overrightarrow{%
p_{g}})$ to reduce the number of samples. The optimal importance sampling pdf $%
f^{\ast }(X_{i}(\overrightarrow{p_{g}});\overrightarrow{p_{g}})$ to
correctly estimate $P_{\overrightarrow{p}}(\gamma )$ thus becomes 
\begin{equation}
f^{\ast }(X_{i}(\overrightarrow{p_{g}});\overrightarrow{p_{g}})=\frac{%
I_{\{W(X_{i}(\overrightarrow{p_{g}}))\geq \gamma \}}f(X_{i}(\overrightarrow{%
p_{g}});\overrightarrow{p})}{P_{\overrightarrow{p}}(\gamma )}.
\label{Eq:G_OPT}
\end{equation}%
In other words, by substituting $f(X_{i}(\overrightarrow{p_{g}});%
\overrightarrow{p_{g}})$ with $f^{\ast }(X_{i}(\overrightarrow{p_{g}});%
\overrightarrow{p_{g}})$ in Eq. (\ref{ab}), $\widehat{P}_{%
\overrightarrow{p}}(\gamma )=\frac{1}{N}\sum\limits_{i=1}^{N}P_{%
\overrightarrow{p}}(\gamma )$ holds, implying that only $1$ sample is
required to estimate the correct $P_{\overrightarrow{p}}(\gamma )$, i.e., $N=1
$. However, it is difficult to find the optimal $f^{\ast }(X(\overrightarrow{%
p_{g}});\overrightarrow{p_{g}})$ since it depends on $P_{\overrightarrow{p}%
}(\gamma )$, which is unknown \emph{a priori} and is therefore not practical
for WASO.

Based on the above observations, \emph{CBAS-ND} optimally finds $\overrightarrow{p_{g}}$ and the importance sampling pdf $f(X(\overrightarrow{p_{g}});\overrightarrow{p_{g}})$ to minimize the Kullback-Leibler cross
entropy (KL) distance between $f(X(\overrightarrow{p_{g}});\overrightarrow{%
p_{g}})$ and optimal importance sampling pdf $f^{\ast }(X(\overrightarrow{%
p_{g}});\overrightarrow{p_{g}})$, where the KL distance measures two
densities $f^{\ast }$ and $f$ as\vspace{-1.5mm} 
\begin{equation}
D(f^{\ast },f)=\sum\limits_{X\in \chi }f^{\ast }(X)\ln f^{\ast
}(X)-\sum\limits_{X\in \chi }f^{\ast }(X)\ln f(X).
\label{KLdistFunc}
\end{equation}%
The first term in the above equation is related to $f^{\ast }$ and is fixed, and
minimizing $D(f^{\ast },f)$ is equivalent to maximizing the second term,
i.e., $\sum_{X\in \chi }f^{\ast }(X)\ln f(X)$. It is worth noting that the
importance sampling pdf $f(X(\overrightarrow{p_{g}});\overrightarrow{p_{g}})$
is referenced to a vector $\overrightarrow{p_{g}}$. Thus, after substituting 
$f^{\ast }(X_{i}(\overrightarrow{p_{g}});\overrightarrow{p_{g}})$ in
Eq. (\ref{Eq:G_OPT}) into the Eq. (\ref{KLdistFunc}), the reference
vector $\overrightarrow{p_{g}}$ of importance sampling pdf $f(X(%
\overrightarrow{p_{g}});\overrightarrow{p_{g}})$ that maximizes the second
term of Eq. (\ref{KLdistFunc}) is the optimal reference vector $%
\overrightarrow{p_{g}}^{\ast }$ with the minimum KL distance\textbf{,}%
\vspace{-1mm} 
\begin{equation}
\overrightarrow{p_{g}}^{\ast }=\arg \max_{\overrightarrow{p_{g}}%
}\sum\limits_{X\in \chi }\frac{I_{\{W(X(\overrightarrow{p_{g}}))\geq \gamma
\}}f(X(\overrightarrow{p_{g}});\overrightarrow{p})}{P_{\overrightarrow{p}%
}(\gamma )}\ln f(X(\overrightarrow{p_{g}});\overrightarrow{p_{g}}).
\label{Eq:CE_FINALEQ}
\end{equation}%
Since $P_{\overrightarrow{p}}(\gamma )$ is not related to $\overrightarrow{p_{g}}$. Eq. (\ref{Eq:CE_FINALEQ}) is equivalent to 
\begin{equation*}
\arg \max_{\overrightarrow{p_{g}}}\mathbb{E}_{\overrightarrow{p_{g}}}
{I_{\{W(X(\overrightarrow{p_{g}}))\geq \gamma\}}\ln f(X(\overrightarrow{p_{g}});\overrightarrow{p_{g}})},
\label{Eq:CE_FINALEQ2}
\end{equation*}

Because it is computationally intensive to generate and compare every
feasible $\overrightarrow{p_{g}}$, we estimate $\mathbb{E}_{\overrightarrow{p_{g}}}
{I_{\{W(X(\overrightarrow{p_{g}}))\geq \gamma\}}}$ by drawing $N$ samples as \vspace{-3mm} 
\begin{equation*}
\arg \max_{\overrightarrow{p_{g}}}\frac{1}{N}\sum\limits_{i=1}^{N}I_{%
\{W(X_{i}(\overrightarrow{p_{g}}))\geq \gamma \}}\ln f(X_{i}(\overrightarrow{%
p_{g}});\overrightarrow{p_{g}}).
\end{equation*}%
Specifically, \emph{CBAS-ND} first generates random samples $%
X_{1}$,...,$X_{i}$,...,\newline
$X_{N},$ where $X_{i}$ is the $i$-th sample and is a Bernoulli vector
generated by a node selection probability vector $\overrightarrow{p_{g}}$, i.e., $%
X_{i}=(x_{i,1},...,$ $x_{i,n})$ $\sim$ $Ber(\overrightarrow{p_{g}})$, where $%
\overrightarrow{p_{g}}=\{p_{1}$,...,$p_{j}$,...,$p_{n}\}$ and $p_{j}$ denotes
the probability of selecting node $v_{j}$. Consequently, the pdf $f(X_{i}(\overrightarrow{p_{g}});\overrightarrow{p_{g}})$ is \vspace{-2mm}
\begin{equation*}
f(X_{i}(\overrightarrow{p_{g}});\overrightarrow{p_{g}})=\prod%
\limits_{j=1}^{N}p_{j}^{x_{i,j}}(1-p_{j})^{1-x_{i,j}}.
\end{equation*}%
\vspace{-0.2mm} To find the optimal reference vector $\overrightarrow{p}%
^{\ast }$ with Eq. (\ref{Eq:CE_FINALEQ}), we first calculate the first
derivative w.r.t. $p_{j}$, \vspace{-2mm} 
\begin{equation}
\frac{\partial }{\partial p_{j}}\ln f(X_{i}(\overrightarrow{p_{g}});%
\overrightarrow{p_{g}})=\frac{\partial }{\partial p_{j}}\ln
p_{j}^{x_{i,j}}(1-p_{j})^{1-x_{i,j}}.  \label{Eq:firstderivativeCE}
\end{equation}%
\vspace{-2mm} Since\ $x_{i,j}$ can be either 0 or 1, Eq. (\ref%
{Eq:firstderivativeCE}) is simplified to\vspace{-0.5mm} 
\begin{equation*}
\frac{\partial }{\partial p_{j}}\ln f(X_{i}(\overrightarrow{p_{g}});%
\overrightarrow{p_{g}})=\frac{1}{(1-p_{j})p_{j}}(x_{i,j}-p_{j}).
\end{equation*}%
The optimal reference vector $\overrightarrow{p}^{\ast }$ is obtained by
setting the first derivative of Eq. (\ref{Eq:CE_FINALEQ}) to zero. 
\vspace{-2.5mm} 
\begin{eqnarray*}
&&\frac{\partial }{\partial p_{j}}\sum\limits_{i=1}^{N}I_{\{W(X_{i,j})\geq
\gamma \}}\ln f(X_{i}(\overrightarrow{p_{g}});\overrightarrow{p_{g}}) \\
&=&\frac{1}{(1-p_{j})p_{j}}\sum\limits_{i=1}^{N}I_{\{W(X_{i})\geq \gamma
\}}(x_{i,j}-p_{j})=0.
\end{eqnarray*}%
Finally, the optimal $p_{j}$ assigned to each node $v_{j}$ is \vspace{-0.5mm}
\begin{equation*}
p_{j}=\frac{\sum_{i=1}^{N}I_{\{W(X_{i})\geq \gamma \}}x_{i,j}}{%
\sum_{i=1}^{N}I_{\{W(X_{i})\geq \gamma \}}}.
\end{equation*}

\begin{theorem}
\label{NE_Thm} The solution quality of CBAS-ND
is better than CBAS under the same computation budget $T$.
\end{theorem}

\vspace{-2mm}
\begin{proof}
Let $\overrightarrow{v}_{t}$ be the node selection vector in the $t$-th
stage, where $\overrightarrow{v}_{t}=\{\overrightarrow{v}_{t,1},%
\overrightarrow{v}_{t,2},...\overrightarrow{v}_{t,n}\}$. We first define the
random variables $\phi _{v}^{_{t,i}}=v_{t,i}I_{\{x_{i}^{\ast
}=1\}}+(1-v_{t,i})I_{\{x_{i}^{\ast }=0\}}$ for all $i=1,...,n$, where $%
I_{\{x_{i}^{\ast }=1\}}$ is the indicator function with 1 if node $v_{i}$
is in the optimal solution, and $0$ otherwise. Then, let $\phi _{v}^{t}$ denote the probability to generate the optimal solution in the $t$-th
stage,
\vspace{-3mm}
\begin{equation*}
\phi _{v}^{t}=f(X^{\ast };\overrightarrow{p}_{v})=\prod\limits_{i=1}^{n}%
\phi _{v}^{t,i}\text{.}
\end{equation*}%
Let $E_{r}$ be the event that does not sample the optimal solution in the
final $r$-th stage. From the previous work \cite{CostaCEConvergence07}, the
probability for the willingness to converge to the optimal solution can be formulated as%
\vspace{-2mm}
\begin{equation}
1-P(E_{r})\geq 1-P(E_{1})\exp (-\frac{N_{i}}{r}\phi
_{u}^{1}\sum_{t=1}^{r-1}w^{tn}),
\label{willeq}
\end{equation}%
where $w$ in Eq. (\ref{willeq}) is the smoothing technique parameter. Therefore, since \emph{CBAS} is identical to \emph{CBAS-ND} with $w=0$,
the convergence rate that \emph{CBAS-ND} samples the optimal solution is
larger than \emph{CBAS.} Therefore, to achieve the same solution quality, \emph{CBAS-ND} requires less computation budget than \emph{CBAS}. When \emph{CBAS} runs out of computation budget, i.e., $T$, the computation budget that \emph{CBAS-ND} achieves the same quality is less than $T$. Let $r_{ND}$ denote the number of stage that \emph{CBAS-ND} achieves the same quality. Since $r_{ND} \leq r$, we have 
\vspace{-2mm}
\begin{equation*}
1-P(E_{1})\exp (-\frac{N_{i}}{r_{ND}}\phi_{u}^{1}\sum_{t=1}^{r_{ND}-1}w^{tn})
\end{equation*}
\vspace{-2mm}
\begin{equation*}
\leq 1-P(E_{1})\exp (-\frac{N_{i}}{r}\phi _{u}^{1}\sum_{t=1}^{r-1}w^{tn})
\end{equation*}
Therefore, the solution quality of \emph{CBAS-ND} is better than \emph{CBAS}. The theorem follows.
\end{proof}\vspace{-2mm}

\textbf{Time Complexity of CBAS-ND.} 
\emph{CBAS-ND }is different from \emph{CBAS} in the second phase to
find the node selection probability vector, which needs $O(r(mn\rho \frac{T}{r}%
+\frac{T}{r}k))=O(mn\rho T)$. Therefore, the time complexity of \emph{CBAS-ND} is $O(E+m\log n+mn\rho T)$. However, in reality we can directly set the
probability to $0$ for every node not neighboring a partial solution of a
start node. Therefore, as shown in Section 5, the experimental result manifests that the execution time of \emph{CBAS-ND} is not far from \emph{CBAS}, and both \emph{CBAS} and \emph{CBAS-ND} are much faster than \emph{RGreedy}.

\subsection{Discussion}
\subsubsection{Online computation}
In the process of social activity planning, some candidate attendees may not accept the invitations, and an online algorithm to adjust the solution according to user responses can help us handle the dynamic situation. If the online decision of multiple attendees are dependent, the situation is similar to the entangled transactions \cite{Gupta11} in databases, in which it is necessary that transactions be processed coordinately in multiple entangled queries.
Therefore, we extend \emph{CBAS-ND} to cope with the dynamic situation as follows. If a user can not attend the activity, it is necessary to invite new attendees. Nevertheless, we have already sent invitations, and some of them have already confirmed to attend. Therefore, \emph{CBAS-ND} regards those confirmed attendees as the initial solution in the second phase and removes the nodes that can not attend the activity from $G$. Therefore, the node selection probability vector $\overrightarrow{p}_{i,t}$ will be updated to identify the new neighbors leading to better solutions according to the confirmed attendees. It is worth noting that the above online computation is fast since the start nodes in the first phase have been decided.
\subsubsection{Backtracking}
In addition to online computation, we extend \emph{CBAS-ND} for backtracking to further improve the solution quality as follows. As shown in the previous work \cite{CostaCEConvergence07,RubinsteinCE01}, the criterion of convergence for Cross-Entropy method is that the node selection probability vector does not change over a number of iterations. Motivated by the above work, given the node selection probability $\overrightarrow{p}_{i,t}$ of \emph{CBAS-ND} at each stage $t$, we derive the difference $z_i$ between $\overrightarrow{p}_{i,t}$ and $\overrightarrow{p}_{i,t-1}$ as follows.
\begin{equation*}
z_i=\sum_j^n{(\overrightarrow{p}_{i,t,j}-\overrightarrow{p}_{i,t-1,j})^2}.
\end{equation*}%
When the difference $z_i$ between $\overrightarrow{p}_{i,t}$ and $\overrightarrow{p}_{i,t-1}$ is lower than a given threshold $z_t$, which indicates that the solution quality converges, we backtrack the solution by resetting the node selection probability $\overrightarrow{p}_{i,t}$ to $\overrightarrow{p}_{i,t-1}$ and re-sample.

\subsubsection{CBAS-ND for Different Scenarios}
For the scenarios of \textbf{couple and foe}, \textbf{invitation}, and \textbf{exhibition}, \emph{CBAS-ND} can be directly applied by modifying the node and edge weights of the graph. For the scenario of separate groups, the start nodes are selected first, and the virtual node $\overline{v}$ is then added to the selection set $V_{S}$ to relax the connectivity constraint.

\section{Experimental Results}

\label{Exp}

In this section, we first present the results of user study and then evaluate the performance of the proposed algorithms with different parameter settings on real datasets. 
\vspace{-3pt}

\subsection{Experiment Setup}

\label{DataPreparation} We implement \emph{CBAS-ND} in Facebook and invite
137 people from various communities, e.g., schools, government, technology
companies, and businesses to join our user study, to compare the solution
quality and the time to answer WASO with manual coordination and \emph{CBAS-ND} for demonstrating the need of an automatic group recommendation service. Each user is asked to plan 10 social
activities with the social graphs extracted from their social networks in
Facebook. The interest scores follow the power-law distribution according to the recent analysis \cite{PowerLawClauset09} on real datasets, which has found the power exponent $%
\beta =2.5$. The social tightness score between two friends is derived according to the widely adopted model based on the number of common friends that represent the proximity interaction \cite{ChaojiNnodeLink12}. Then, social tightness scores and interest scores are normalized. Nevertheless, after the scores are returned by the above renowned models, each user is still allowed to
fine-tune the two scores by themselves. The 10 problems explore
various network sizes and different numbers of attendees in two different
scenarios. In the first 5 problems, the user needs to participate
the group activity and is inclined to choose her close friends, while the
following 5 problems allow the user to choose an arbitrary group of people
with high willingness. In other words, \textit{CBAS-ND }in the first 5
problems always chooses the user as a start node. In addition to the user study, three real datasets are tested in the experiment.
The first dataset is crawled from Facebook with $90,269$\ users in the New
Orleans network\footnote{%
http://socialnetworks.mpi-sws.org/data-wosn2009.html.}. The second dataset
is crawled from DBLP dataset with $511,163$ nodes and $1,871,070$ edges. The
third dataset, Flickr\footnote{%
http://socialnetworks.mpi-sws.org/data-imc2007.html.}, with $1,846,198$
nodes and $22,613,981$ edges, is also incorporated to demonstrate the
scalability of the proposed algorithms. 

In the following, we compare \emph{DGreedy}, \emph{RGreedy}, \emph{CBAS}, 
\emph{CBAS-ND}, and \emph{IP (Integer Programming) }solved by IBM CPLEX in
an HP DL580 server with four Intel E7-4870 2.4 GHz CPUs and 128 GB RAM. IBM
CPLEX is regarded as the fastest general-purpose parallel optimizer, and we
adopt it to solve the Integer Programming formulation for finding the
optimal solution to WASO\footnote{Note that because WASO is NP-Hard, it is only possible to find the optimal
solutions to WASO with IBM CPLEX in small cases.}. The details of Integer Programming formulation is presented in Appendix \ref{integerprogrammingformulation}.It is worth noting that even though \emph{RGreedy} performs much better than its counterpart \emph{DGreedy} and is closer to \emph{CBAS} and \emph{CBAS-ND}, it is computation intensive and not scalable to support a large group size. Therefore, we can only plot a few results of \emph{RGreedy} in some figures. The default $m$ is set to be $n/k$ since $n/k$ different $k$-person groups can be partitioned from a network with $n$. With $m$ equal $n/k$, the start nodes averagely cover the whole network. Nevertheless, the experimental analysis manifests that $m$ can be set to be smaller than $n/k$ in WASO since the way we select start nodes efficiently prunes the start nodes which do not generate good solutions. The computational budget of \emph{CBAS-ND} is not wasted much since the start node that do not generate good solutions will be pruned after the first stage. The default cross-entropy parameters $\rho $ and $w$ are $0.3$ and $0.9$ respectively, and $\alpha $ is $0.99$ as recommended by the cross-entropy method \cite{RubinsteinCE01}. The results with different
settings of parameters will be presented. Since \emph{CBAS}\ and \emph{%
CBAS-ND} natively support parallelization, we also implemented them with
OpenMP for parallelization, to demonstrate the gain in
parallelization with more CPU cores.

\subsection{User Study}

\label{UserStudy}

\begin{figure}[tp]
\centering
\subfigure[] {\
\includegraphics[scale=0.15]{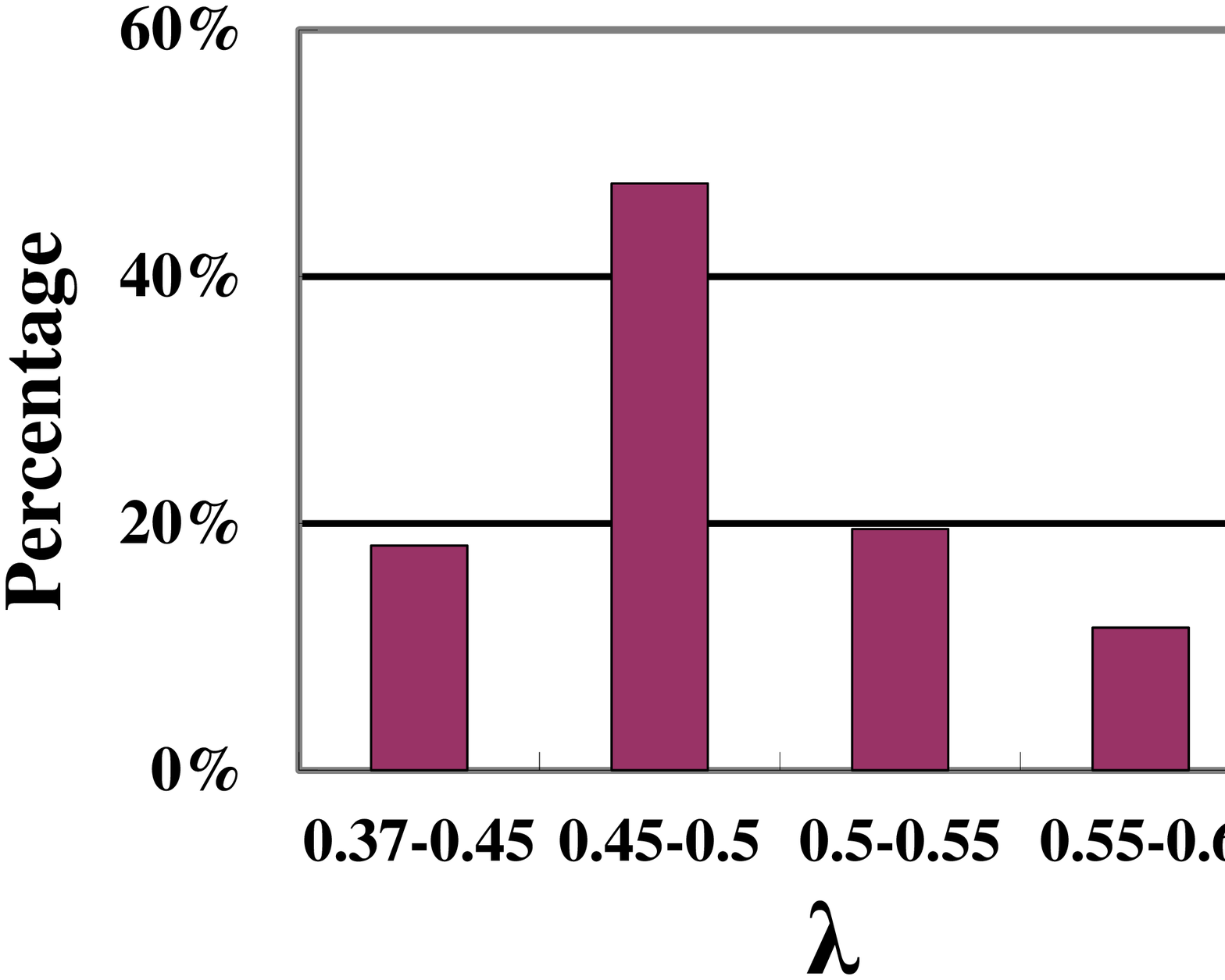} } \vspace{+25pt}
\subfigure[] {\
\includegraphics[scale=0.15]{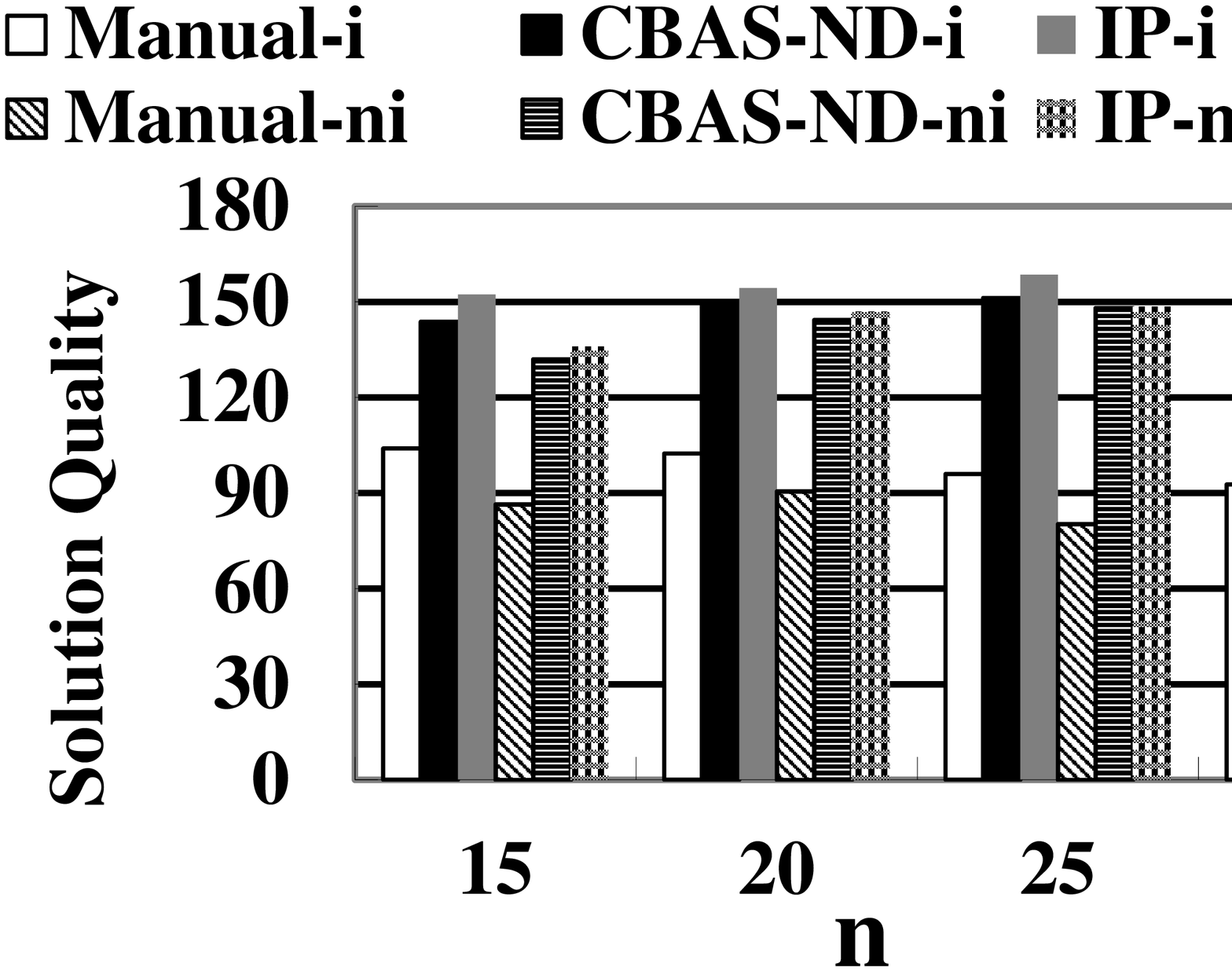} } 
\subfigure[] {\
\includegraphics[scale=0.15]{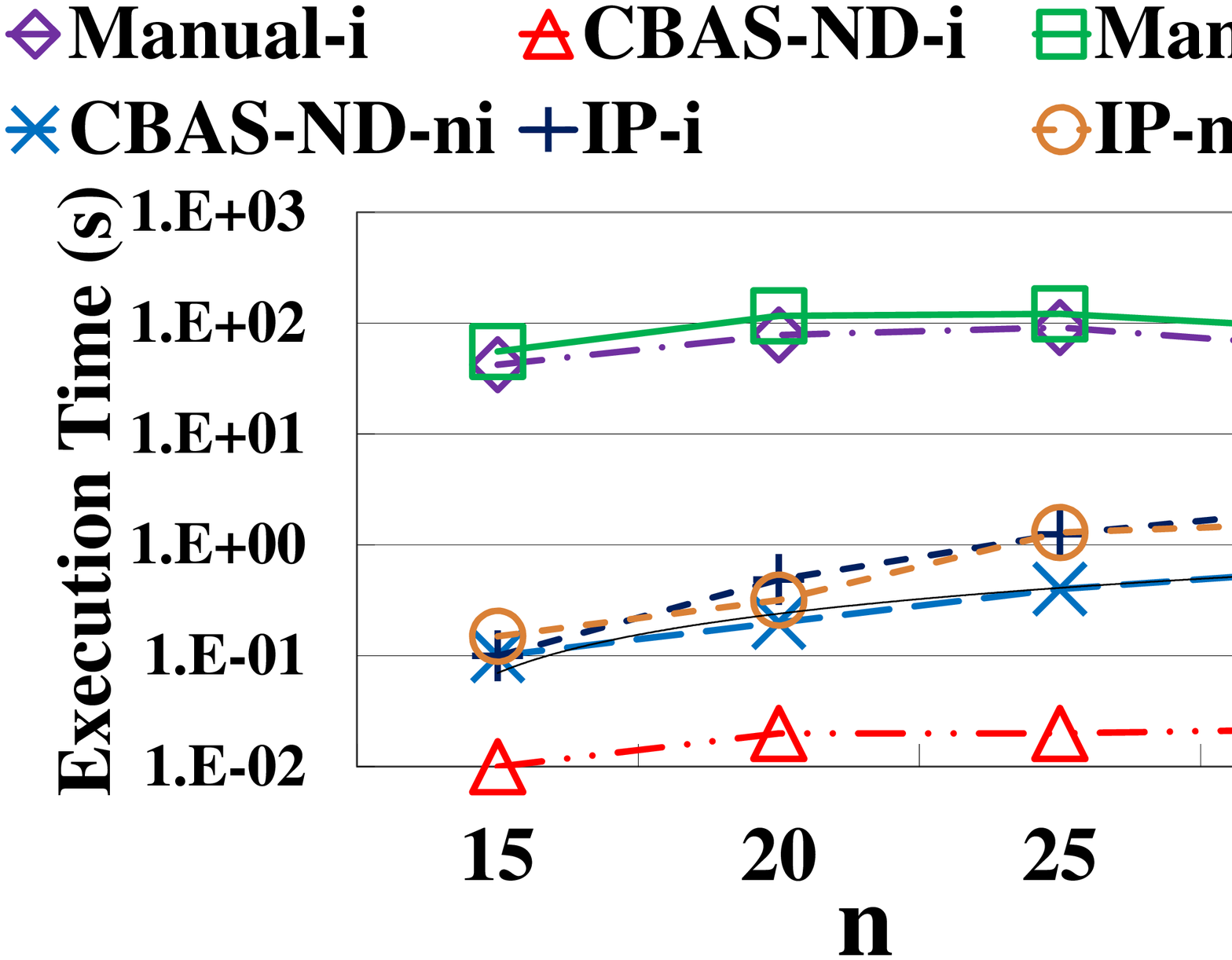} } \vspace{+25pt} 
\subfigure[] {\
\includegraphics[scale=0.15]{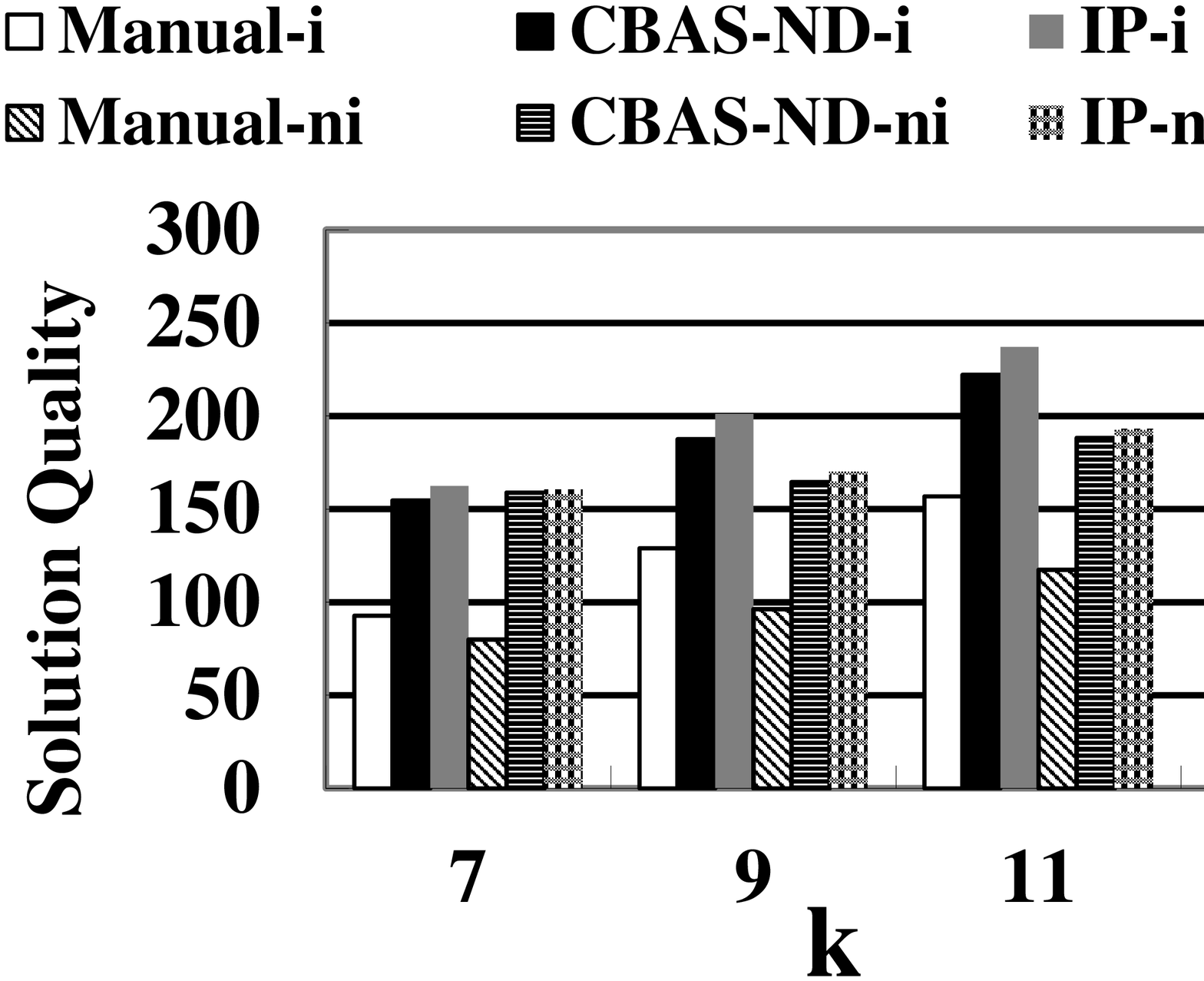} } 
\subfigure[] {\
\includegraphics[scale=0.15]{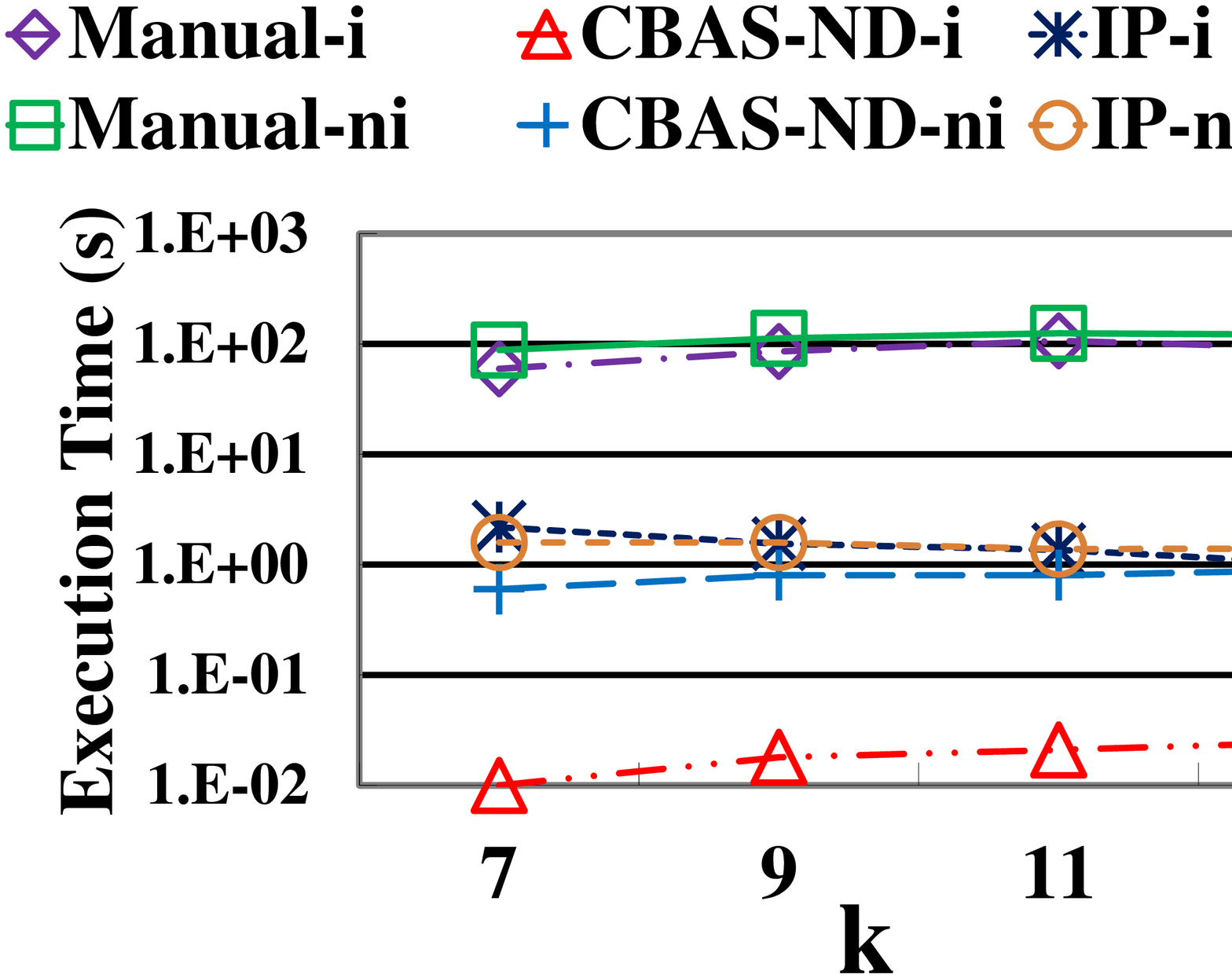} } \vspace{+25pt}
\subfigure[] {\
\includegraphics[scale=0.15]{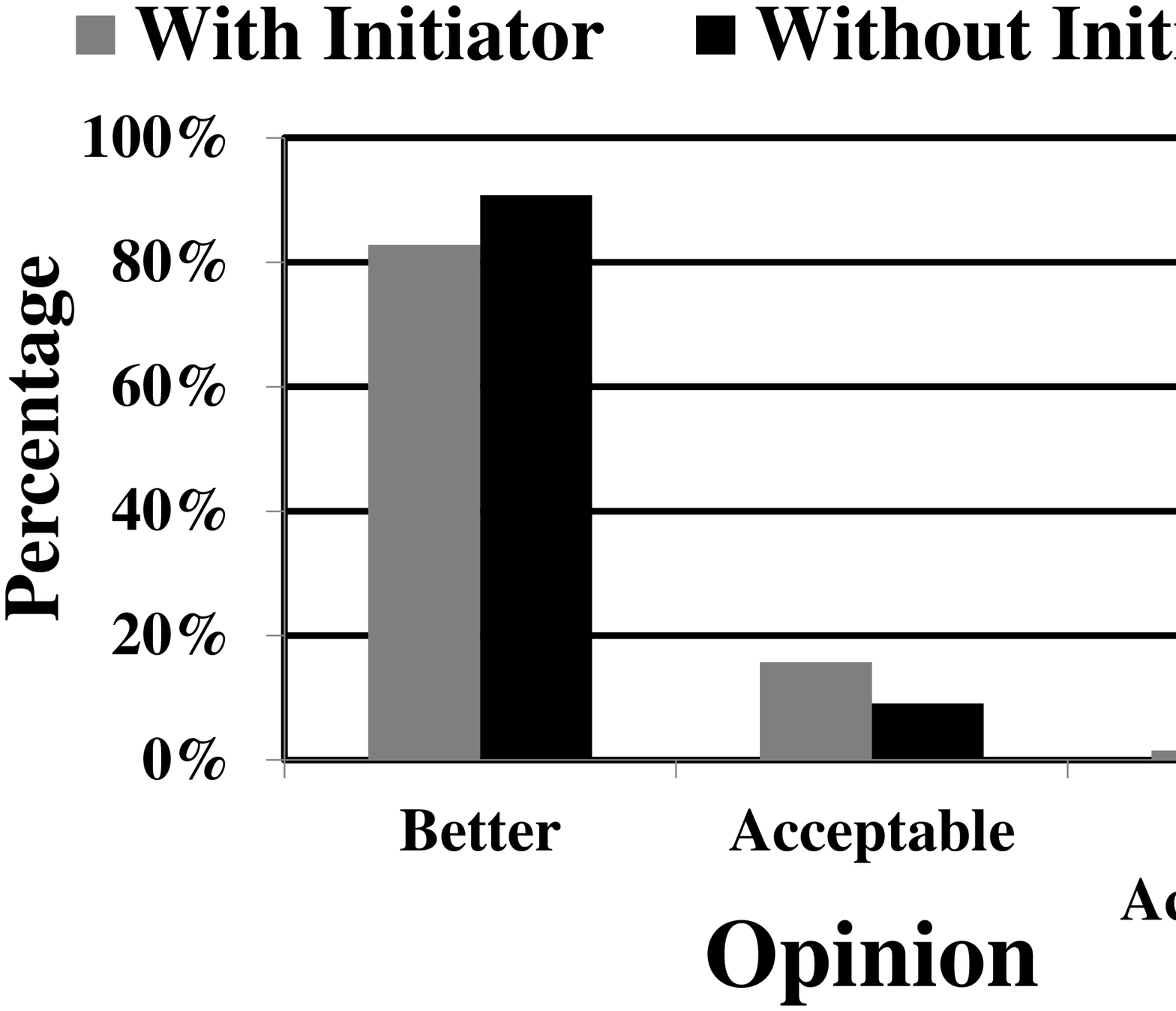} } \vspace{-35pt} 
\caption{Results of user study}
\label{FigUserStudy}
\vspace{-6mm}
\end{figure}

The weights $\lambda $ and (1-$\lambda $) in Section \ref{Prilim} for
interest scores and social tightness scores are directly specified by the users according to their preferences, and Figure \ref{FigUserStudy}(a) shows that the range of the weight
mostly spans from 0.37 to 0.66 with the average as 50.3, indicating that
both social tightness and interest are crucial factors in activity planning.
Figures \ref{FigUserStudy}(b)-(e) compare manual coordination and \emph{%
CBAS-ND} in the user study. It is worth noting that we generate the ground
truth of user study with \emph{IP} solved by IBM CPLEX to evaluate the
solution quality. Figures \ref{FigUserStudy}(b) and (c) present the solution
quality and running time with different network sizes, where the expected
number of attendees $k$ is $7$. The user must be included in the group for 
\textit{Manual-i} and \emph{CBAS-ND-i}, and in the other two cases the
user can arbitrarily choose a group with high willingness. The result
indicates that the solutions obtained by \emph{CBAS-ND }is very close to the
optimal solutions acquired from solving \emph{IP }with IBM CPLEX. WASO is
challenging for manual coordination, even when the network contains only
dozens of nodes. It is interesting that $n=30$ is too difficult for manual
coordination because some users start to give up thus require smaller time for finding a solution.
In addition, WASO is more difficult and more time-consuming in \textit{%
Manual-ni} because it considers many more candidate groups.

Figures \ref{FigUserStudy}(d) and (e) presents the results with different k. The results show that the
solution quality obtained by manual coordination with $k=7$ is only 66\% of \emph{%
CBAS-ND}, since it is challenging for a person to jointly maximize the
social tightness and interest.  Similarly, we discover that some users start
to give up when $k=13$, and the processing time of manual selection grows
when the user is not going to join the group activity. Finally, we return
the solutions obtained by \emph{CBAS-ND} to the users, and Figure \ref%
{FigUserStudy}(f) manifests that 98.5\% of users think the solutions are better
or acceptable, as compared to the solutions found by themselves. Therefore, it is desirable to deploy \emph{CBAS-ND }as an automatic group recommendation service, especially to address the need of a large group in a massive social network nowadays. \vspace{-5pt%
}

\subsection{Performance Comparison and Sensitivity Analysis}
\vspace{-8pt}
\subsubsection{Facebook}
\label{PerforAnalysis}Figure \ref{exp2}(a) first presents the running time
with different group sizes, i.e., $k$. \emph{RGreedy} is computationally
intensive since it is necessary to sum up the interest scores and social
tightness scores during the selection of a node neighboring each partial
solution. Therefore, \emph{RGreedy} is unable to return a solution within even $12$ hours when the group size is larger than $20$. In addition, the difference
between \emph{CBAS-ND} and \emph{RGreedy} becomes more significant as $k$
grows. Figure \ref{exp2}(b) presents the
solution quality with different activity sizes, where $m=\frac{n}{k}$, $\rho
=0.3$, and $w=0.9$, respectively. The results indicate that \emph{CBAS-ND}
outperforms \emph{DGreedy}, \emph{RGreedy}, and \emph{CBAS}, especially
under a large $k$. The willingness of \emph{CBAS-ND} is at least twice of
the one from \emph{DGreedy} when $k=100$. On the other hand, 
\emph{RGreedy} outperforms \emph{DGreedy} since it has a chance to jump out
of the local optimal solution.

In addition to the activity sizes, we compare the running time of \emph{%
RGreedy}, \emph{CBAS-ND}, and \emph{DGreedy} with different social network
sizes in Figure \ref{exp2}(c) with $k=10$. \emph{DGreedy} is always the
fastest one since it is a deterministic algorithm and generates only one
final solution, but \emph{CBAS} and \emph{CBAS-ND} both require less than $%
10$ seconds, whereas \emph{RGreedy }requires more than $10^{3}$ seconds. To evaluate the performance
of \emph{CBAS-ND} with multi-threaded processing, Figure \ref{exp2}(d) shows that we can accelerate the processing speed to around $%
7.6$ times with $8$ threads. The acceleration ratio is slightly lower than $8
$ because OpenMP forbids different threads to write at the same memory position
at the same time. Therefore, it is expected that \emph{CBAS-ND} with parallelization is promising to be deployed as a value-added \textit{cloud service}.

\begin{figure}[t]
\centering
\hspace{4pt} \subfigure[] {\
\includegraphics[scale=0.14]{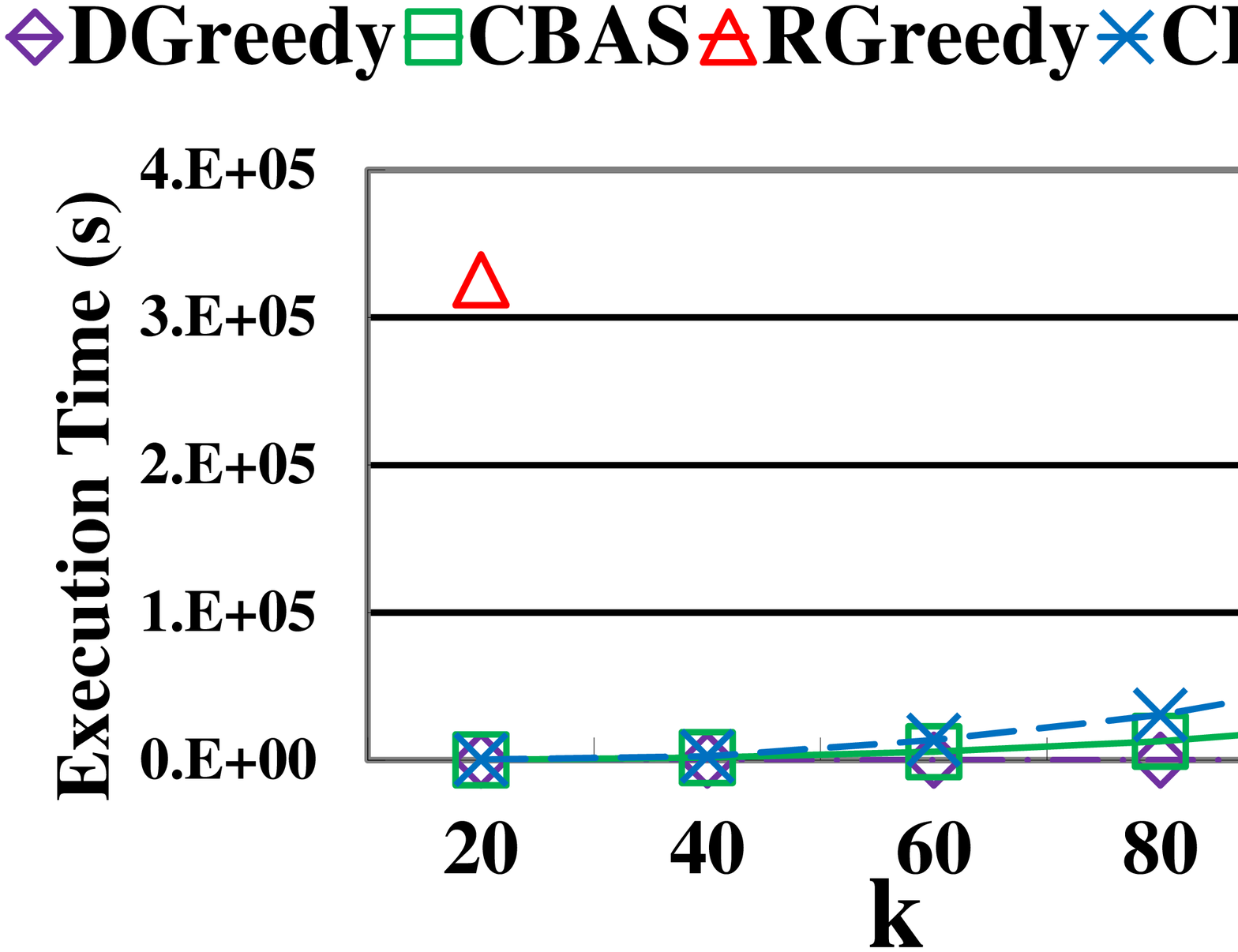} } \hspace{+22pt} \vspace{+25pt}
\subfigure[] {\
\includegraphics[scale=0.15]{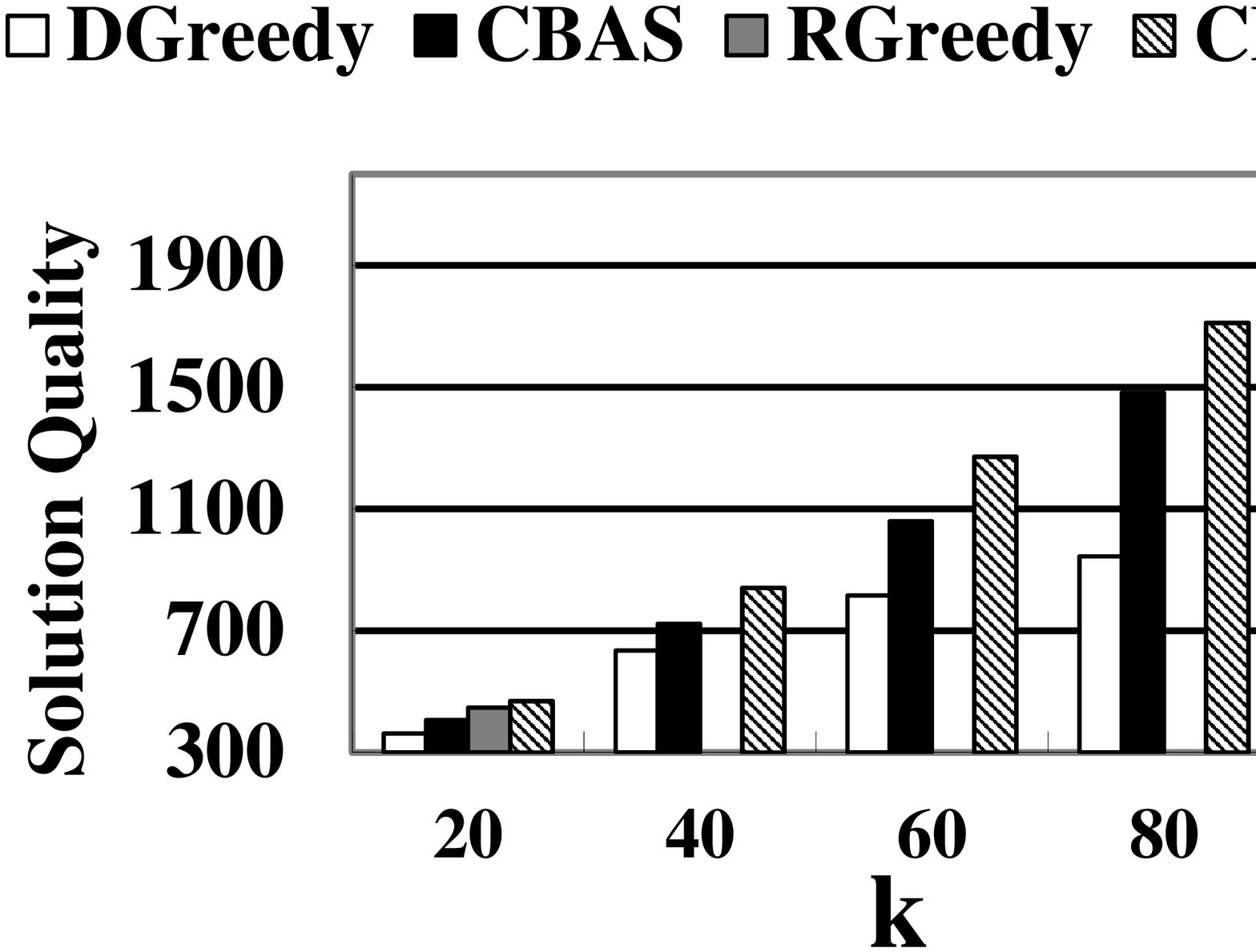} } 
\subfigure[] {\
\includegraphics[scale=0.14]{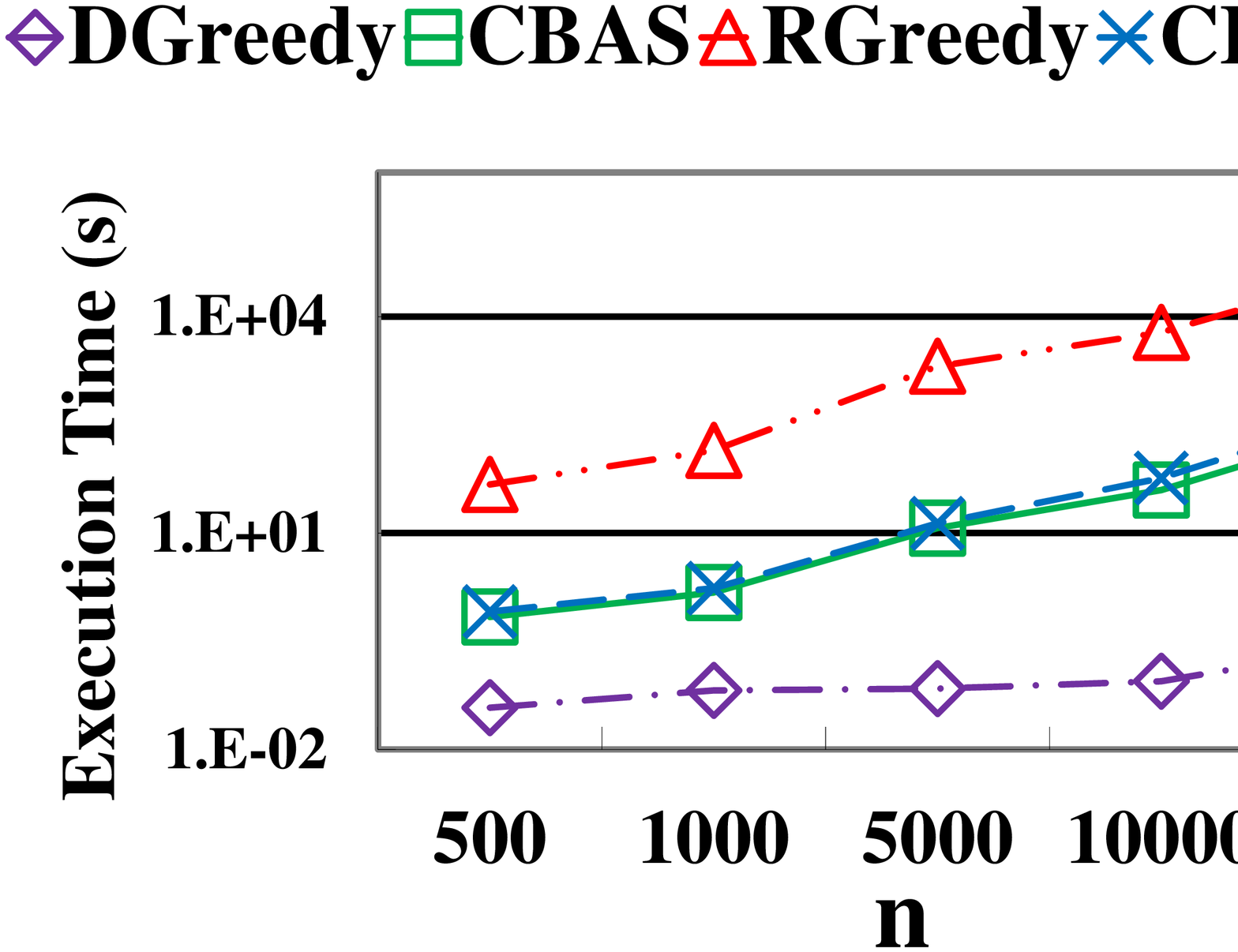} } \hspace{+22pt} \vspace{+25pt}
\subfigure[] {\
\includegraphics[scale=0.15]{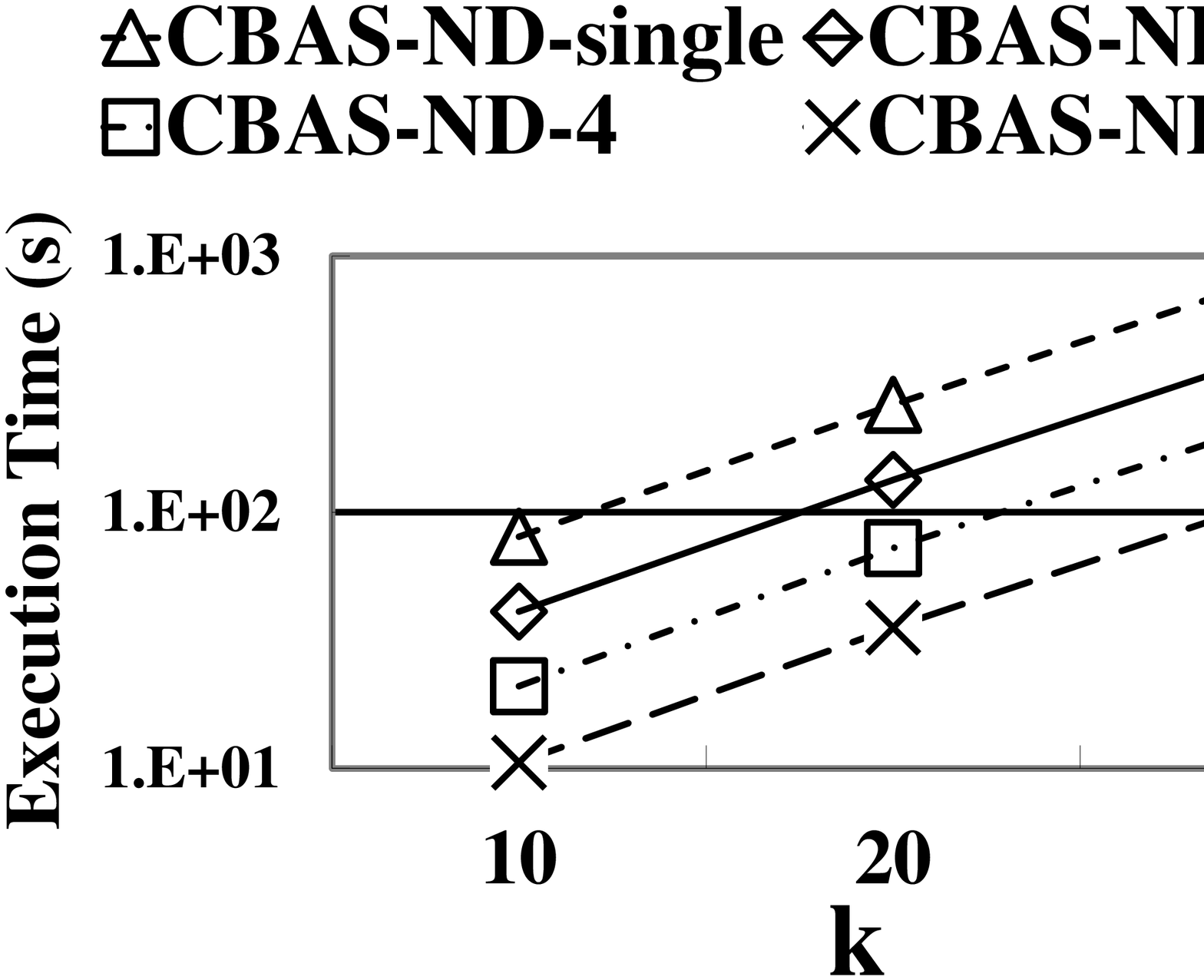} } 
\subfigure[] {\
\includegraphics[scale=0.15]{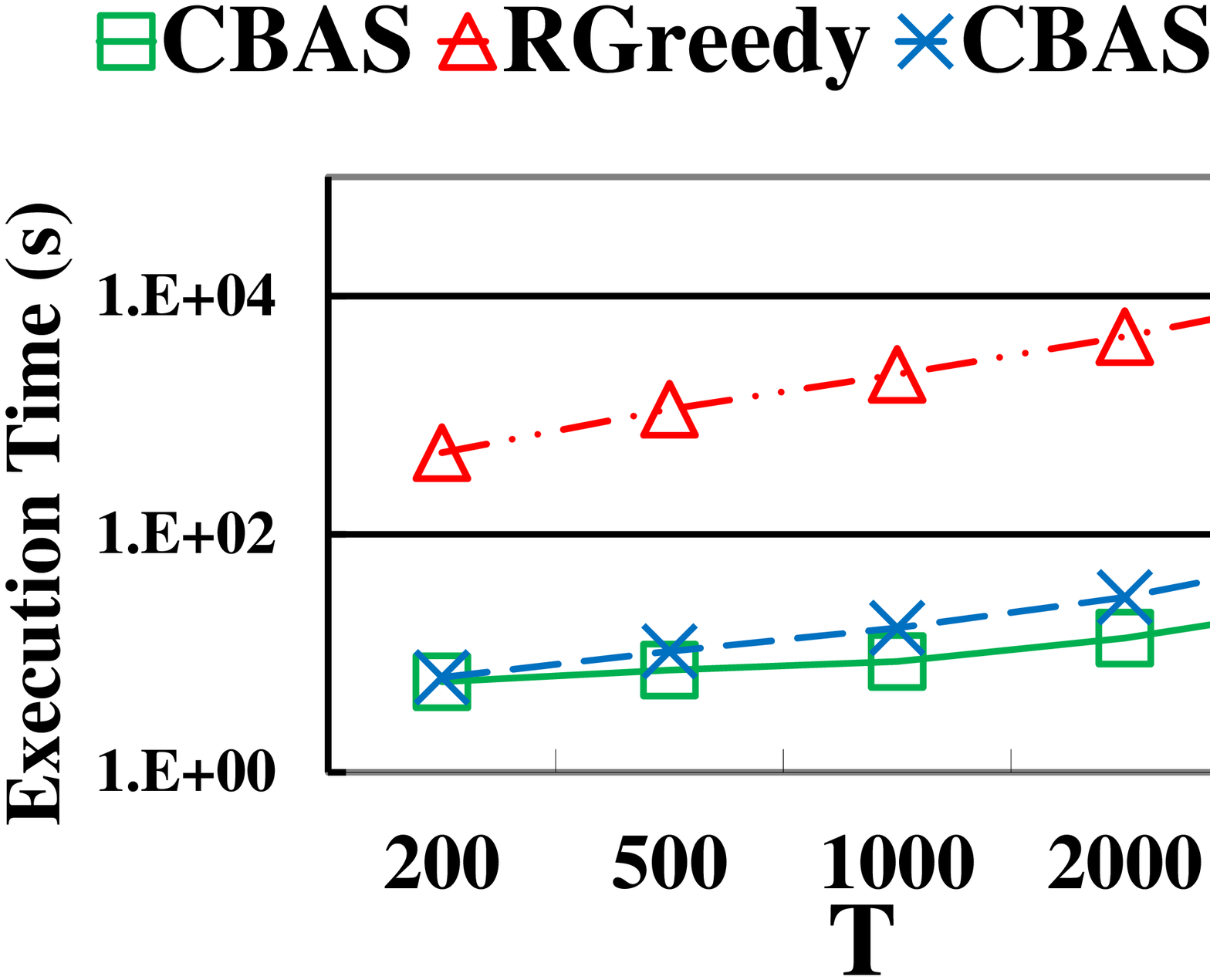} } \hspace{+22pt} \vspace{+25pt}
\subfigure[] {\
\includegraphics[scale=0.15]{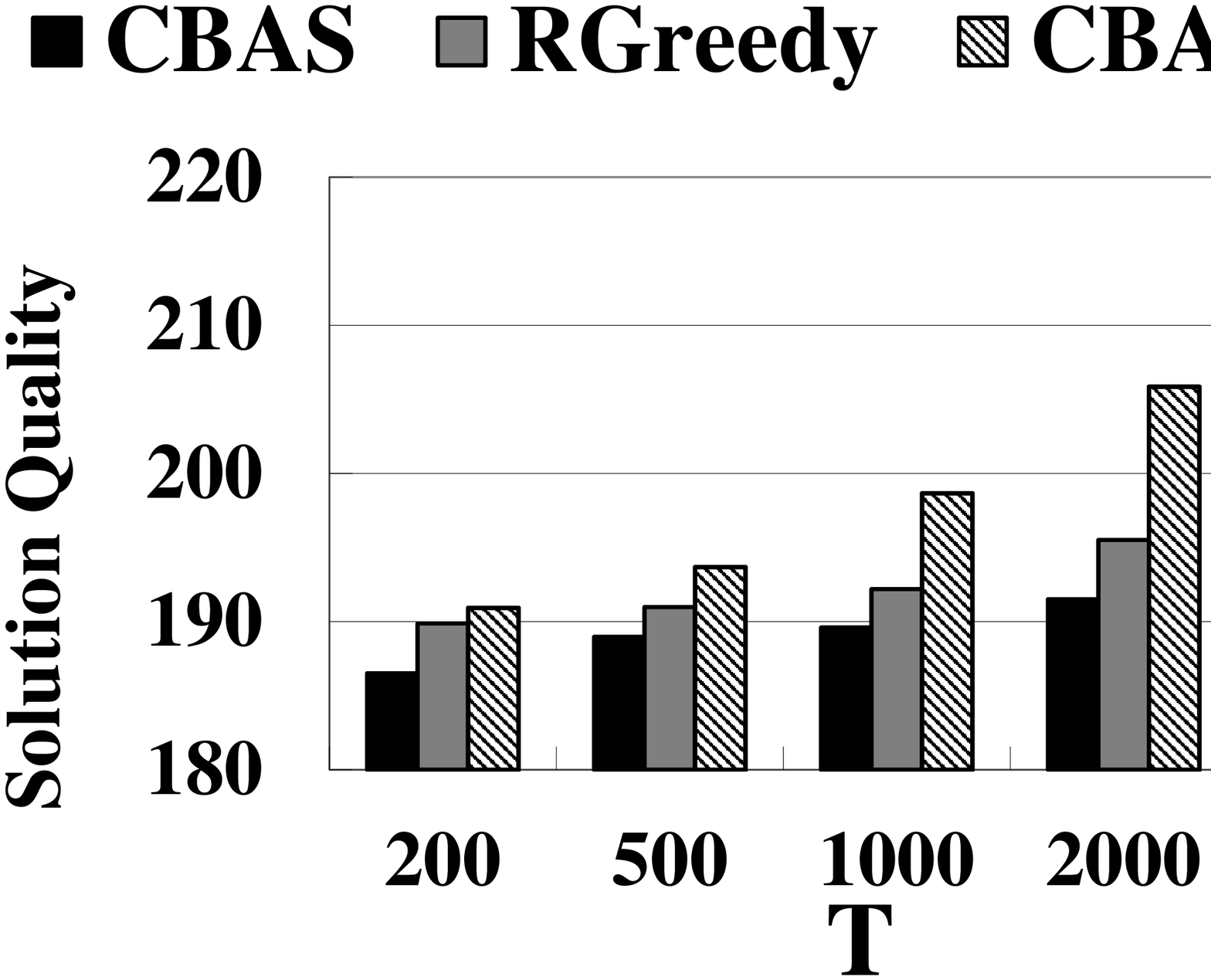} } 
\subfigure[] {\
\includegraphics[scale=0.15]{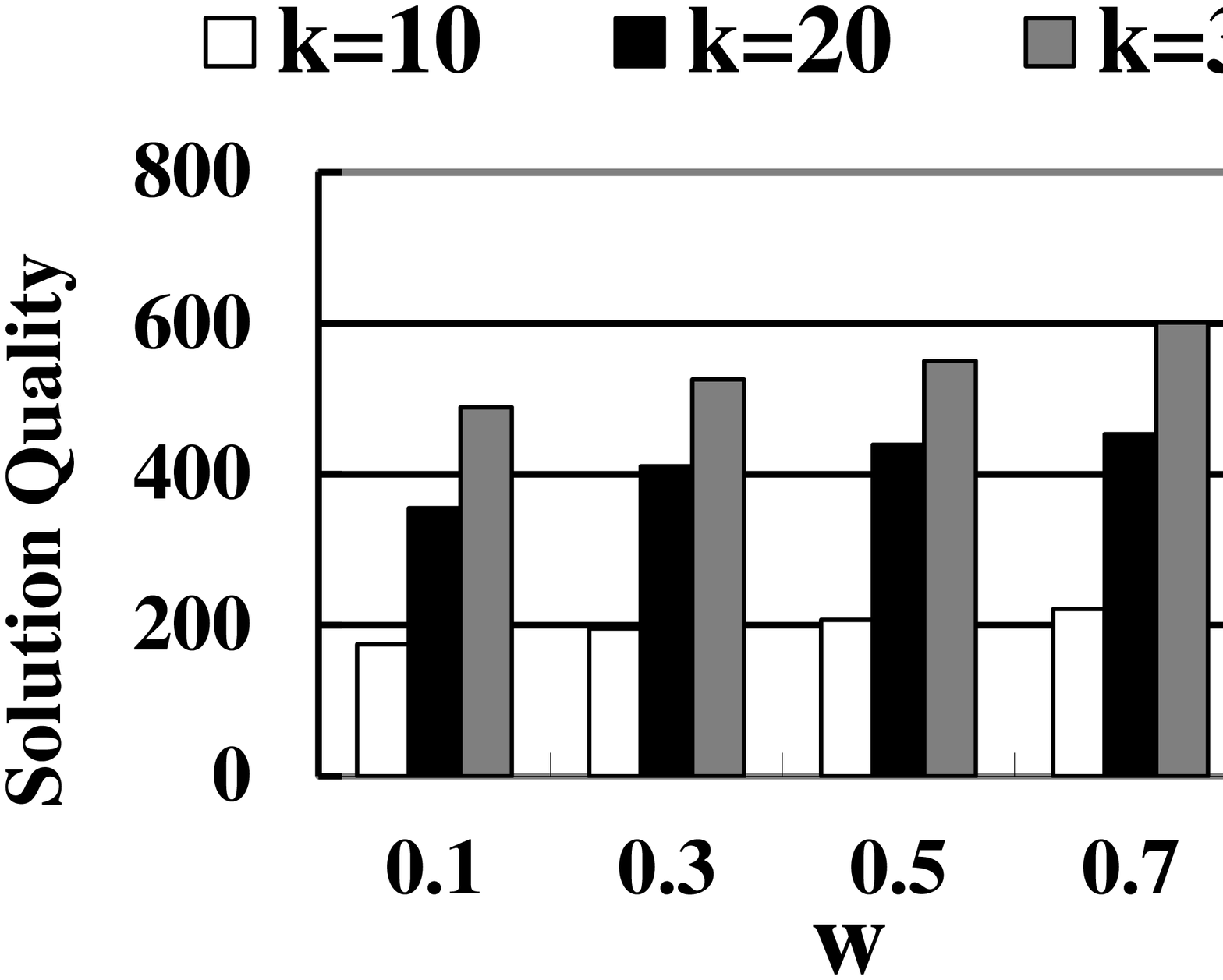} } \hspace{+22pt} \vspace{+25pt}
\subfigure[] {\
\includegraphics[scale=0.15]{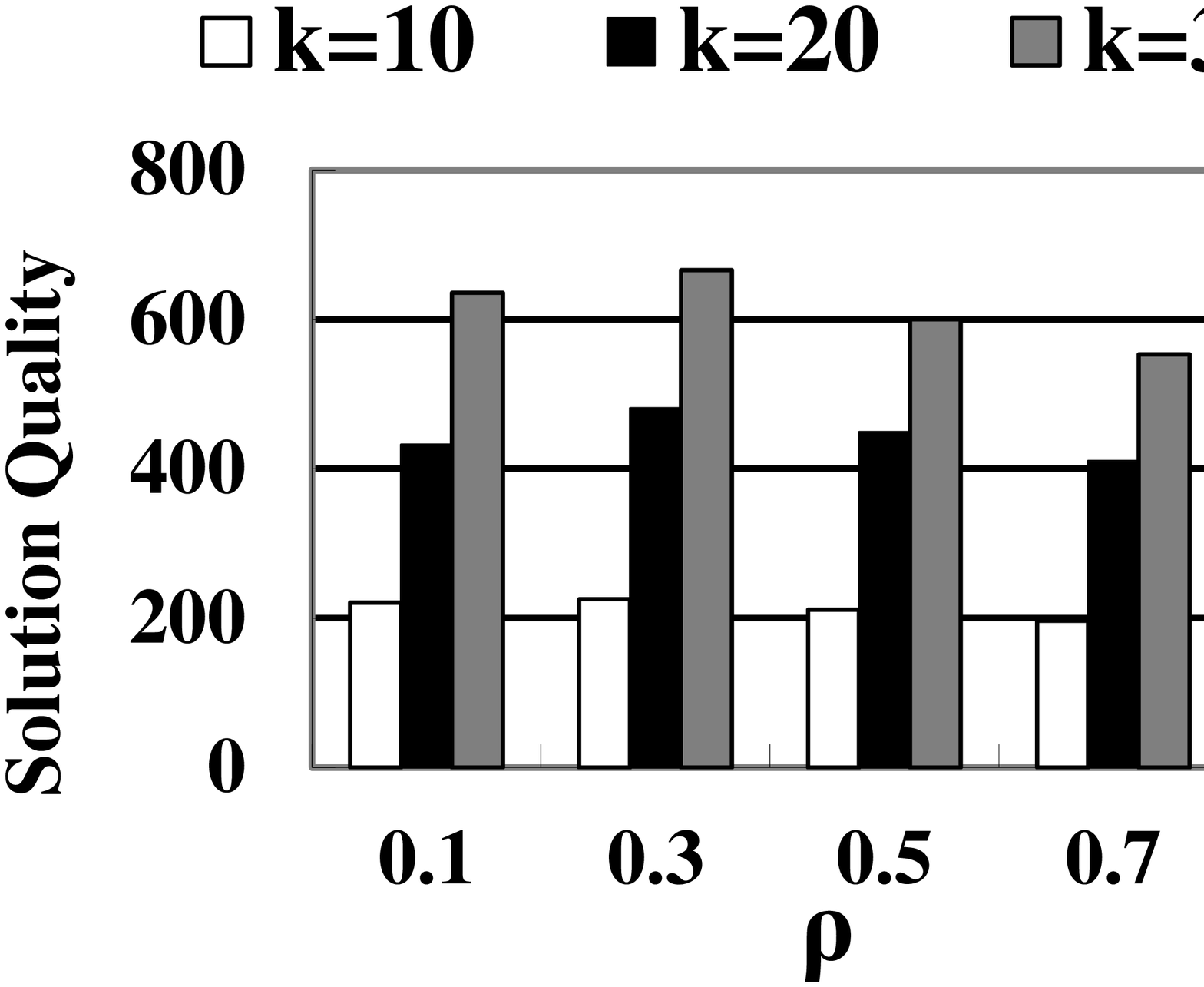} } 
\subfigure[] {\
\includegraphics[scale=0.15]{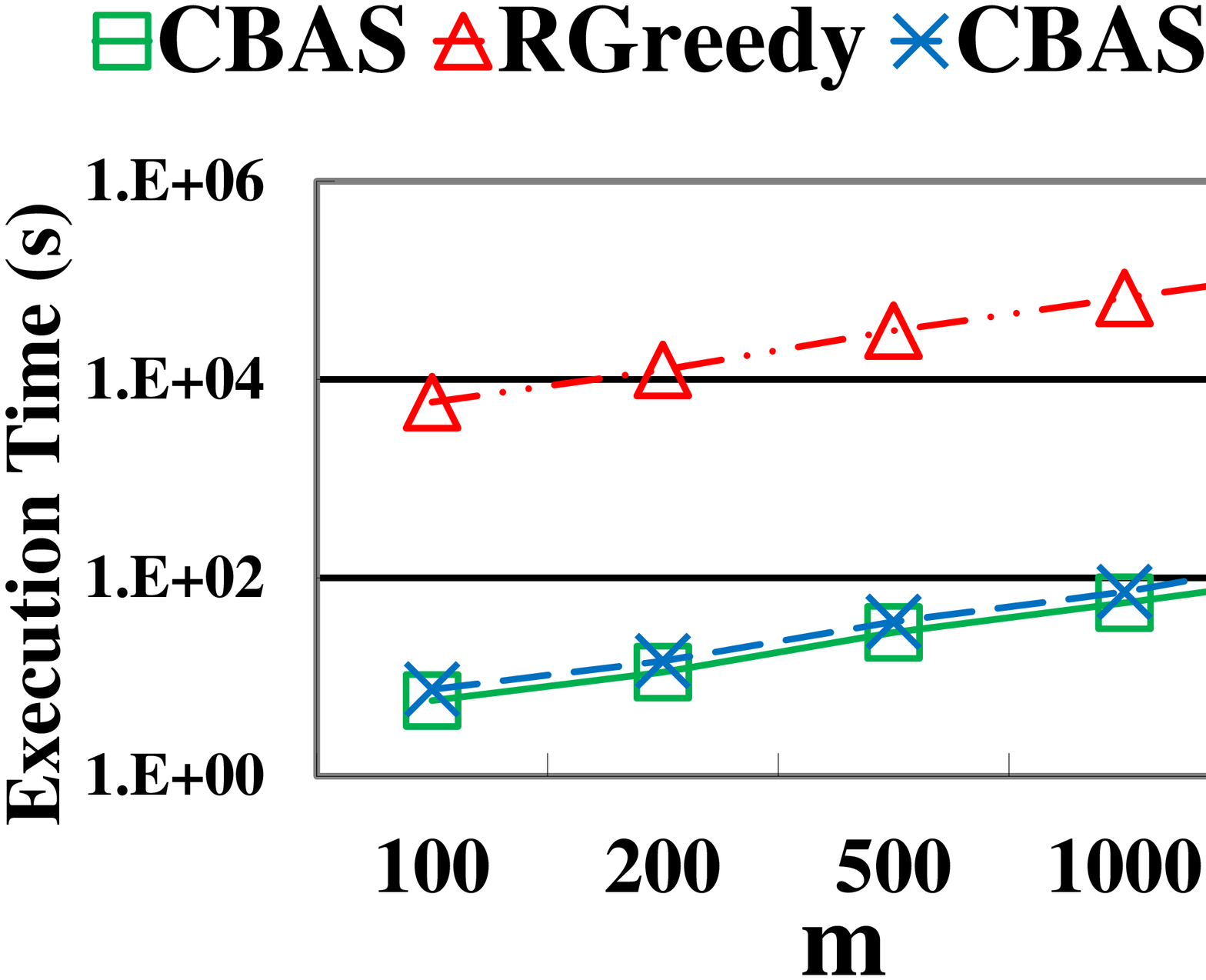} } \hspace{+22pt} \vspace{+25pt}
\subfigure[] {\
\includegraphics[scale=0.15]{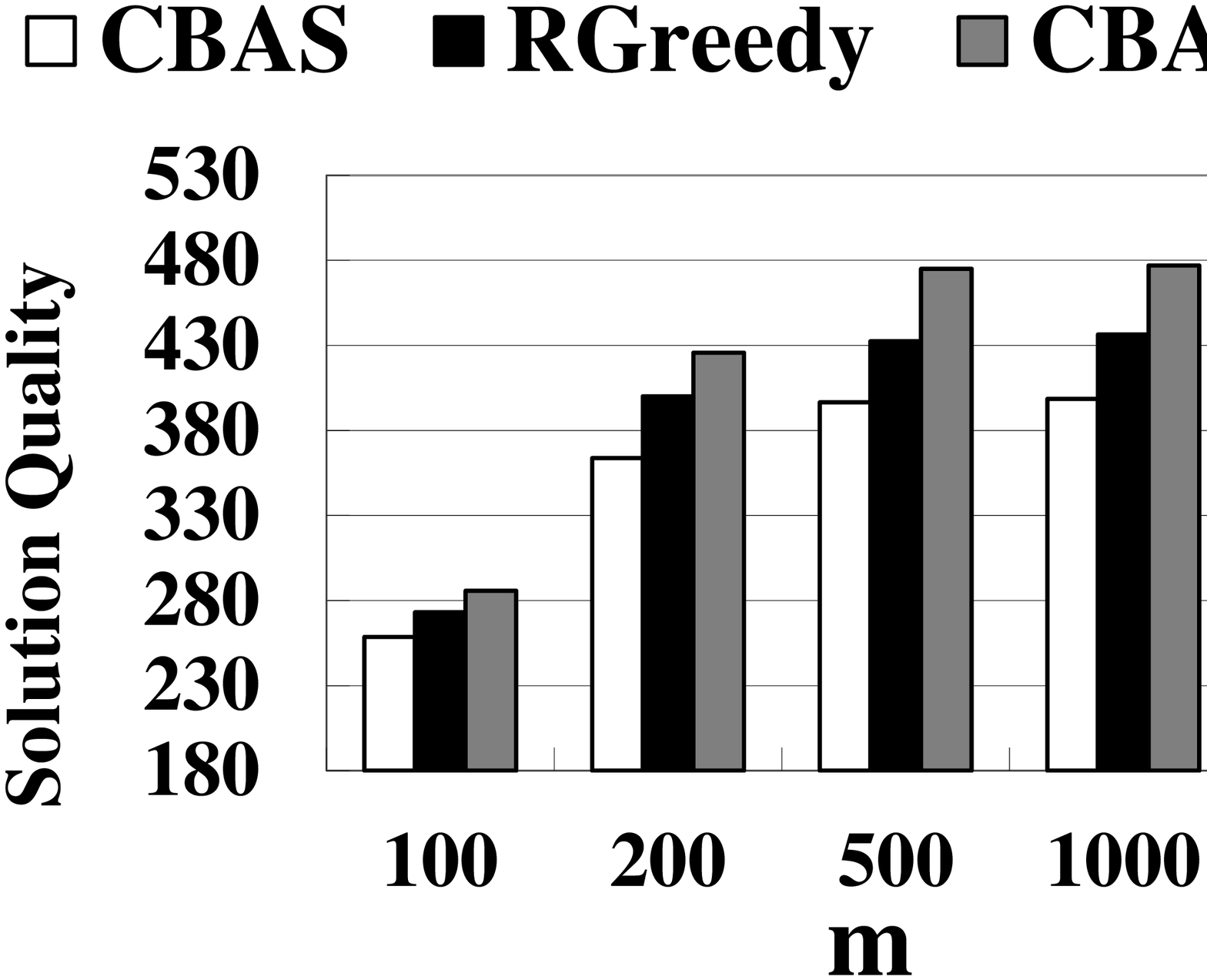} } \vspace{-37pt} 
\caption{Experimental results on Facebook dataset}
\vspace{-20pt}
\label{exp2}
\end{figure}

Figures \ref{exp2}(e) and (f) compare the running time and solution quality
of three randomized approaches under different total computational budgets,
i.e., $T$. As $T$ increases, the solution quality of \emph{CBAS-ND}
increases faster than that of the others because it can optimally allocate the
computation resources. The running time of \emph{CBAS-ND }is slightly larger than that of \emph{CBAS} since \emph{CBAS-ND }needs to sort and extract the samples with
high willingness in previous stages to generate better samples in the
following stage. Even though the solution quality of \emph{RGreedy} is closer to \emph{CBAS-ND} in some cases, both \emph{CBAS }and \emph{CBAS-ND} are faster than \emph{%
RGreedy} by an order of $10^{-2}$.

Figure \ref{exp2}(g) presents the solution quality of \emph{CBAS-ND%
} with different smoothing technique parameters, i.e., $w$. Notice that the
node selection probability vector is homogeneous if we set $w$ to zero. The
result shows that the best result is generated by $w=0.9$ for $k=10$, $20$,
and $30$, implying that the convergence rate with $w=0.9$ is most suitable for WASO in the Facebook dataset. Figure \ref{exp2}(h) compares the top
percentile of performance sample value $\rho $. The result manifests that
the solution quality is not inversely proportional to $\rho $, because for a
smaller $\rho $, the number of samples selected to generate the node
selection probability vector decreases, such that the result converges faster to a solution.

Figures \ref{exp2}(i) and (j) present the running time and solution quality
of \emph{RGreedy}, \emph{CBAS}, and \emph{CBAS-ND} with different numbers of
start nodes, i.e., $m$. The results show that the solution quality in Figure %
\ref{exp2}(j) converges when $m$ is equal to $500$, which indicates that it
is sufficient for $m$ to be set as a value smaller than $\frac{n}{k}$ 
as recommended by OCBA \cite{OCBA10}. By assigning $m=500$ in the Facebook dataset, we
can reduce the running time to only $20\%$ of the running time in $m=2000$, while
the solution quality remains almost the same.

\begin{figure}[tp]
\centering
\hspace{-10pt} \subfigure[] {\
\includegraphics[scale=0.15]{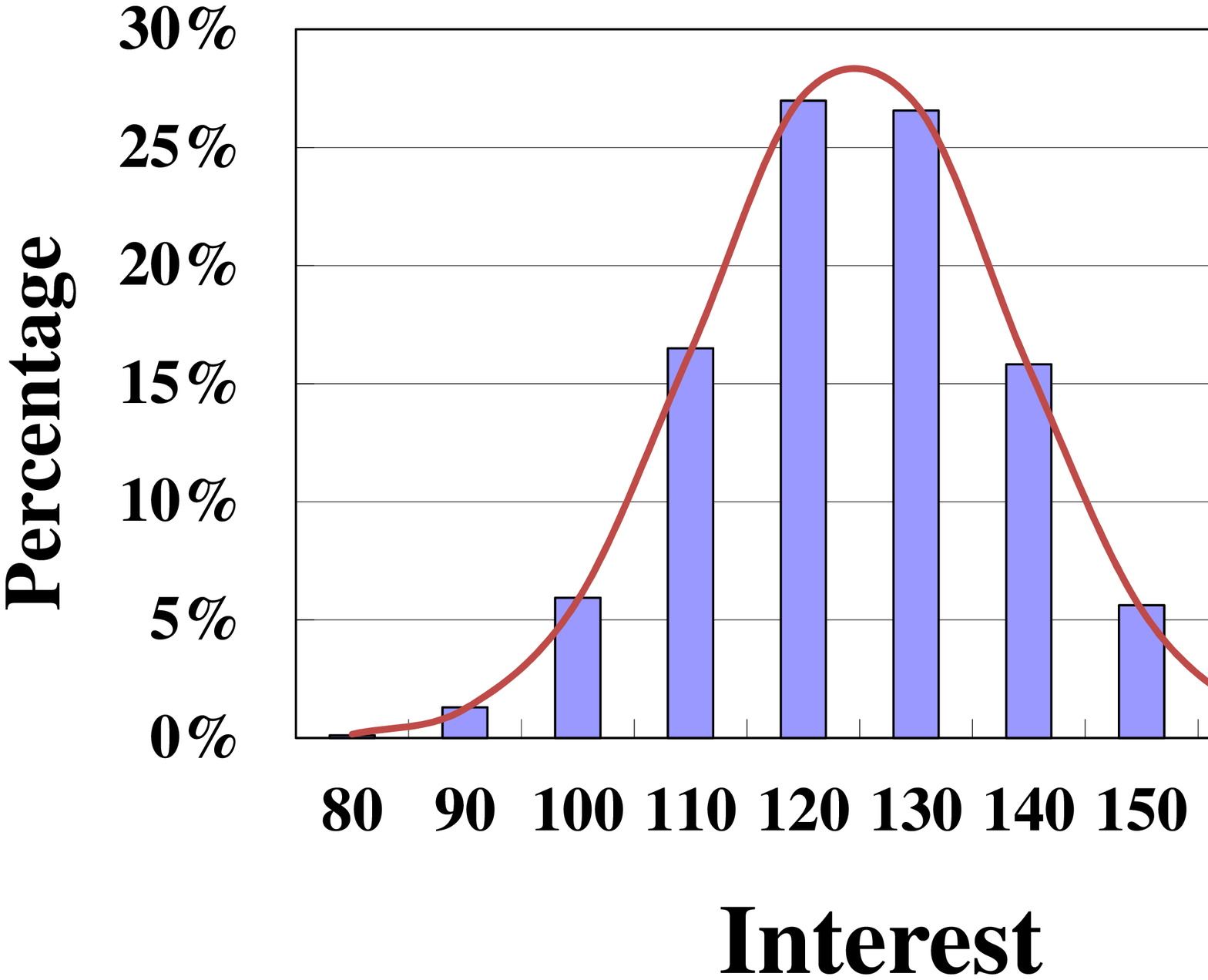} } \hspace{+22pt}
\subfigure[] {\
\includegraphics[scale=0.15]{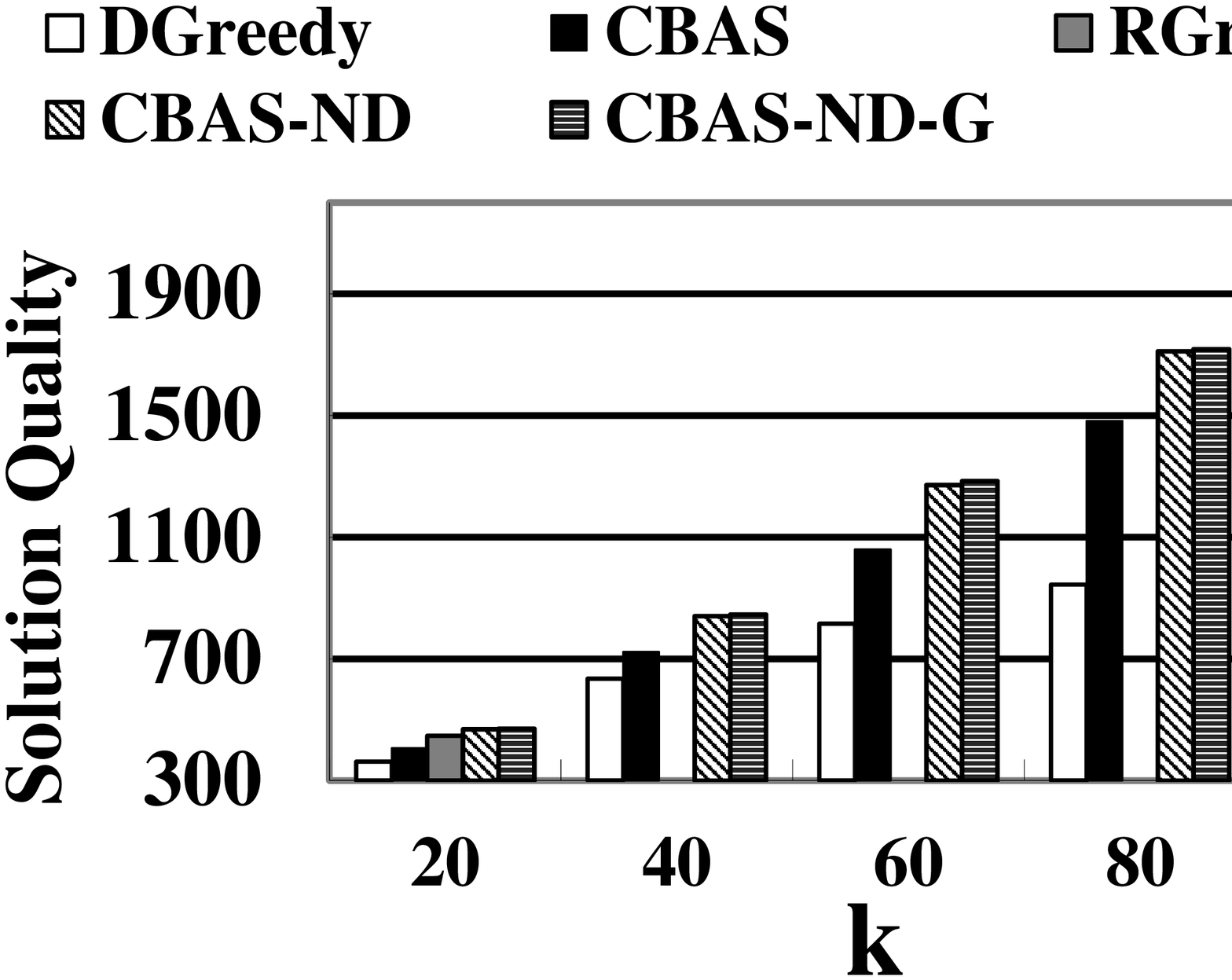} }
\vspace{-10pt}
\caption{Experimental results of WASO with Gaussian distribution}
\vspace{-1mm}
\label{Expgauss}
\end{figure}

Figure \ref{Expgauss}(a) shows the interest histogram of random samples on
Facebook, which indicates that the distribution follows a Gaussian
distribution with the mean as 124.71 and variance as 13.83. The allocation ratio for the variant \emph{CBAS-ND-G} of \emph{%
CBAS-ND} by replacing the uniform distribution with the Gaussian distribution in
Theorem \ref{Prababilitya} is derived in Appendix.
Figure \ref{Expgauss}(b) indicates that the solution quality of \emph{CBAS-ND%
} and \emph{CBAS-ND-G} is very close. In contrast to \emph{CBAS-ND-G},
however, \emph{CBAS-ND} is more efficient and easier to be implemented
because it does not involve the probability integration to find the probability of the best start node.

\subsubsection{DBLP}
\emph{CBAS} and \emph{CBAS-ND} is also
evaluated on the DBLP dataset. Figures \ref{exp3}(a) and (b) compare the
solution quality and running time. The results show that \emph{CBAS-ND}
outperforms \emph{DGreedy} by $92\%$ and \emph{RGreedy} by $32\%$ in
solution quality. Both \emph{CBAS} and \emph{CBAS-ND} are still faster than 
\emph{RGreedy} by an order of $10^{-2}$. However, \textit{RGreedy }runs
faster on the DBLP dataset than on the Facebook dataset, because the DBLP dataset is
a sparser graph with an average node degree of $3.66$. Therefore, the
number of candidate nodes for each start node in the DBLP dataset increases much
more slowly than in the Facebook dataset with an average node degree of $26.1$. Nevertheless, \emph{RGreedy} is still not able to generate a solution for a large group size $k$ due to its unacceptable efficiency. 

\begin{figure}[tp]
\centering
\vspace{+24pt}
\subfigure[] {\
\includegraphics[scale=0.15]{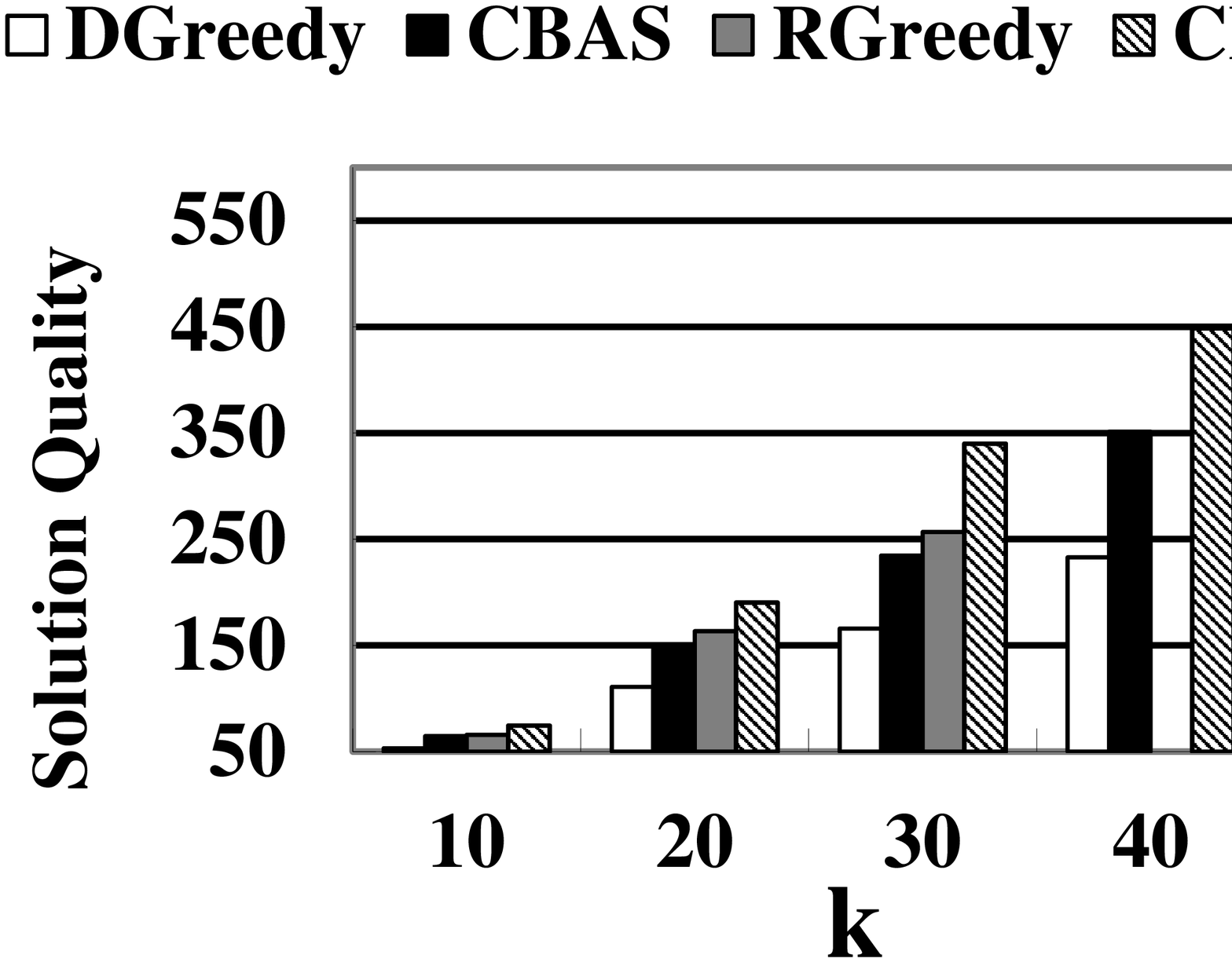} } \hspace{+22pt}  \vspace{+25pt}
\subfigure[] {\
\includegraphics[scale=0.15]{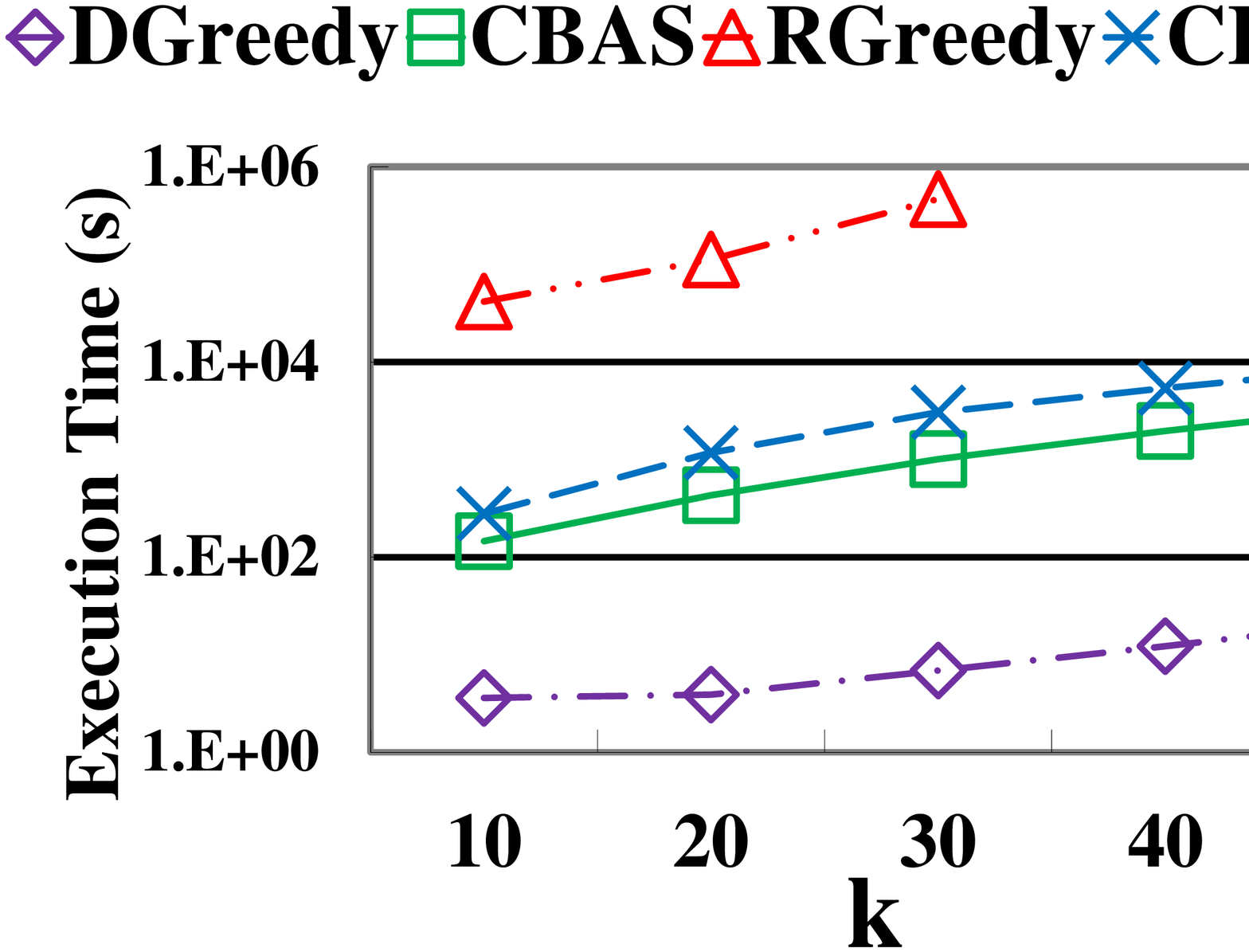} }  
\subfigure[] {\
\includegraphics[scale=0.14]{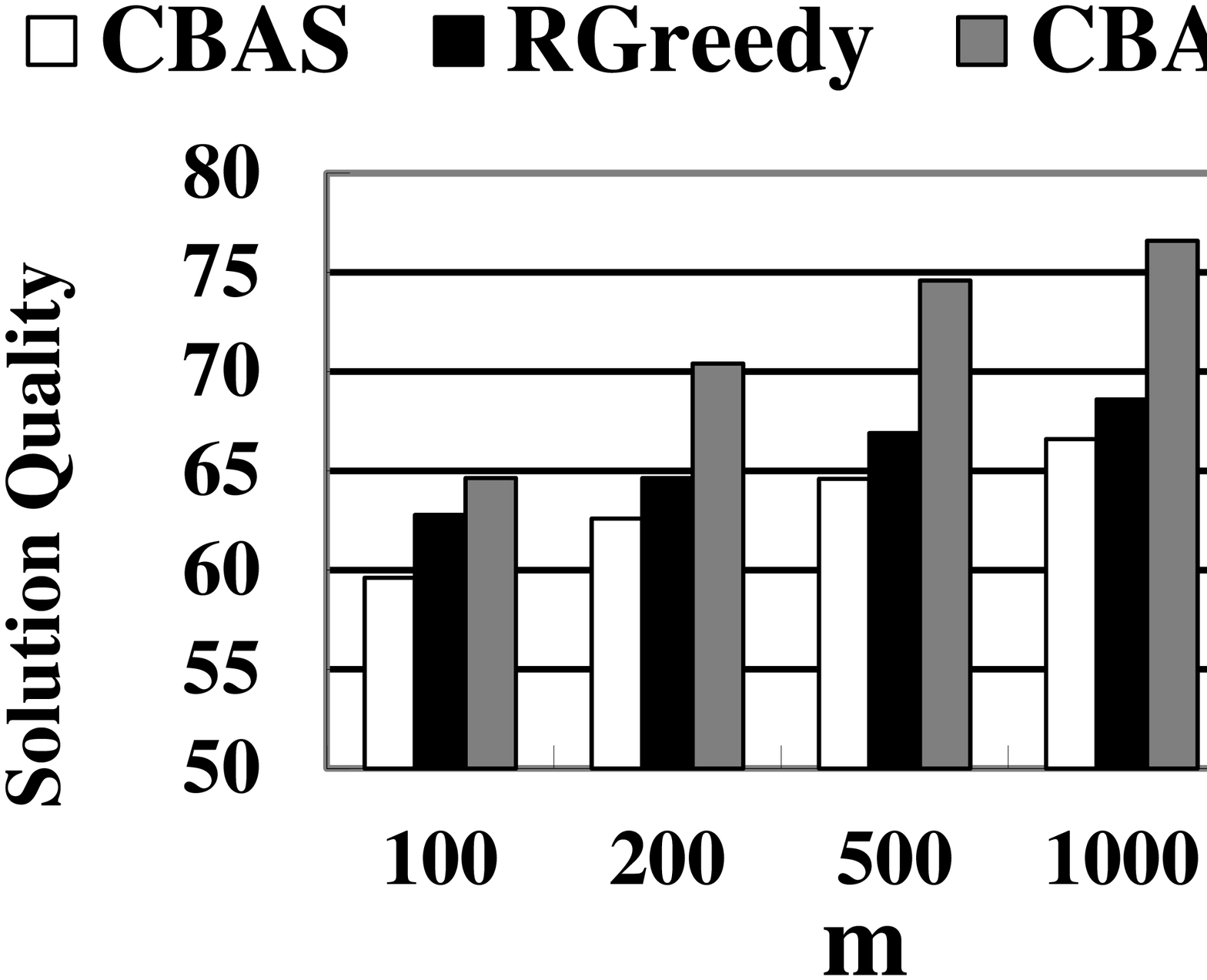} } \hspace{+22pt}  \vspace{+25pt}
\subfigure[] {\
\includegraphics[scale=0.15]{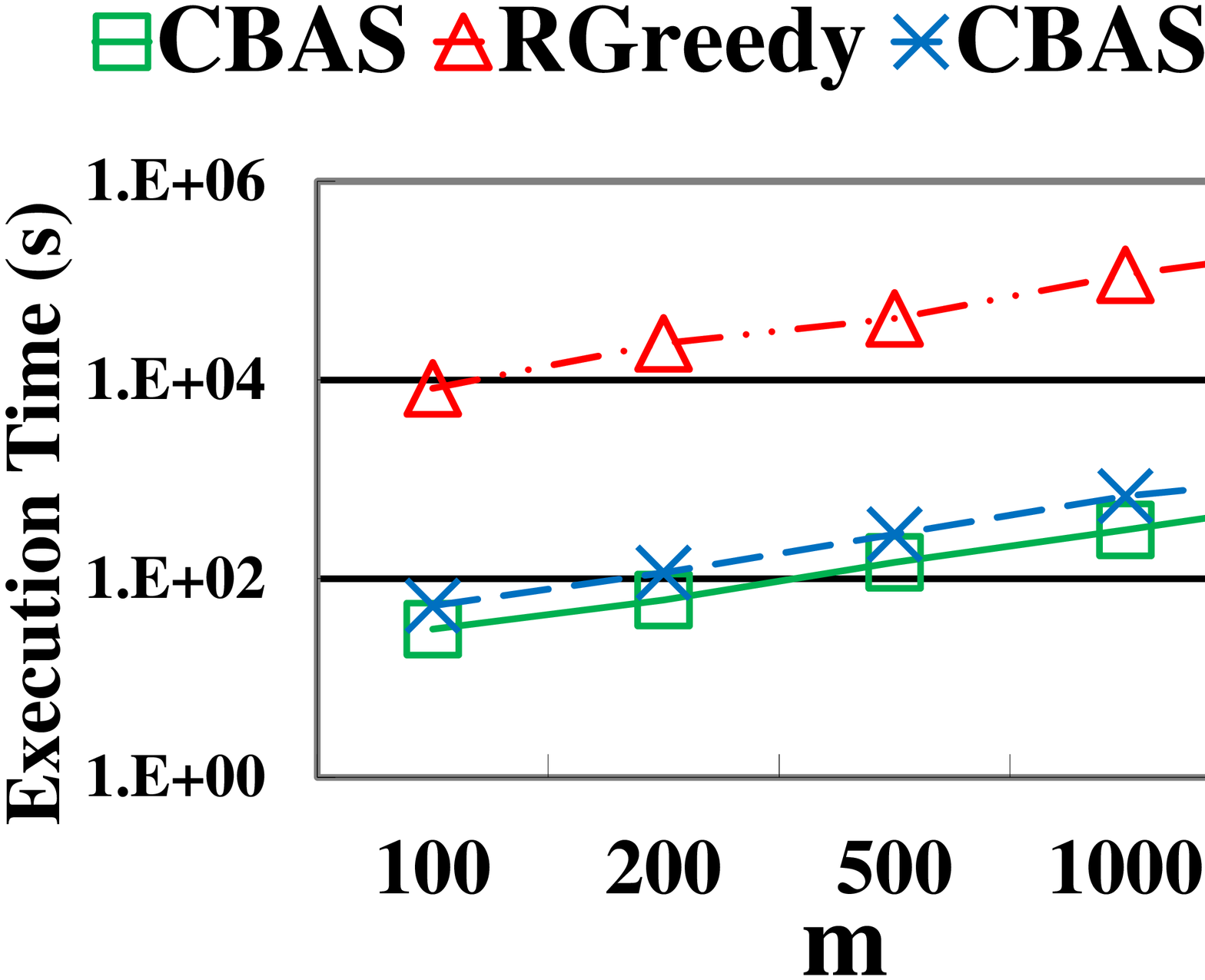} } \vspace{-37pt} 
\subfigure[] {\
\includegraphics[scale=0.15]{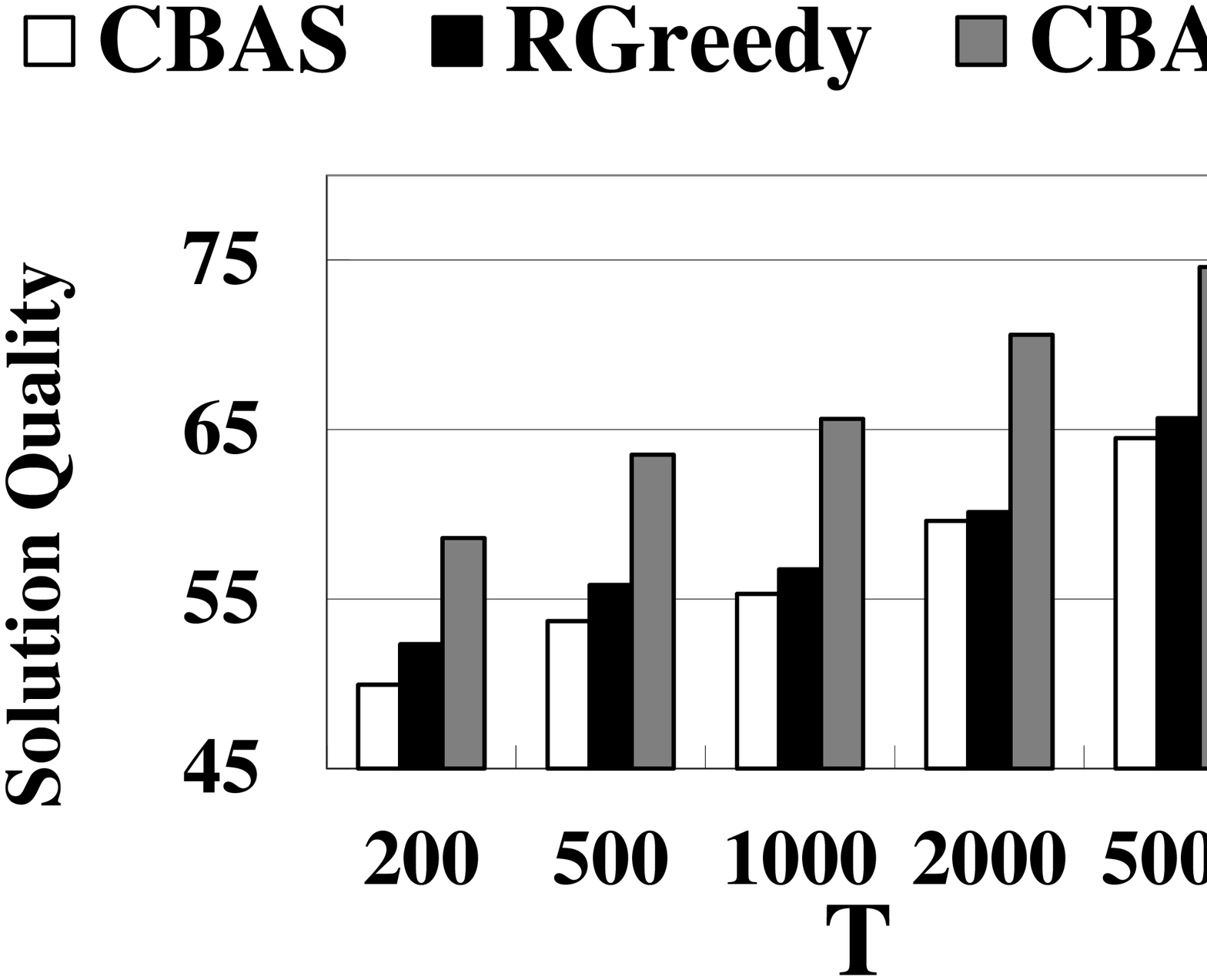} } \hspace{+22pt}  \vspace{+25pt}
\subfigure[] {\
\includegraphics[scale=0.15]{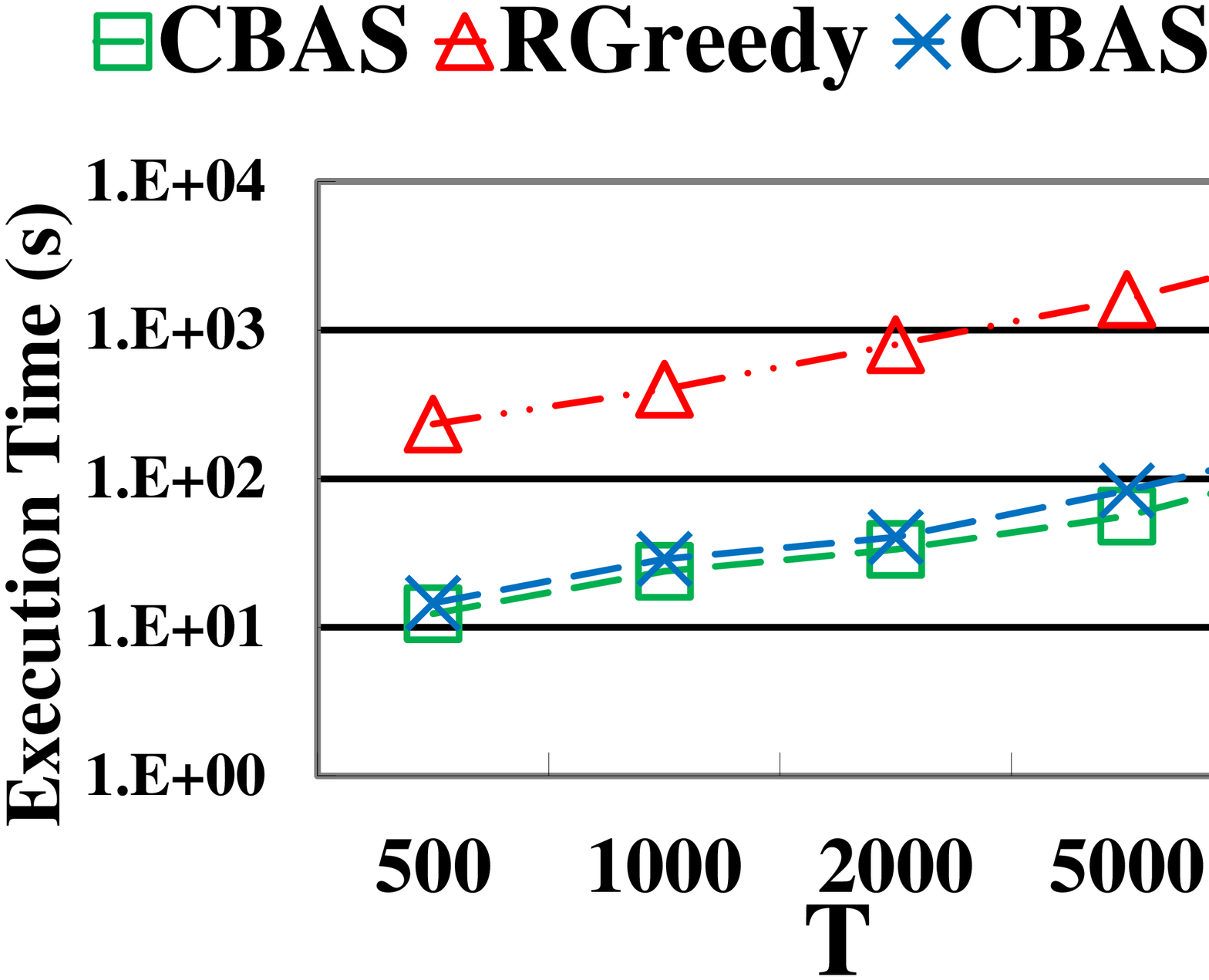} }
\caption{Experimental results on DBLP dataset}
\vspace{-6mm}
\label{exp3}
\end{figure}

Figures \ref{exp3}(c) and (d) present the solution quality and running time of \emph{%
RGreedy}, \emph{CBAS,} and \emph{CBAS-ND} with different numbers of start
nodes, i.e., $m$. The solution quality of \emph{%
CBAS-ND }converges when $m$ is $1000$, indicating that here it is sufficient
to assign $m$ as a number much smaller than $\frac{n}{k}=\frac{511163}{10}%
\approx 51116$, because the way we select start node efficiently filter out the start nodes that do not generate good solutions. Compared to $m=500$ in Facebook dataset, \emph{CBAS }and \emph{CBAS-ND} need a larger $m$ as $1000$ due to a larger network size in DBLP dataset. Figures \ref{exp3}(e) and (f) compare the solution quality and running time with different $T$. As $T$ increases, the solution quality of \emph{CBAS-ND} also grows faster than the other approaches. Both \emph{CBAS }and 
\emph{CBAS-ND }outperform \emph{RGreedy} by an order of $10^{-1}$.

\subsubsection{Flickr}
Finally, to evaluate the scalability of \emph{CBAS} and \emph{CBAS-ND},
Figures \ref{fig_Flickr}(a) and (b) compare the solution quality and
running time on Flickr dataset. The results show that \emph{CBAS-ND}
outperforms \emph{DGreedy} by $31\%$ in solution quality when $k=50$. 
\emph{CBAS} and \emph{CBAS-ND} are both faster than \emph{RGreedy} in an
order of $10^{-2}$. The trend of running time on
Flickr dataset is similar to Facebook dataset, instead of DBLP dataset, because the average node degrees of the Flickr dataset and Facebook
dataset are similar. Moreover, \emph{RGreedy} can support only $k=20$ in the Flickr dataset, smaller than $k=30$ in the DBLP dataset, manifesting that it is not practical to deploy \emph{RGreedy} in a real massive social network.


\begin{figure}[t]
\centering
\subfigure[] {\
\includegraphics[scale=0.15]{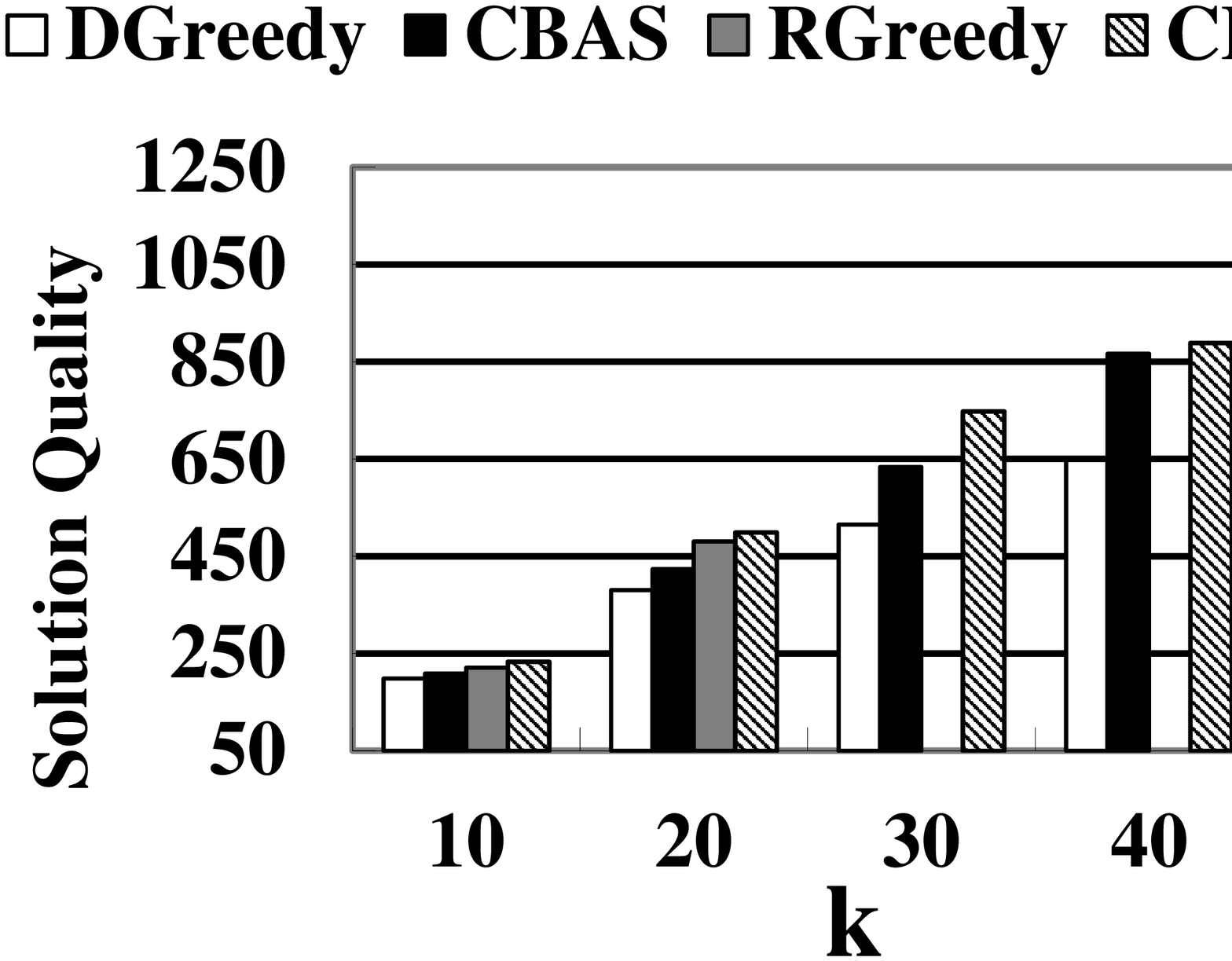} } \hspace{+22pt}\vspace{+25pt}
\subfigure[] {\
\includegraphics[scale=0.15]{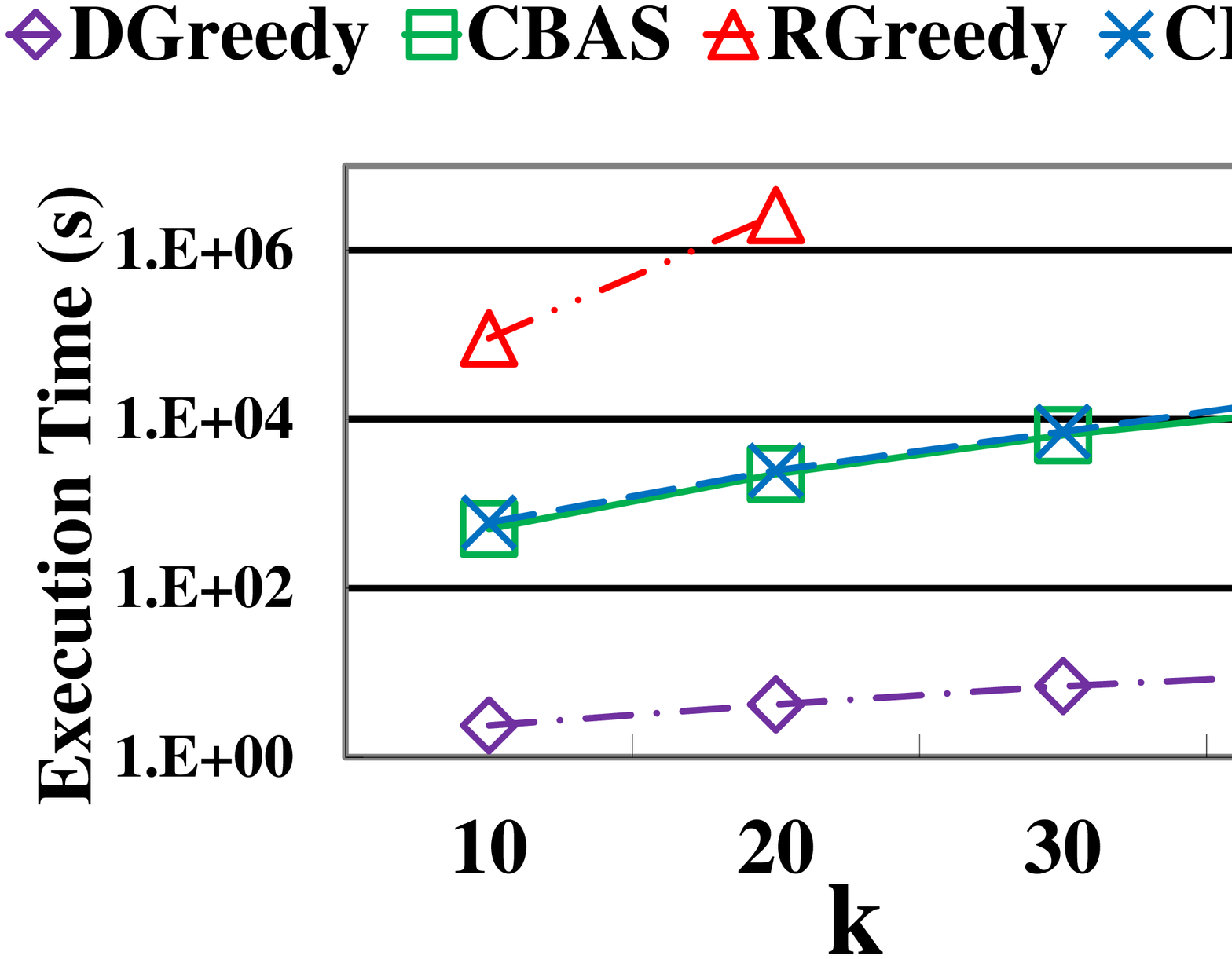} }
\vspace{-37pt}
\caption{Experimental results on Flickr dataset}
\vspace{-18pt}
\label{fig_Flickr}
\end{figure}

\vspace{-2mm}
\subsubsection{Integer Programming and WASO-dis}

To evaluate the solution quality of \emph{CBAS-ND}, Figures \ref%
{fig_IP}(a) and (b) compare the solution quality and running time
of \emph{IP} (ground truth) with $k=10$. Since WASO is NP-hard,
i.e., the running time for obtaining the ground truth is unacceptably large,
we extract 1000 small real datasets from the DBLP dataset with the node sizes as 25, 100, and
500 respectively. The result shows that the solution quality of \emph{%
CBAS-ND} is very close to \emph{IP}, while the running time is smaller by an order
of $10^{-2}$.  It is worth noting that \emph{CBAS-ND} here is single-threaded, but IP is solved by IBM CPLEX (parallel version).

For separate groups, Figure \ref{fig_IP}(c) first presents the running time
with different group sizes, i.e., $k$, where $m=\frac{n}{k}$, $\rho=0.3$, and $w=0.9$, respectively. For all algorithms, the virtual node $\overline{v}$ is added to the selection set $V_{S}$ to relax the connectivity constraint. \emph{RGreedy} computes the incremental willingness of every node in $V_{A}$ to the selection set $V_{S}$, where $V_{A}$ includes all nodes, and thus are computationally intractable. Therefore, \emph{RGreedy} is unable to return a solution within $24$ hours when the group size is larger than $20$. Figure \ref{fig_IP}(d) presents the solution quality with different activity sizes. The results indicate that \emph{CBAS-ND} outperforms \emph{DGreedy}, \emph{RGreedy}, and \emph{CBAS}, especially under a large $k$. In addition, compared to the experimental results in WASO, the difference between \emph{CBAS-ND} and \emph{DGreedy} becomes more significant as $k$ increases. The reason is that the greedy algorithm selects the node with the largest incremental willingness to the current group and thus is inclined to select a connected group, where the optimal solution may be disconnected. 

\begin{figure}[tp]
\vspace{+5pt}
\centering
\subfigure[] {\
\includegraphics[scale=0.14]{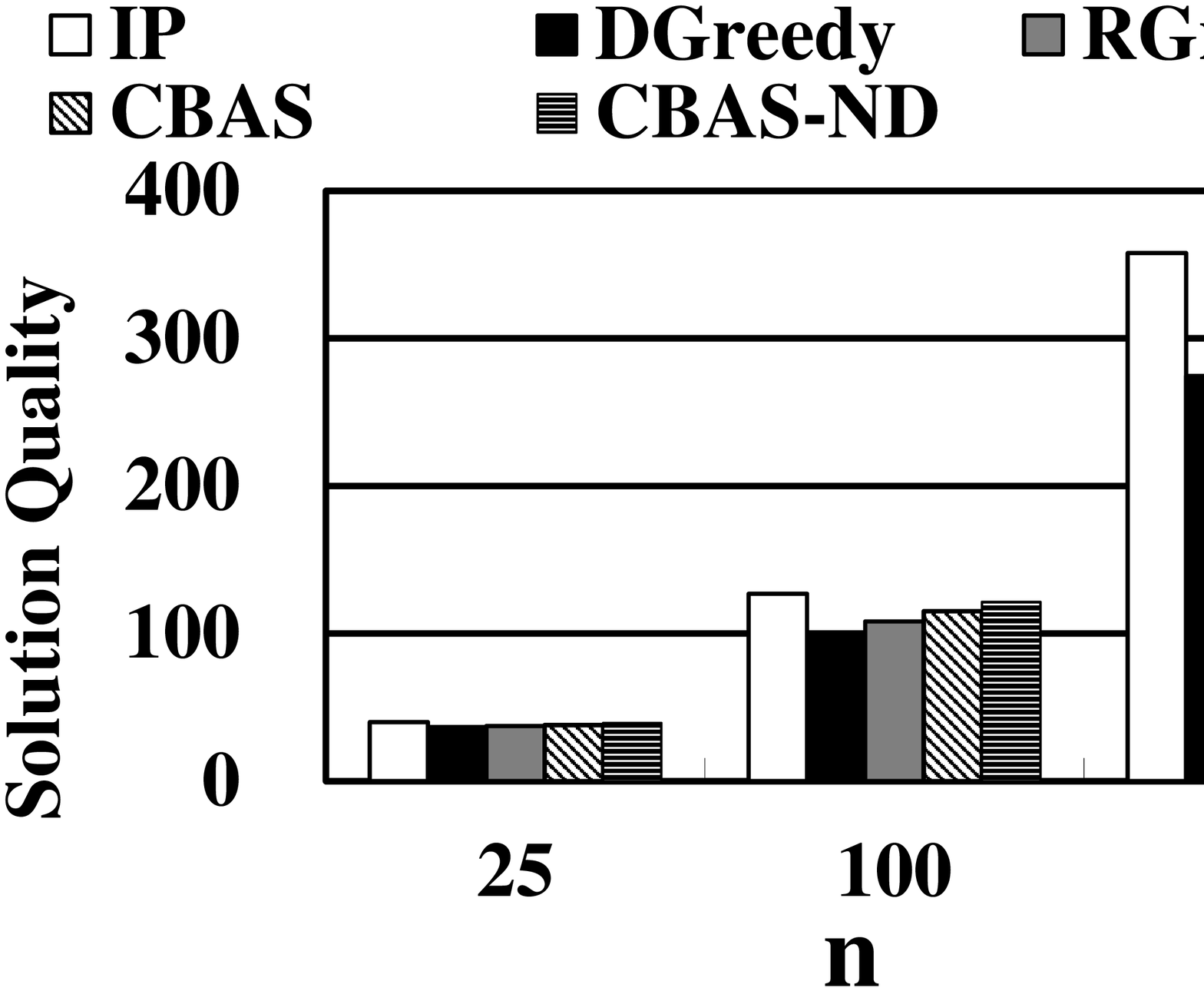} } \hspace{+22pt} \vspace{+25pt}
\subfigure[] {\
\includegraphics[scale=0.15]{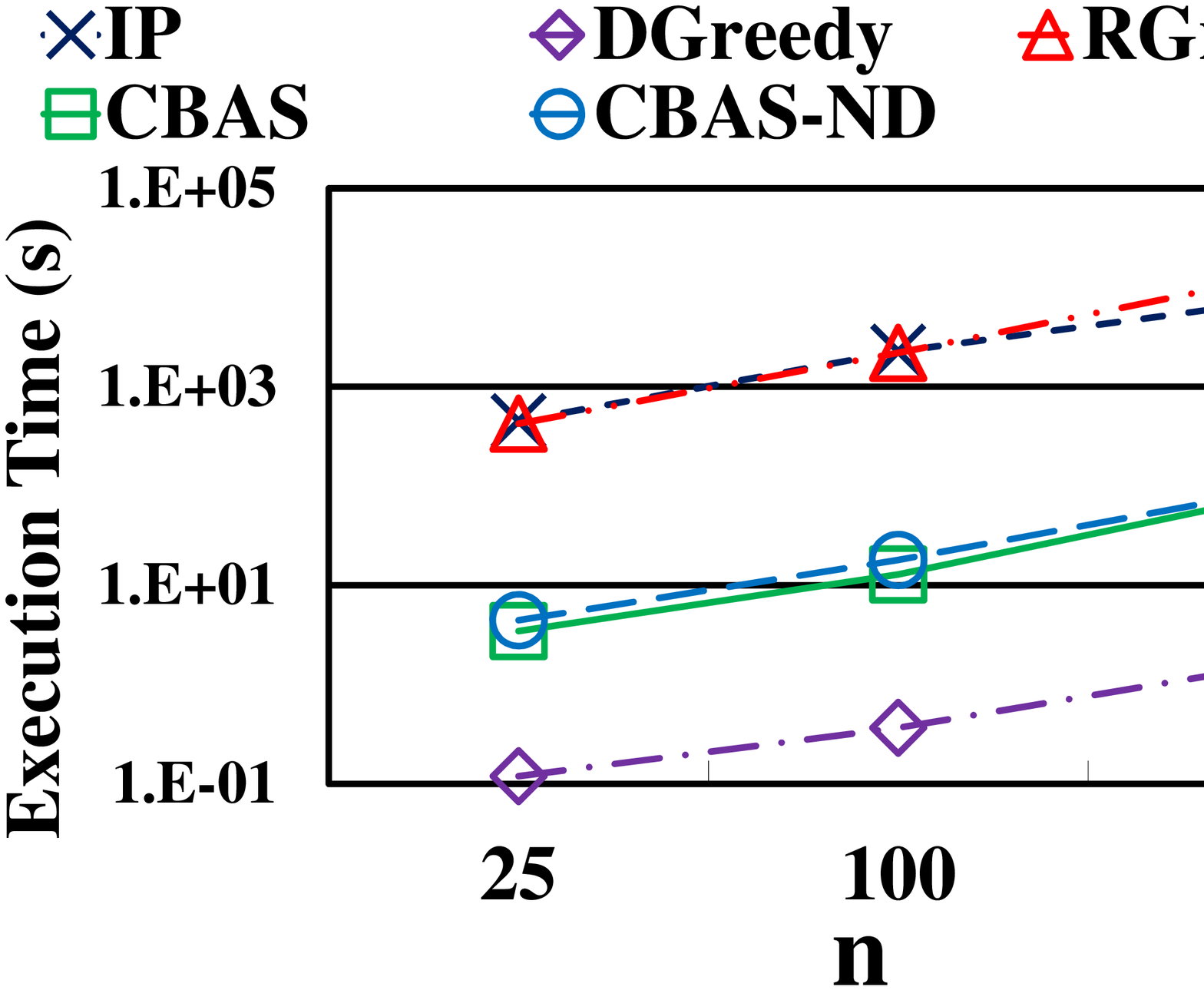} }
\subfigure[] {\
\includegraphics[scale=0.14]{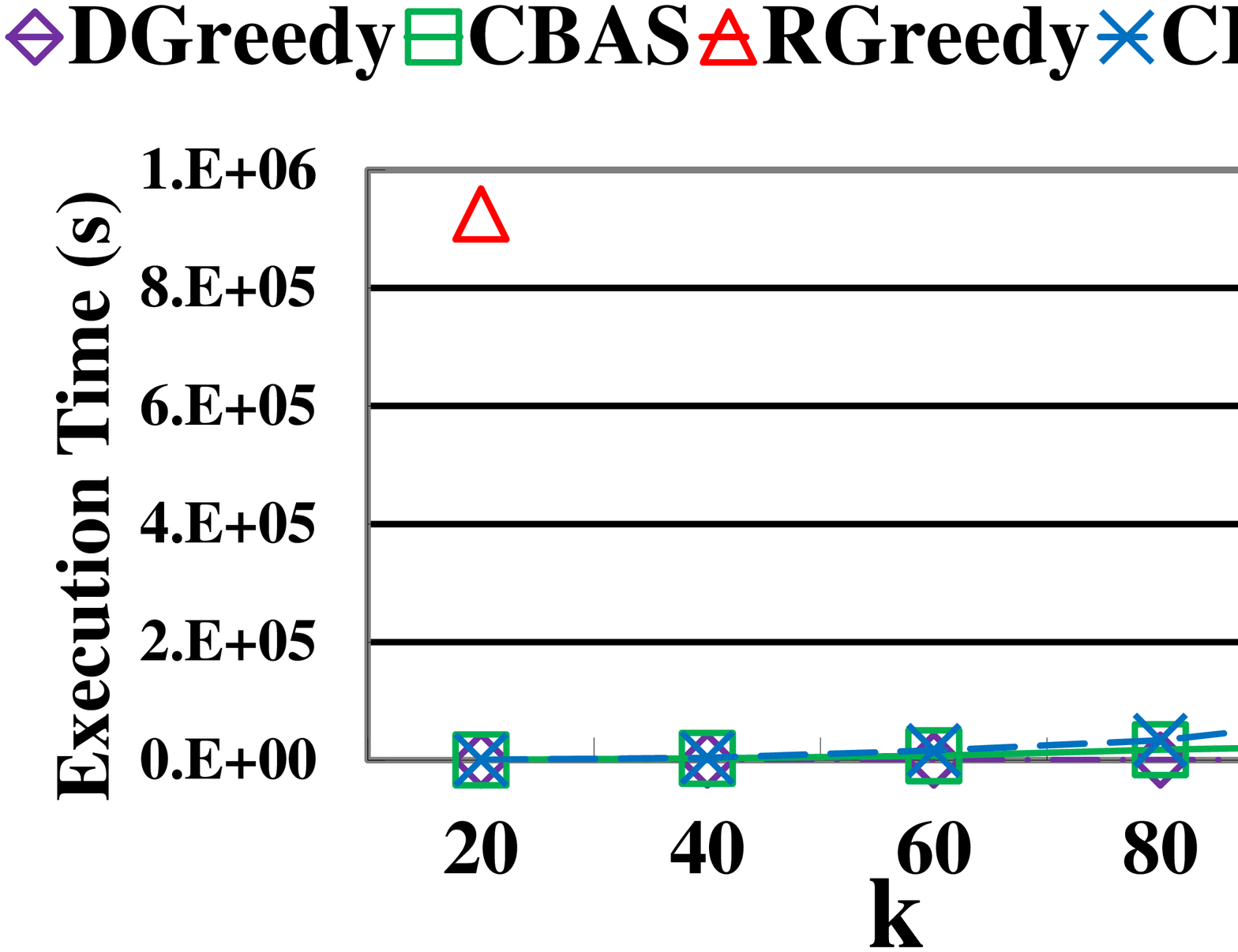} } \hspace{+22pt} \vspace{+22pt}
\subfigure[] {\
\includegraphics[scale=0.14]{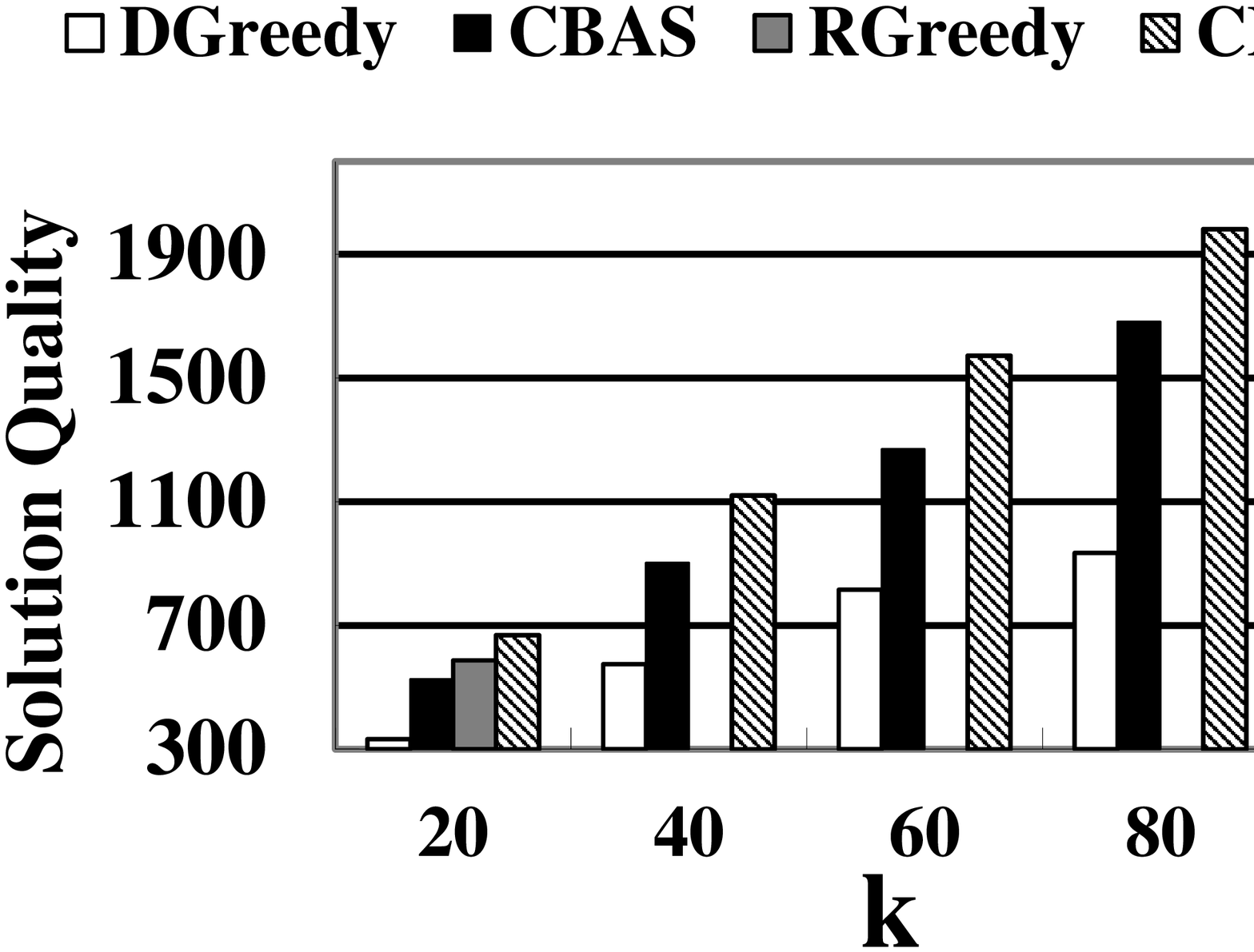} } \vspace{-35pt}
\caption{Experimental results on Integer Programming and WASO-dis}
\vspace{-15pt}
\label{fig_IP}
\end{figure}

\vspace{-5pt}
\section{Conclusion and Future Work}

\label{Conclu}

To the best of our knowledge, there is no real system or existing work in
the literature that addresses the issues of automatic activity planning
based on topic interest and social tightness. To fill this research gap and satisfy an important practical need, this paper formulated a new optimization problem called WASO to
derive a set of attendees and maximize the willingness. We proved that WASO
is NP-hard and devised two simple but effective randomized algorithms,
namely \emph{CBAS} and \emph{CBAS-ND}, with an approximation ratio. The user
study demonstrated that the social groups obtained through the proposed
algorithm implemented in Facebook significantly outperforms the manually configured solutions by users. This research result thus holds much promise to be profitably adopted in social networking websites as a value-added service.

The user study resulted in practical directions to enrich WASO for future research.
Some users suggested that we integrate the proposed willingness optimization
system with automatic available time extraction to filter unavailable users, such as by integrating the proposed system with Google Calendar. Since
candidate attendees are associated with multiple attributes in Facebook,
e.g., location and gender, these attributes can be specified as input
parameters to further filter out unsuitable candidate attendees. Last but
not the least, some users pointed out that our work could be extended to allow
users to specify some attendees that must be included in a certain group
activity. \vspace{-2mm} 
\bibliographystyle{abbrv}
\bibliography{hhshuai_2012}
\appendix
\section{Computational Budget Allocation with Gaussian Distribution}
In the following, we derive the theoretical results for $J_{i}$ following
the normal distribution with mean $\mu _{i}$ and standard deviation of $%
\sigma _{i}$. The probability density function and cumulative distribution
function is as follows.
\vspace{-2mm}
\begin{equation*}
p_{J_{i}}(x)=\phi (\frac{x-\mu _{i}}{\sigma _{i}})=\frac{1}{\sigma _{i}\sqrt{%
2\pi }}e^{-\frac{1}{2}(\frac{x-\mu _{i}}{\sigma _{i}})^{2}}
\end{equation*}
\begin{equation*}
P_{J_{i}}(x)=\Phi (\frac{x-\mu _{i}}{\sigma_{i}})=\frac{1}{2}(1+erf(%
\frac{x-\mu_{i}}{\sigma_{i}\sqrt{2}}))
\end{equation*}

The distribution of maximal value $J_{i}^{\ast }$
\vspace{-2mm}
\begin{equation*}
p_{J_{i}^{\ast }}(x)=N_{i}P_{J_{i}}(x)^{N_{i}-1}p_{J_{i}}(x),
\end{equation*}%
\begin{equation*}
P_{J_{i}^{\ast }}(x)=P_{J_{i}}(x)^{N_{i}}. 
\end{equation*}

Therefore, we derive the probability that $J_{b}^{\ast }$ is smaller than $J_{i}^{\ast }$ as follows.
\vspace{-2mm}
\begin{align*}
&p(J_{b}^{\ast }-J_{i}^{\ast }\leq 0) \\
&=1-p(J_{i}^{\ast }\leq J_{b}^{\ast })\text{ \ \ \ } \\
&=1-\int_{-\infty }^{\infty }p_{J_{b}^{\ast }}(x)P_{J_{i}^{\ast }}(x)dx \\
&=1-\int_{-\infty }^{\infty
}N_{b}P_{J_{b}}(x)^{N_{b}-1}p_{J_{b}}(x)P_{J_{i}}(x)^{N_{i}}dx \\
&=1-N_{b}\int_{-\infty }^{\infty }\Phi (\frac{x-\mu _{b}}{\sigma _{b}}%
)^{N_{b}-1}\phi (\frac{x-\mu _{b}}{\sigma _{b}})\Phi (\frac{x-\mu _{i}}{%
\sigma _{i}})^{N_{i}}dx
\end{align*}

As shown above, the probability is necessary to be computed numerically
because the $\Phi(x)$ function contains $erf(x)$ function which has no
closed-form representation after being integrated. Although we can approximate
the $\Phi(x)$ function with previous works \cite{BrycGP02}, the $\Phi(x)$ function still becomes too complex after raising to the $N_{b}$-th or $N_{i}$-th power.

\vspace{-3pt}
\section{Integer Programming for WASO}
\label{integerprogrammingformulation}
In the following, we describe the Integer Programming (IP) formulation for
WASO. Binary variable $x_{i}$ denotes if node $v_{i}$ is selected in the
solution $F$, and binary variable $y_{i,j}$ denotes if two neighboring nodes 
$v_{i}$ and $v_{j}$ are both selected in $F$. The objective function is
\vspace{-3pt}
\begin{equation*}
\max \sum_{v_{i}\in V}\eta _{i}x_{i}+\sum_{e_{i,j}\in E}\tau _{i,j}y_{i,j}%
\text{,}
\end{equation*}%
where the first term is the total interest score, and the second term is the
total social tightness score of the selected nodes. The basic constraints of
WASO include%
\vspace{-3pt}
\begin{equation}
\sum\limits_{v_{i}\in V}x_{i}=k  \label{eq:B1}
\end{equation}%
\begin{equation}
x_{i}+x_{j}\geq 2y_{i,j},\forall v_{i}\in V,\forall v_{j}\in N_{i}
\label{eq:B2}
\end{equation}%
Constraint (\ref{eq:B1}) states that exactly $k$ nodes are selected in $F$,
while constraint (\ref{eq:B2}) ensures that the social tightness score $\tau
_{i,j}$ of any edge $e_{i,j}$ can be added to the objective function (i.e., $%
y_{i,j}=1$) only when the two terminal nodes $v_{i}$ and $v_{j}$ are both
selected (i.e., $x_{i}=x_{j}=1$); otherwise, $y_{i,j}$ are enforced to be $0$%
.

However, the above basic constraints cannot guarantee that $F$ is a
connected component of $G$, since nodes are allowed to be chosen
arbitrarily. To effective address the issue, we propose the following
advanced constraints for WASO to ensure that there is a path from a root
node in $F$ to every other selected node in $F$, where all nodes in the path
must also belong to $F$. More specifically, let binary variable $r_{i}$
denote if node $v_{i}$ is the root node, and let binary variable $p_{i,j,m,n}
$ denote if edge $e_{m,n}$ in $E$ is located in the path from root node $%
r_{i}$ to another node $v_{j}$ in $F$. It is worth noting that since $F$ is
unknown, variables $r_{i}$ and $p_{i,j,m,n}$ in the advanced constraints are
correlated to $x_{i}$ and $x_{j}$, respectively.

WASO contains the following advanced constraints.%
\vspace{-3pt}
\begin{equation}
\sum\limits_{v_{i}\in V}r_{i}=1  \label{eq:A1}
\end{equation}
\vspace{-3pt}
\begin{equation}
r_{i}\leq x_{i},\forall v_{i}\in V  \label{eq:A2}
\end{equation}%
Constraint (\ref{eq:A1}) states that only one root node will be selected,
while constraint (\ref{eq:A2}) guarantees that the selected root node must
appear in $F$ (i.e, $r_{i}=1$ only when $x_{i}=1$). Equipped with the root
node $r_{i}$, let $N_{j}$ denote the set of neighboring node of $v_{j}$, let 
$d_{i,j,m}$ denote the maximal number of edges in the path from $r_{i}$ to $%
v_{m}$ with $v_{j}$ as the destination of the path, and the following four
constraints identify the path from $r_{i}$ to every node $v_{j}$ in $F$.%
\begin{equation}
r_{i}+x_{j}-1\leq \sum\limits_{n\in N_{i}}p_{i,j,i,n},\forall v_{i},v_{j}\in
V,v_{i}\neq v_{j}  \label{eq:A3}
\end{equation}%
\begin{equation}
r_{i}+x_{j}-1\leq \sum\limits_{m\in N_{j}}p_{i,j,m,j},\forall v_{i},v_{j}\in
V,v_{i}\neq v_{j}  \label{eq:A4}
\end{equation}%
\begin{equation*}
\sum\limits_{q\in N_{m}}p_{i,j,q,m}=\sum\limits_{n\in N_{m}}p_{i,j,m,n},\ \ \ \ \ \ \ \ \ \ \ \ \ \ \ \ \ \ \ \ \ \ \ \ \ \ \ \ \ \ \ \ \ \ \ \ \ \ \ \ \ \ \ \ \ \ \ \ \ \ \ \ \ \ \ \ \ \ \ \ 
\end{equation*}
\begin{equation}
\ \ \ \ \ \ \ \ \ \ \ \ \ \ \ \ \ \ \ \ \ \ \ \ \ \ \ \forall v_{i},v_{j},v_{m}\in V,v_{i}\neq v_{j},v_{i}\neq
v_{m},v_{j}\neq v_{m}  \label{eq:A5}
\end{equation}
\begin{equation}
d_{i,j,m}+(p_{i,j,m,n}-1)\left\vert V\right\vert <d_{i,j,n},\forall
v_{i},v_{j}\in V,\forall e_{m,n}\in E  \label{eq:A6}
\end{equation}%
For the selected root node $r_{i}$ and every other node $v_{j}$ in $F$
(i.e., $x_{j}=1$), the left hand side (LHS) of constraints \ref{eq:A3} and %
\ref{eq:A4} become 1, enforcing that at least one incident edge $e_{i,n}$ of 
$v_{i}$ and one incident edge $e_{m,j}$ of $v_{j}$ must be included in the
path. After obtaining the first and last edge (i.e., $e_{i,n}$ and $e_{m,j}$%
) in the path from $r_{i}$ to $v_{j}$, constraint \ref{eq:A5} is a flow
continuity constraint. For each node $v_{m}$, if it is an intermediate node
in the path, flow continuity constraint states that the flow from $r_{i}$ to 
$v_{m}$ must be identical to the flow $v_{m}$ to $v_{j}$. In other words,
constraint \ref{eq:A5} chooses a parent node $v_{q}$ and a child node $v_{n}$
for $v_{m}$ in the path

Constraint \ref{eq:A6} guarantees that the node sequence in the path
contains no cycle; otherwise, for every edge $e_{m,n}$ in the cycle, $%
p_{i,j,m,n}=1$, and the following inequality holds,%
\vspace{-3pt}
\begin{equation*}
d_{i,j,m}<d_{i,j,n}
\end{equation*}%
and it is thus impossible to find a $d_{i,j,n}$ for every node $v_{n}$ in
the cycle. On the other hand, for any edge with $p_{i,j,m,n}=0$, the
constraint becomes redundant since $d_{i,j,m}-\left\vert V\right\vert
<d_{i,j,n}$ always holds. 

The following constraint ensures that every two terminal nodes $v_{m}$ and $%
v_{n}$ of an edge $e_{m,n}$ in the path (i.e., $p_{i,j,m,n}=1$) must
participate in $F$ (i.e., $x_{m}=x_{n}=1$).

\begin{equation}
p_{i,j,m,n}\leq 2(x_{m}+x_{n}),\forall v_{i},v_{j}\in V,\forall e_{m,n}\in E
\label{eq:A7}
\end{equation}%
Therefore, it is not allowed to arbitrarily choose a path in $G$ to connect
the root node $r_{i}$ to another node $v_{j}$ in $F$.

\newpage
\section{Pseudo Codes}
\label{Pseudocode}

\vspace{-10mm}
\begin{algorithm}[h]
\caption{CBAS}
\label{SIIGS-CBA}
\begin{algorithmic}[1]
\renewcommand{\algorithmicrequire}{\textbf{Input:}}
\renewcommand{\algorithmicensure}{\textbf{Output:}}
\REQUIRE Graph $G(V,E)$, social network size $n$, activity size $k$, correctly select probability $P(CS)$, and solution quality $Q$
\ENSURE The best group F generating maximum willingness
\STATE $c_{i}=\infty$, $d_{i}= 0$ for all $i$;
\STATE $m=\left\lceil \frac{n}{k}\right\rceil$, $w=0$;
\STATE Select $m$ candidate nodes to candidate set $\mathcal{M}$;
\STATE $T_{1}=\left\lceil m\frac{\ln (\frac{2(1-P(CS))}{m-1})}{\ln \alpha }\right\rceil$;
\STATE Find the number of stages $r$ by first consulting $N_{b}$ table with solution $q$, and $r$=$\left\lceil \frac{4N_{b}}{T_{1}}-\frac{4k}{n}+1\right\rceil$;
\FOR {$t=1$ to $r$}
     \IF {$t=1$}
       \FOR {$i=1$ to $m$}
         \STATE $A_{i}=\frac{T_{1}}{m}$;
       \ENDFOR
    \ELSE
           \STATE $A_{total}=0$;
           \FOR {$i=1$ to $m$}
             \STATE $A_{i}$= $\frac{1}{2}(\frac{d_{i}-c_{b}}{d_{b}-c_{b}})^{N_{b}}$;
             \STATE $A_{total}$=$A_{total}$+$A_{i}$;
           \ENDFOR
           \STATE  $A_{i}$= $T_{1}A_{i}$/$A_{total}$; 
    \ENDIF

    \FOR {$i=1$ to $m$}
        \STATE $V_{S}=\mathcal{M}_{i}$
        \STATE $V_{A}=\emptyset$
        \FOR{$x=1$ to $A_{i}$}
           \STATE $V_{A}=N(\mathcal{M}_{i})$
           \FOR{$y=1$ to $k-1$}
              \STATE Random select a node $v$ in $V_{A}$ to $V_{S}$;
              \STATE $V_{A}=V_{A} \cup N(v)$
           \ENDFOR
              \STATE $w= W(V_{S})$;
              \IF {$w>d_{i}$}
                \STATE $d_{i}=w$;
              \ENDIF
              \IF {$w<c_{i}$}
                \STATE $c_{i}=w$;
              \ENDIF
              \IF {$w>S(F)$}
                \STATE $b=j$;
                \STATE $F=V_{S}$;
              \ENDIF
        \ENDFOR
    \ENDFOR
\ENDFOR

    \STATE Output $F$;
\end{algorithmic}
\end{algorithm}


\label{SIIGS-CBACE-ARG}

\begin{algorithm}[H]
\caption{CBAS-ND}
\label{SIIGS-CBACE}
\begin{algorithmic}[1]
\renewcommand{\algorithmicrequire}{\textbf{Input:}}
\renewcommand{\algorithmicensure}{\textbf{Output:}}
\REQUIRE Graph $G(V,E)$, social network size $n$, activity size $k$, correctly select probability $P(CS)$, solution quality $Q$, percentile of CE $\rho$, and smoothing weighting $w$
\ENSURE The best group F generating maximum willingness
\STATE $c_{i}=\infty$, $d_{i}= 0$ for all $i$;
\STATE $m=\left\lceil \frac{n}{k}\right\rceil$, $w=0$;
\STATE Select $m$ candidate nodes to candidate set $\mathcal{M}$;
\STATE $T_{1}=\left\lceil m\frac{\ln (\frac{2(1-P(CS))}{m-1})}{\ln \alpha }\right\rceil$;
\STATE Find the number of stages $r$ by first consulting $N_{b}$ table with solution $q$, and $r$=$\left\lceil \frac{4N_{b}}{T_{1}}-\frac{4k}{n}+1\right\rceil$;
\FOR {$t=1$ to $r$}
     \IF {$t=1$}
       \FOR {$i=1$ to $m$}
         \STATE $A_{i}=\frac{T_{1}}{m}$;
         \STATE Set the node selection probability vector $p_{i,t}$ as uniform;
       \ENDFOR
    \ELSE
           \STATE $A_{total}=0$;
           \FOR {$i=1$ to $m$}
             \STATE $A_{i}$= $\frac{1}{2}(\frac{d_{i}-c_{b}}{d_{b}-c_{b}})^{N_{b}}$;
             \STATE $A_{total}$=$A_{total}$+$A_{i}$;
           \ENDFOR
           \STATE  $A_{i}$= $T_{1}A_{i}$/$A_{total}$; 
    \ENDIF

    \FOR {$i=1$ to $m$}
        \STATE $V_{S}=\mathcal{M}_{i}$
        \STATE $V_{A}=\emptyset$
        \STATE $X=\emptyset$
        \FOR{$x=1$ to $A_{i}$}
           \STATE $V_{A}=N(\mathcal{M}_{i})$
           \FOR{$y=1$ to $k-1$}
              \STATE Random select a node $v$ in $V_{A}$ in accordance with $p_{i,t}$ to $V_{S}$;
              \STATE $V_{A}=V_{A} \cup N(v)$
           \ENDFOR
             \STATE $w=W(V_{S})$;
             \STATE $X.add(V_{S},w)$;
              \IF {$w>d_{i}$}
                \STATE $d_{i}=w$;
              \ENDIF
              \IF {$w<c_{i}$}
                \STATE $c_{i}=w$;
              \ENDIF
              \IF {$w>W(F)$}
                \STATE $b=i$;
                \STATE $F=V_{S}$;
              \ENDIF
        \ENDFOR
    \ENDFOR
                \COMMENT{Update node selection probability $p_{i}$}
                \STATE X=$DescendingSort(X,w)$;
                \IF {$\gamma_{t}>X_{(\left\lceil \rho A_{i}\right\rceil)}.w$}
                    \STATE $\gamma_{t+1}=\gamma_{t}$; 
                \ELSE
                    \STATE $\gamma_{t+1}=X_{(\left\lceil \rho A_{i}\right\rceil)}.w$; 
                \ENDIF             
                \FORALL {Sample $x$ in $X$}
                      \IF {$x.w>\gamma_{t+1}$}
                        \FORALL {$v_{j} \in x$}
                           \STATE  $p_{i,t+1,j}=p_{i,t+1,j}+1$;
                        \ENDFOR
                      \ENDIF
                \ENDFOR
                \FOR{$j=1$ to $n$}
                   \STATE  $p_{i,t+1,j}=p_{i,t+1,j}/\left\lceil \rho A_{i}\right\rceil$;
                   \STATE  $p_{i,j,t+1}=w p_{i,t+1,j}+(1-w)p_{i,t,j}$;
                \ENDFOR
\ENDFOR

    \STATE Output $F$;
\end{algorithmic}
\end{algorithm}
\end{document}